\documentclass[11pt]{article}
\usepackage{jheppub}

\usepackage[vcentermath]{youngtab}
\usepackage{graphicx}
\usepackage{amscd}
\usepackage{epstopdf}
\usepackage{amsthm}
\usepackage{enumerate}

\newtheorem{theorem}{Theorem}[section]
\newtheorem{lemma}[theorem]{Lemma}
\newtheorem{proposition}[theorem]{Proposition}
\newtheorem{corollary}[theorem]{Corollary}

\newcommand{\cA}{\mathcal{A}}

\newcommand{\cC}{\mathcal{C}}

\newcommand{\cF}{\mathcal{F}}
\newcommand{\cG}{\mathcal{G}}

\newcommand{\cJ}{\mathcal{J}}

\newcommand{\cN}{\mathcal{N}}
\newcommand{\cO}{\mathcal{O}}

\newcommand{\cQ}{\mathcal{Q}}
\newcommand{\cR}{\mathcal{R}}

\newcommand{\cW}{\mathcal{W}}
\newcommand{\cX}{\mathcal{X}}

\newcommand{\cZ}{\mathcal{Z}}

\newcommand{\bC}{\mathbb{C}}

\newcommand{\bP}{\mathbb{P}}
\newcommand{\bQ}{\mathbb{Q}}
\newcommand{\bR}{\mathbb{R}}

\newcommand{\bZ}{\mathbb{Z}}
\newcommand{\fA}{\mathfrak{A}}
\newcommand{\fB}{\mathfrak{B}}

\newcommand{\fm}{\mathfrak{m}}

\newcommand{\sabs}[1]{\left| \phantom{\rule{0pt}{7pt}} #1 \right|}

\newcommand{\ie}{{i.e.}}
\newcommand{\eg}{{e.g.}}

\DeclareMathOperator{\sign}{sign}
\DeclareMathOperator{\Tr}{Tr}
\DeclareMathOperator{\re}{\mathbb{R}e}
\DeclareMathOperator{\im}{\mathbb{I}m}

\newcommand{\ex}{{\bf x}}
\newcommand{\why}{{\bf y}}
\newcommand{\eN}{{\bf N}}
\newcommand{\zee}{{\bf z}}
\newcommand{\Qyoo}{{\bf Q}}

\def\vev#1{\langle#1\rangle}

\def\eq#1{(\ref{#1})}
\def\p{\partial}

\def\Om{{\Omega}}

\def\a{{\alpha}}
\def\b{{\beta}}

\def \ra {\rightarrow}

\def\tir{{\tilde r}}
\def\Ga{{\Gamma}}

\def\bz{{\bar{z}}}

\def\s{{\sigma}}
\def\S{{\Sigma}}

\def\vF{{\vec F}}

\def\tS{{\widetilde{\Sigma}}}

\def\tir{{\tilde r}}

\def\ts{{\tilde \sigma}}
\def\tis{{\tilde s}}

\def\m{\mathfrak{m}}
\def\tm{\tilde{\mathfrak{m}}}

\def\Id{{I'}}
\def\Jd{{J'}}

\def\vFc{{\vec F^c}}

\newcommand{\be}{\begin{equation}}
\newcommand{\ee}{\end{equation}}
\newcommand{\bea}{\begin{equation} \begin{aligned}}
\newcommand{\eea}{\end{aligned} \end{equation}}
\newcommand{\nn}{\nonumber}

\newcommand{\bln}{\begin{align}}
\newcommand{\eln}{\end{align}}
\newcommand{\bst}{\begin{split}}
\newcommand{\est}{\end{split}}
\newcommand{\bi}{\begin{itemize}}
\newcommand{\ei}{\end{itemize}}
\newcommand{\ben}{\begin{enumerate}}
\newcommand{\een}{\end{enumerate}}

\title{Cluster algebras from \\
dualities of 2d $\boldsymbol{\cN{=}(2,2)}$ quiver gauge theories}

\author{Francesco Benini,}
\author{Daniel S. Park}
\author{and Peng Zhao}

\affiliation{Simons Center for Geometry and Physics \\
State University of New York \\
Stony Brook, NY 11794-3636, USA}

\emailAdd{fbenini}
\emailAdd{dpark}
\emailAdd{pzhao {\rm at} scgp.stonybrook.edu}

\abstract{We interpret certain Seiberg-like dualities of two-dimensional $\cN{=}(2,2)$ quiver gauge theories with unitary groups as cluster mutations in cluster algebras, originally formulated by Fomin and Zelevinsky. In particular, we show how the complexified Fayet-Iliopoulos parameters of the gauge group factors transform under those dualities and observe that they are in fact related to the dual cluster variables of cluster algebras. This implies that there is an underlying cluster algebra structure in the quantum K\"ahler moduli space of manifolds constructed from the corresponding K\"ahler quotients. We study the $S^2$ partition function of the gauge theories, showing that it is invariant under dualities/mutations, up to an overall normalization factor whose physical origin and consequences we spell out in detail. We also present similar dualities in $\cN{=}(2,2)^*$ quiver gauge theories, which are related to dualities of quantum integrable spin chains.}

\begin{document}

\setcounter{tocdepth}{2}

\maketitle

%
%

\section{Introduction}

Two-dimensional quantum field theory has been proven to be a fruitful subject in physics, playing a crucial role in key developments in statistical physics, condensed matter physics and, of course, string theory. The physics of two-dimensional theories---in particular, those with supersymmetry or superconformal symmetry---also is well-known to have far-reaching implications for topics in mathematics, such as quantum cohomology, mirror symmetry and integrable systems. In this paper, we add to the list of the many links between 2d physics and mathematics, a relation between certain two-dimensional gauge theories with $\cN=(2,2)$ supersymmetry and cluster algebras. The core of the connection relies on a set of infrared (IR) dualities, reminiscent of four-dimensional Seiberg duality \cite{Seiberg:1994pq}.

Cluster algebras are intricate discrete dynamical systems based on simple algebraic recurrences, that were originally formulated by mathematicians Fomin and Zelevinsky \cite{FominZ1} (see also \cite{FominZ2, BFominZ3, FominZ4}) to describe the coordinate rings of groups and Grassmannians.%
\footnote{It is possible to argue (see \eg{}, \cite{ZelevinskyWCM}) that certain elements of cluster algebras were already known to physicists in the form of Seiberg dualities of four-dimensional quiver gauge theories \cite{Berenstein:2002fi}.}
Since then, cluster algebra has developed into a rich subject, finding use in a wide range of topics in mathematics and physics. In mathematics, cluster algebras have appeared in the study of Teichm\"uller theory \cite{gekhtman2005, FockGoncharov, Fock:2003xxy}, tilting theory, pre-projective and Hall algebras, Donaldson-Thomas invariants and wall-crossing \cite{Kontsevich:2008, Kontsevich:2009}. They have also found applications in many topics in physics such as Zamolodchikov periodicity in the thermodynamic Bethe ansatz and integrable Y-systems \cite{FominZ:Y}, the identification of BPS spectra in gauge theories \cite{Gaiotto:2010be, Alim:2011ae, Alim:2011kw, Xie:2012gd, Cecotti:2014zga}, and more recently, the study of four-dimensional quiver gauge theories \cite{Xie:2012mr, Heckman:2012jh, Franco:2014nca}, 3d theories constructed from M5-branes \cite{Terashima:2013fg, Dimofte:2013iv}, line operators \cite{Xie:2013lca, Cordova:2013bza}, and scattering amplitudes \cite{ArkaniHamed:2012nw, Golden:2013xva}.%
\footnote{A more extensive list of applications of cluster algebras can be found in \cite{clusterportal}.}

Cluster algebra puts three discrete dynamical systems under one roof. The basic objects used in defining a cluster algebra are ``seeds'' $(B,\why, \ex)$ consisting of three sets of data. Each seed contains a skew-symmetric $n\times n$ matrix $B = b_{ij}$ with integer entries, which can be thought of as the adjacency matrix of an oriented quiver diagram consisting of $n$ nodes and arrows connecting them. To each node of this quiver diagram, a ``coefficient'' $y_i$ and a ``cluster variable'' $x_i$ are assigned, which in turn constitute the coefficient $n$-tuple $\why$ and the ``cluster'' $\ex$. Given a seed, one can mutate it by a set of rules. Mutations are involutions that can be applied to any given node: $\mu_k$, the mutation at node $k$, acts on a given seed to produce a new seed
\be
\label{mutations rough}
(B', \why', \ex') = \mu_k( B, \why, \ex) \;.
\ee
The precise transformation rules are spelled out in section \ref{ss:CA}. The mutations define a discrete dynamical system; by repeated applications of mutations in arbitrary sequences one generates a ``tree" of new seeds. One can choose to study different dynamical subsystems of the full cluster algebra structure, since it is possible to consistently restrict to matrix (or quiver) mutations $B' = \mu_k(B)$, coefficient dynamics $(B',\why') = \mu_k(B,\why)$, or cluster algebras with vanishing coefficients $(B',\ex') = \mu_k(B,\ex)$.

The cluster $\ex$ can be thought of as a particular set of coordinates on a manifold, and the mutations $\mu_k(B,\why): \ex \mapsto \ex'$ can be interpreted as birational coordinate transformations. The cluster algebra is then defined to be a commutative algebra generated by all possible coordinates related to each other by such birational relations. We present precise definitions in section \ref{ss:CA}.

In this paper, we show that the cluster algebra mutation rules are realized in two-dimensional $\cN=(2,2)$ quiver gauge theories in the guise of Seiberg-like dualities. To do so, we first study some aspects of 2d $\cN=(2,2)$ supersymmetric SQCD-like gauge theories with unitary groups, which are used as building blocks of the quiver theories that we eventually study. In particular, we study the Seiberg-like IR dualities among the SQCD-like theories, originally proposed in \cite{Benini:2012ui} building on the works \cite{Hanany:1997vm, Hori:2006dk, Hori:2011pd}.%
\footnote{\label{foo: 3d case}%
Such dualities have similar features with the three-dimensional dualities proposed in \cite{Benini:2011mf}, and some elements of cluster algebras have been noticed in the latter context in \cite{Closset:2012eq, Xie:2013lya}.}
The duality is between a $U(N)$ gauge theory with $N_f$ fundamentals and $N_a$ antifundamentals, and a $U(N')$ gauge theory where
$$
N' = \max(N_f,N_a)-N \;,
$$
with $N_a$ fundamentals, $N_f$ antifundamentals and $N_fN_a$ extra gauge singlets coupled by a cubic superpotential. Both theories are deformed by a Fayet-Iliopoulos (FI) term parametrized by $\xi$ and the theta angle of the unitary gauge group, $\theta$. These can be grouped into the complexified FI parameter
\be
t = 2\pi \xi + i\theta \;,
\ee
which, under the duality, simply transforms as $t \to -t$. The theory can also have twisted masses for the flavor symmetries and superpotential interactions. Our prime tool of study is the two-sphere partition function $Z_{S^2}$, \ie{}, the Euclidean path integral of the non-twisted theory on $S^2$. The partition function, defined and computed in \cite{Benini:2012ui,Doroud:2012xw}, allows us to check the duality and, more importantly, to determine the precise map of parameters.\footnote{More recently, the hemisphere partition function has been computed \cite{Sugishita:2013jca, Honda:2013uca, Hori:2013ika} which can also be used to test dualities.}
In fact, $Z_{S^2}$ detects some subtle contact terms which, although irrelevant for the SQCD-like theories, become crucial in determining the precise map of parameters under dualities of more complicated theories obtained by putting together such ``building blocks."

Indeed, in the second part of the paper we study \emph{quiver} gauge theories, in which the gauge group is a product of unitary groups
$U(N_1) \times \cdots \times U(N_n)$ and all the chiral fields transform in the bifundamental representation of a pair of factors. We also consider more general quiver theories in which some of the gauge groups are un-gauged, \ie{}, ``demoted" to flavor symmetry groups. The gauge and matter content of the theory can be represented by the very same quiver diagrams that appear in cluster algebras. We can perform a Seiberg-like duality on a single node, say $k$, of the quiver%
---we refer to this duality as \emph{cluster duality}. This duality precisely realizes the cluster algebra mutation $\mu_k$. The flavor symmetry groups correspond to ``frozen nodes" that do not mutate. The appearance of mutations and the cluster algebra structure in this setting should not come as a surprise: it has been known for some time \cite{Cachazo:2001sg, Berenstein:2002fi, Feng:2002kk, Herzog:2003zc, Derksen1} that four-dimensional Seiberg duality transforms the quiver diagram of a quiver gauge theory as a mutation,%
\footnote{This is true up to some subtleties related to the superpotential, which we discuss in the context of 2d theories in section \ref{sec: nondegenerateFTs}. Moreover, the quivers and the rank assignments admissible in four dimensions are limited by gauge anomalies, while there are no such constraints in two dimensions.}
and that the gauge group ranks transform as the tropical limit of the cluster variables.%
\footnote{Similar features have been observed in three dimensions. See footnote \ref{foo: 3d case}.}
The extent, however, to which the full cluster algebra structure appears in two dimensions is quite intriguing and, to our knowledge, previously unnoticed. While it remains true that the quiver and the gauge group ranks transform according to the mutation rules of cluster algebra, more elements turn out to be present in 2d theories.

In 2d $\cN=(2,2)$ gauge theories, the FI parameters are classically marginal and their quantum beta functions are one-loop exact. It turns out that the beta function coefficients $\beta_i$ can be identified with the cluster algebra coefficients $y_i$. More precisely, given that $y_i$ are elements of the tropical semifield $\bP$ with
\be
y_i = u^{\beta_i} \,,
\ee
the transformation of $y_i$ under cluster mutations precisely reproduce the transformation of the beta functions $\beta_i$ under cluster dualities.%
\footnote{Precise definitions and explanation of notations regarding tropical semifields are presented in section \ref{ss:CA}.}
Furthermore, $u$ can be interpreted as a ratio of renormalization scales given a certain ultraviolet (UV) construction of the quiver theories.

More impressively, the FI parameters can be related to the cluster coordinates $x_i$, whose behavior under mutations exhibits the most intricate structure among the data contained in a seed. More precisely, upon defining the \emph{K\"ahler coordinates}
\be
z_i \,\sim\, e^{-t_i}
\ee
up to a subtle sign explained in the main text, $z_i$ can be identified with the ``dual cluster variables"
\be
z_i =\prod\nolimits_j x_j^{b_{ji}} \;,
\ee
whose transformation rules follow from those of $x_i$. In fact, for quiver gauge theories that flow to conformal fixed points in the IR (and thus all beta functions $\beta_i$ vanish), the K\"ahler coordinates $z_i$ can be identified with the ``$\mathcal{X}$-coordinates" of Fock and Goncharov \cite{FockGoncharov}. We will see that instanton corrections arising from vortices play a crucial role in the duality transformation rules; in fact, the FI parameters do not transform as cluster algebra variables in other dimensions, for instance in quiver quantum mechanics.

Finally, cluster algebra mutations can be observed in the twisted chiral sector of the quiver gauge theory. The twisted chiral ring of the quiver theory is generated by the coefficients of the $Q$-polynomials $Q_i(x) = \det (x-\sigma_i)$,%
\footnote{Here $x$ is a formal variable, which is not to be confused with the cluster variables.}
where $\sigma_i$ are the adjoint complex scalar operators in the vector multiplets. We find that the ``dressed'' $Q$-polynomials $\cQ_i(x) \sim x_i \, Q_i(x)$ transform under cluster dualities as cluster variables.%
\footnote{We thank Davide Gaiotto for suggesting this possibility.}

Our simple observation potentially has many fascinating consequences, three of which we elaborate on. First, the existence of a cluster algebra structure in a discrete dynamical system automatically implies certain properties, including total positivity, the Laurent phenomenon and the existence of a natural Poisson structure \cite{FominZ1, gekhtman2003cluster}. We do not fully understand the implications that these properties have on the physics of the gauge theories, but nevertheless, they lie at our disposal.

Second, since two-dimensional $\cN=(2,2)$ quiver gauge theories can be used to engineer many interesting geometries, our results imply that there is a cluster algebra structure in the geometric moduli spaces of the manifolds they engineer. For instance, the A-twisted theories compute the (equivariant) quantum cohomology of K\"ahler manifolds obtained as holomorphic sub-manifolds of K\"ahler quotients \cite{Witten:1988xj, Witten:1991zz, Witten:1993yc}. The gauge theory FI parameters---or equivalently the K\"ahler coordinates $z_i$---are in fact coordinates on their K\"ahler moduli spaces, and cluster mutations can be thought of as coordinate transformations. Our observation then implies the existence of a cluster algebra structure in the quantum cohomology of such manifolds, some of which are compact and Calabi-Yau (CY).

Third, since a quantum field theory has many more observables than a mere quiver, such observables may provide additional data associated to the mathematical systems where cluster algebras appear. For instance, in the two-dimensional gauge theory we can explicitly compute the sphere partition function, which provides a non-trivial K\"ahler metric on the space of K\"ahler coordinates $z_i$ \cite{Jockers:2012dk}. Can partition functions of quiver gauge theories produce useful metrics for other spaces whose coordinates exhibit cluster algebra structures, \eg{}, Teichm\"uller space? We leave these questions to future work.

Lastly, we touch upon how cluster dualities are related to the dualities of $\cN=(2,2)^*$ gauge theories, which are $\cN=(4,4)$ gauge theories whose supersymmetry is softly broken to $\cN=(2,2)$ by twisted masses. A beautiful connection between $\cN=(2,2)^*$ theories and quantum integrable systems---the so-called Gauge/Bethe correspondence---has been discovered by Nekrasov and Shatashvili \cite{Nekrasov:2009uh, Nekrasov:2009ui}. While $\cN=(2,2)^*$ SQCD-like theories do not exhibit the Seiberg-like dualities we mainly study, they have ``Grassmannian dualities" which can be interpreted as ``particle-hole" dualities of the corresponding quantum spin chain. We verify the equality of the $S^2$ partition functions under these dualities, extend the duality to $\cN=(2,2)^*$ quiver gauge theories, and briefly discuss a decoupling limit in which the Grassmannian dualities can be related to cluster dualities. More details on this subject will appear elsewhere.

The paper is organized as follows. In section \ref{sec: dualities}, we study the Seiberg-like dualities of $U(N)$ SQCD-like theories. We compare the chiral and twisted chiral rings of the dual descriptions, as well as their two-sphere partition functions $Z_{S^2}$ disclosing some peculiar contact terms. In the process, we also study the behavior of the vortex partition functions under the duality. In section \ref{s:CA}, exploiting the Seiberg-like dualities as building blocks, we reveal the full cluster algebra structure within the dualities of quiver gauge theories. For completeness, we review cluster algebras in section \ref{ss:CA}, and summarize some subtle aspects related to the superpotential that lead to the concept of ``non-degenerate'' quiver theories in section \ref{sec: nondegenerateFTs}. In section \ref{s:applications}, we comment on the implications of our results on the Gromov-Witten theory of Calabi-Yau manifolds obtained from quivers. Section \ref{sec: 22star} focuses on the dualities of $\cN=(2,2)^*$ quiver theories: we present them, compare the sphere partition functions of dual descriptions and discuss their relation to cluster dualities. We draw our conclusions in section \ref{s:future}, while technical computations are collected in the appendices.

\medskip

\emph{Note: During the completion of this paper, we became aware of a related work \cite{Gomis} which also studies dualities of two-dimensional theories on a two-sphere.}

\section{Dualities of 2d SQCD-like theories}
\label{sec: dualities}

\begin{figure}[t]
\begin{center}
\includegraphics[width=.8\textwidth]{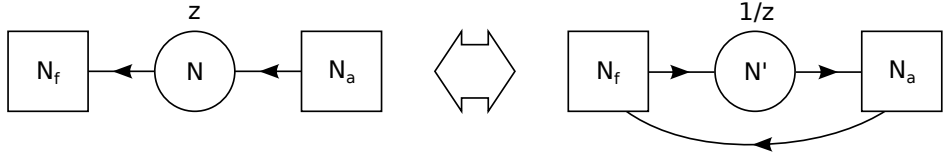}
\end{center}
\caption{Seiberg-like duality of SQCD-like theories in quiver notation. The circles represent unitary gauge groups, and the K\"ahler coordinates associated to the gauge groups are indicated above them. The squares represent unitary flavor groups, and arrows represent chiral multiplets transforming in the fundamental representation of the group at the tail and in the antifundamental of the group at the head. Since all matter fields are bifundamental, the flavor symmetry is actually given by $S\big[ U(N_f) \times U(N_a) \big]$. On the left is theory $\fA$: a $U(N)$ theory with $N_f$ fundamentals and $N_a$ antifundamentals. On the right is theory $\fB$: a $U(N')$ theory, where $N' = \max(N_f,N_a) - N$, with $N_a$ fundamentals, $N_f$ antifundamentals and $N_fN_a$ extra gauge singlets, as well as a superpotential $W = q'M\tilde q'$. The two theories are IR dual.
\label{fig: Seiberg duality}}
\end{figure}

Let us begin by studying certain Seiberg-like dualities of two-dimensional $\cN=(2,2)$ supersymmetric SQCD-like theories, \ie{}, of $U(N)$ gauge theories with $N_f$ chiral multiplets in the fundamental representation and $N_a$ in the antifundamental, as proposed in \cite{Benini:2012ui} on the wave of \cite{Hori:2006dk, Hori:2011pd, Jockers:2012zr}. Details on the exact contents and Lagrangians of such theories can be found in \cite{Witten:1993yc}, although we adhere to the conventions of \cite{Benini:2012ui}. Note that in two dimensions, contrary to the four-dimensional case, $N_f$ can be taken to differ from $N_a$, as this does not induce a gauge anomaly. In quiver notation, the matter content is represented in figure \ref{fig: Seiberg duality} on the left: the circle represents the $U(N)$ gauge group, while the squares represent the unitary flavor groups. The arrows represent bifundamental chiral multiplets (fundamental of the group at the tail, antifundamental of the group at the head). In the case at hand, the flavor group is actually $S\big[ U(N_f) \times U(N_a)\big]$ as the remaining $U(1)$ is gauged. We consider deforming these theories by twisted masses (that generically break the flavor group to $U(1)^{N_f + N_a - 1}$), a superpotential and a complexified Fayet-Iliopoulos parameter, whose imaginary part parametrizes the topological theta-angle term. As we see in section \ref{s:CA}, understanding these dualities is instrumental to understanding those of general quiver theories: the cluster dualities are local operations with respect to a node of the quiver, and most of the relevant structure is already encoded in the dualities of theories with a single gauge group.

We denote the chiral fields of the theory under consideration, which we call ``theory $\fA$," as $q_F$ and $\tilde q_A$, where the indices lie in the domains $F\in [N_f]$ and $A\in [N_a]$ respectively. We use the shorthand notation
\be
[n] = \{ 1, 2, \ldots, n \}
\ee
to condense our expressions. The gauge indices are left implicit unless stated otherwise. The twisted masses for the flavor symmetry are denoted by $s_F$ and $\tis_A$. We recall that, in a general $\cN=(2,2)$ theory, every time there is a flavor symmetry whose conserved current sits in a linear multiplet,%
\footnote{The conserved current of a flavor symmetry could sit in a twisted linear multiplet or in more general multiplets, and the gauge multiplet used to gauge it changes accordingly \cite{Lindstrom:2007vc, Lindstrom:2008hx}. In theories where all matter multiplets are chiral, the flavor currents sit in linear multiplets.}
we can introduce twisted masses. That is, we can couple the flavor symmetry to an external vector multiplet, and give a vacuum expectation value (VEV) to the complex adjoint scalar in that multiplet. Thus, up to flavor rotations, twisted masses take values in the Cartan subalgebra of the flavor symmetry. Denoting the VEV of the adjoint scalar by $s$, the actual mass of a chiral multiplet transforming as a weight $\rho$ under the flavor symmetry is $m = \rho(s)$. Here, $s_F, \tilde s_A$ are in fact the diagonal VEVs of the external vector multiplet scalars, not the actual masses. In the present case, although the flavor symmetry has maximal torus $U(1)^{N_f + N_a - 1}$, it turns out to be convenient to use $N_f + N_a$ parameters with an equivalence relation $\{s_F, \tilde s_A \} \simeq \{s_F + s, \tilde s_A + s\}$ for any $s\in\bC$. Such redundancy can be fixed if desired.

When $\max(N_f,N_a)<N$, the theory has no supersymmetric vacua, irrespective of the choice of the superpotential. This is explained in more detail in section \ref{sec: twisted chiral ring}. As a result of supersymmetry breaking, such theories turn out to have a vanishing sphere partition function \cite{Benini:2012ui, Doroud:2012xw}.

When $\max(N_f,N_a) \geq N$, the theory is IR dual to a $U(N')$ theory with $N_f'$ fundamentals $q'_{F'}$, $N_a'$ antifundamentals $\tilde q'_{A'}$ where
\be
\label{map ranks}
N'=\max(N_f,N_a)-N \;,\qquad\qquad N_f' = N_a \;, \qquad\qquad N_a' = N_f \;,
\ee
and gauge singlets $M_{F'A'}$ transforming in the bifundamental representation of the flavor symmetry group. We use primed variables to refer to quantities of the dual theory. The quiver diagram of this theory is depicted in figure \ref{fig: Seiberg duality} on the right. The theory also has a superpotential
\be
\label{suppot}
W_\text{dual} = \sum_{F',A'} \Tr \big( q'_{F'} M_{F'A'} \tilde q'_{A'} \big) \;,
\ee
where the trace is taken to be over gauge indices. We denote this theory as ``theory $\fB$."

We can further describe the map of parameters under the duality. Since the twisted masses are related to flavor symmetries, their map is trivial:
\be
\label{map masses}
s'_{F'} = \tilde s_{F'} \;,\qquad\qquad \tilde s'_{A'} = s_{A'} \;.
\ee
We define the complexified FI parameter of theory $\fA$
\be
t = 2\pi\xi + i\theta \;,
\ee
where $\xi$ is the real FI parameter and $\theta \simeq \theta+2\pi$ is the topological coupling; we denote the FI parameter of theory $\fB$ by $t'$. It is convenient to work with the \emph{K\"ahler coordinates}
\be
\label{def z SQCD}
z = (-1)^\text{\# outgoing arrows - \# colors} \; e^{-t} = (-1)^{N_f - N} \, e^{-t} \;,
\ee
and $z' = (-1)^{N_a - N'} e^{-t'}$. This is because the extra minus signs, which can be thought of as a shift of the theta angles, make the duality look simpler. In these coordinates, the map is given by
\be
\label{map FI}
z' = z^{-1} \;.
\ee

This duality has a nice geometric origin \cite{Witten:1993xi, Donagi:2007hi, Jia:2014ffa}. To illustrate, let us first consider the case when $N_a = 0$. When the FI parameter $\xi$ is large and positive,%
\footnote{If we take $\xi$ large negative in theory $\fA$, the Higgs branch is empty but there are $\binom{N_f}{N}$ quantum vacua on the Coulomb branch \cite{Witten:1993yc}. The geometric analysis presented here is not valid in this case. Similarly, the geometric analysis at large and positive $\xi$ is valid only as long as $N_a \leq N_f$ for non-zero $N_a$.}
the Higgs branch of theory $\fA$ is given by the Grassmannian $\text{Gr}(N,N_f)$: the space of complex $N$-planes in $\bC^{N_f}$. In the IR, this theory flows to a nonlinear sigma model (NLSM) with that target space. In the dual theory $\fB$, $\xi'$ becomes large in the negative direction and the theory flows to the NLSM of the Grassmannian $\text{Gr}(N_f-N,N_f)$. The two NLSMs are equivalent as
\be
\text{Gr}(N,N_f) = \text{Gr}(N_f-N,\, N_f) \;.
\ee
Hence, when $N_a =0$, the duality can be understood as the canonical isomorphism of Grassmannians. Let us now consider the case when $N_a \leq N_f$. In theory $\fA$, each antifundamental realizes a copy of the ``tautological bundle" $S$,%
\footnote{The fundamentals define $N$ vectors $v_{i=1,\dots,N}$ in $\bC^{N_f}$: upon quotienting by the gauge group $U(N)$, the total space reduces to the space of complex $N$-planes. Then, each antifundamental defines the coordinates $w^i$ of a vector $\vec w = w^i v_i$ lying on the $N$-plane, thereby realizing the tautological bundle $S$. $S$ is also referred to as the ``universal subbundle."}
and the theory flows in the IR to the NLSM of the total space $S^{\oplus N_a} \to \text{Gr}(N,N_f)$. Under the equality of Grassmannians, we also have the equality of bundles
\be
S \to \text{Gr}(N,N_f) \qquad=\qquad Q^* \to \text{Gr}(N_f - N, N_f) \;,
\ee
where $Q^*$ is the dual of the universal quotient bundle $Q$ defined by the short exact sequence
\be
0 \to S \to \cO^{N_f} \to Q \to 0 \;.
\ee
Here, $\cO$ is the trivial bundle. In the gauge theory, $(\cO^{N_f})^{N_a}$ is realized by the gauge singlets $M_{F'A'}$, while $N_a$ copies of the short exact sequence are encoded in the F-term equations imposed by the superpotential $W_\text{dual}$.

For the rest of this section, we present various nontrivial checks of the proposed dualities. Classically, both theories have a vector-like and an axial R-symmetry, $U(1)_V \times U(1)_A$. If $N_f \neq N_a$, the axial R-symmetry is anomalous and broken to $\bZ_{2|N_f - N_a|}$. In sections \ref{sec: chiral ring} and \ref{sec: twisted chiral ring}, we compare the chiral and twisted chiral rings  \cite{Lerche:1989uy} of the dual theories. We in fact allow for a generic superpotential in theory $\fA$ that preserves the $U(1)_V$ symmetry, and we show in section \ref{sec: chiral ring} what the superpotential of the dual theory $\fB$ becomes in such circumstances. The complexified FI term itself can be understood as a twisted superpotential $\widetilde \cW_\text{FI}$, linear in the twisted chiral multiplet $\Sigma$ constructed out of the vector multiplet.%
\footnote{One could in principle consider a theory with a generic twisted superpotential that breaks $U(1)_A$ classically, but we do not do so in this paper.}

In section \ref{sec: partition function}, we study the Euclidean path integrals of the two theories on the two-sphere \cite{Benini:2012ui, Doroud:2012xw} and compare them. This enables us to determine the exact map of parameters between the two theories. The analysis reveals the presence of some subtle contact terms in theory $\fB$, which include a twisted superpotential function of the twisted mass parameters. These contact terms do not affect correlators, and can be removed by local counterterms. Therefore, they can be ignored in the SQCD theories we investigate in this section. Once, however, the flavor nodes are promoted to gauge nodes---for example, by embedding the SQCD-like theories into quivers---the contact terms become dynamical twisted superpotential terms and have physical consequences. In fact, they are responsible for the cluster duality transformations analyzed in section \ref{s:CA}. Finally, the elliptic genera of theories $\fA$ and $\fB$ have been computed and shown to match in \cite{Gadde:2013dda, Benini:2013nda, Benini:2013xpa}.

\subsection{The chiral ring}
\label{sec: chiral ring}

We proceed with checking that the dual theories $\fA$ and $\fB$ have the same chiral ring. Let us first consider the case that theory $\fA$ does not have a superpotential. In this case, all chiral gauge-invariant operators can be expressed as (linear combinations of) products of the mesonic operators $\tilde q^A q_F$. One could have tried to construct independent baryonic operators, but since the gauge group is $U(N)$ rather than $SU(N)$, the basic gauge invariants are
\be
\big( \epsilon^{a_1 \ldots a_N} \tilde q^{A_1}_{a_1} \cdots \tilde q^{A_N}_{a_N} \big)\big( \epsilon_{b_1\ldots b_N} q^{b_1}_{F_1} \cdots q^{b_N}_{F_N} \big) \;\simeq\; \tilde q^{[A_1} q_{[F_1} \; \cdots\; \tilde q^{A_N]} q_{F_N]}
\ee
up to a numerical coefficient, where we have made the gauge indices $a_i, b_i$ explicit. These dibaryons are proportional to products of mesons and are not independent, hence the mesonic operators generate the chiral ring.

In theory $\fB$, the generator of the chiral ring are the mesonic operators $\tilde q'_{A'} q'_{F'}$, as well as the gauge singlets $M_{F'A' }$. The superpotential \eq{suppot}, however, imposes that the operators $\tilde q'_{A'} q'_{F'}$ must be trivial in the chiral ring. Therefore, the gauge singlets $M_{F'A' }$ form a complete set of generators for the chiral ring of theory $\fB$. The map
\be
\label{chiral ring map}
\tilde q_A q_F = M_{AF}
\ee
identifies the generators in the dual theories.%
\footnote{Notice that in four-dimensional $SU(N)$ SQCD the chiral ring can contain quantum F-term relations which cannot be inferred from the constituent fields and the superpotential, but are instead described by Seiberg duality \cite{Seiberg:1994pq}. In our case, the chiral ring is essentially classical, since the theory is deformed by the FI parameter.}

Now consider adding a generic superpotential $W = W_0(\tilde q_A q_F)$ to theory $\fA$. This leads to the following gauge-invariant F-term relations in the chiral ring:
\be
\label{general F-terms}
0 = \sum_A \tilde q_A q_G \, \partial_{AF} W_0(\tilde q_A q_F) \;,\qquad\qquad 0 = \sum_F \tilde q_B q_F \, \partial_{AF} W_0 (\tilde q_A q_F) \;.
\ee
Theory $\fB$ then has the superpotential
\be
\label{suppot complete}
W' = W_0(M_{AF}) + W_\text{dual} = W_0 (M_{AF}) + \sum_{AF} \Tr( \tilde q'_A M_{AF} \tilde q'_F) \;.
\ee
The F-term relations tell us that $\tilde q'_F q'_A = - \partial_{AF} W_0(M_{AF})$, \ie{},  the mesonic operator of the dual theory is not independent of $M_{AF}$. Contracting this relation with $M_{AG}$ or $M_{BF}$ and using the F-term relations $0 = \sum_A q'_A M_{AG} = \sum_F M_{BF} \tilde q'_F$, we exactly reproduce the relations (\ref{general F-terms}) under the map (\ref{chiral ring map}).

\subsection{The twisted chiral ring}
\label{sec: twisted chiral ring}

Let us now turn to the twisted chiral rings of the two theories. The twisted chiral ring of theory $\fA$ is generated by the gauge-invariant operators
\be
\Tr \s^k \qquad\qquad k= 1,\dots, N
\ee
where the complex scalar $\s$ is the lowest component of the adjoint twisted chiral multiplet constructed out of the vector multiplet. An alternative basis of $N$ generators can be obtained by diagonalizing the operator $\s = \text{diag} (\s_1, \ldots, \s_N)$, and then forming symmetric polynomials in the $\s_I$'s. A gauge invariant way to do so is to construct the ``$Q$-polynomial''
\be
Q(x) = \det (x-\s) = x^N - x^{N-1} \Tr\sigma + \ldots + \det(-\sigma) \;,
\ee
where $x$ is a formal variable. This degree $N$ monic polynomial is the generating function for the elementary symmetric polynomials in the $\s_I$'s that generate the ring. The ring is then defined by the relations between these operators. A convenient way to obtain the relations is to go on the Coulomb branch \cite{Witten:1993xi} by giving generic vacuum expectation values to the $\sigma_I$'s, so that the gauge group is broken to $U(1)^N$. All chiral multiplets and off-diagonal vector multiplets are massive and can be integrated out, leaving an effective theory for the diagonal vector multiplets. The computation is valid as long as all $\sigma_I$'s are well separated and far from the origin, which can be achieved by turning on generic twisted masses $s_F$, $\tilde s_A$. It turns out that there are no vacua where some of the $\sigma_I$'s coincide. The generic twisted masses, furthermore, remove the Higgs branch of the theory so that all supersymmetric vacua are on the Coulomb branch. The effective twisted superpotential includes a linear term, which is the bare FI term, and one-loop corrections coming from integrating out the massive fields:%
\footnote{Our conventions for the effective twisted superpotential, which agree with the sphere partition function computations, are such that integrating out fields $\Phi_j$ of mass $m_j$ in representation $R_j$ (including off-diagonal vector multiplets) contributes
$$
\widetilde \cW_\text{eff} = -\hat t \sum_I \sigma_I - \sum_j \sum_{\rho\in R_j} \big( \rho(\sigma) + m_j \big) \Big[ \log \Big( -i \big( \rho(\sigma) + m_j \big) \Big) - 1 \Big] \;,
$$
to the twisted superpotential. Here, $\rho$ are the weights of the representation $R_j$. We stress that $m_j$ is the actual mass, related to the twisted mass parameter $s$ by the flavor charge.}
\begin{multline}
\widetilde \cW_\text{eff} = - \sum_{I=1}^N \bigg\{ \Big( \hat t + i\pi(N+1-2I) \Big)\sigma_I + \sum_{F=1}^{N_f} (\sigma_I - s_F) \big[ \log(\sigma_I - s_F) - \tfrac{i\pi}2 - 1 \big] \\
- \sum_{A=1}^{N_a} (\sigma_I - \tilde s_A) \big[ \log(\sigma_I - \tilde s_A) + \tfrac{i\pi}2 - 1 \big] \bigg\} \;.
\end{multline}
Here, $t = 2\pi \xi +  i \theta$ is the complexified FI parameter, while $\hat t = t + 2 \pi i n$, where $n\in\bZ$ has to be chosen to minimize the potential energy $\big| \partial \widetilde W/\partial\sigma\big|$ \cite{Witten:1993yc}.  The supersymmetric vacua are solutions to the equations
\be
\frac{\partial \widetilde \cW_\text{eff}}{\partial \sigma_I} = 0 \;,
\ee
supplemented by the condition that the $\sigma_I$'s are all distinct for generic values of the parameters, and quotiented by the Weyl group. The equations can be exponentiated and written in the form
\be
\label{sigeq}
\prod_{F=1}^{N_f} (\s_I-s_F) + i^{N_a-N_f} \, z \prod_{A=1}^{N_a} (\s_I-\tis_A) =0 \qquad\qquad \forall I = 1,\dots,N \;,
\ee
where the K\"ahler coordinate $z=(-1)^{N_f-N} e^{-t}$ was defined in (\ref{def z SQCD}). These equations, along with the condition that the $\s_I$'s are distinct, determine the $\binom{\max(N_f,N_a)}{N}$ Coulomb vacua of the theory. These vacua turn out to determine the twisted chiral ring relations, as we see shortly. Note that when $\max(N_f,N_a) < N$, there cannot be $N$ distinct roots of \eq{sigeq}, hence supersymmetry is broken.

An efficient way to rewrite the vacuum equations is to express them as
\be
\label{TCR}
\prod_{F=1}^{N_f} (x-s_F) + i^{N_a-N_f} z \prod_{A=1}^{N_a} (x-\tis_A) = C(z) \, Q(x) \, T(x) \;.
\ee
This is a polynomial equation, where $x$ is treated as a formal variable, that should be solved for $Q(x)$ and for a degree $N'=\max(N_f,N_a)-N$ monic polynomial $T(x)$. Here,
\be
\label{function C}
C(z) = \begin{cases}
1 & \text{when $N_f > N_a$} \\
1+z & \text{when $N_f=N_a$} \\
i^{N_a-N_f}z & \text{when $N_f < N_a$} \,,
\end{cases}
\ee
which insures that both $Q(x)$ and $T(x)$ are monic. Equation \eq{TCR} implies that $Q(x)$ must pick $N$ distinct roots chosen from the $\max(N_f,N_a)$ roots of the left-hand side, while the remaining $N'$ roots are taken by $T(x)$; in particular, the eigenvalues of $\sigma$ must be distinct solutions to \eq{sigeq}. When expanding (\ref{TCR}) in powers of $x$, one gets $\max(N_f,N_a)$ equations. Starting from the highest order in $x$, the first $N'$ equations can be linearly rearranged to express the coefficients of $T(x)$ in terms of those of $Q(x)$, \ie{}, they fix $T(x)$ as a function of $Q(x)$. Substituting these relations into the remaining $N$ equations, one obtains $N$ relations among the elementary symmetric polynomials of $\s_I$'s. These are precisely the twisted chiral ring relations.

There is much more to the polynomial $T(x)$ than meets the eye. Upon expressing the twisted chiral ring relations of theory $\fB$ using the dual $Q$-polynomial $Q'(x) = \det(x-\s')$, we find
\be
\label{TCRprime}
\prod_{A=1}^{N_a}(x-\tilde s_A) + i^{N_f - N_a} z^{-1} \prod_{F=1}^{N_f} (x-s_F) = C'(z^{-1}) \, Q'(x) \, T'(x) \;,
\ee
where the function $C'(z^{-1}) = i^{N_f-N_a}z^{-1} C(z)$, as in (\ref{function C}), is merely chosen to match the leading coefficient in $x$. Now $Q'(x)$ is a degree $N'$ monic polynomial, and $T'(x)$ is a degree $N$ monic polynomial. The equations \eq{TCR} and \eq{TCRprime} are in fact equivalent, and can be merged into
\be
\prod_{F=1}^{N_f} (x-s_F) + i^{N_a-N_f} \, z \prod_{A=1}^{N_a} (x-\tis_A) = C(z) \, Q(x) \, Q'(x) \;. \label{operator map}
\ee
In other words, $T(x)$ in (\ref{TCR}) can be identified with the $Q$-polynomial $Q'(x)$ of the dual theory! Then the $N'$ equations that express $T(x)$ in terms of $Q(x)$ define the operator map between the twisted chiral rings of the two theories. Given a vacuum characterized by a solution $Q(x)$, $T(x)\equiv Q'(x)$ characterizes the same vacuum in the dual description. This has been explained for the case when $N_a=0$ in \cite{Witten:1993xi}.

\subsection{The $S^2$ partition function}
\label{sec: partition function}

Further information on the duality can be obtained by studying the two-sphere partition function $Z_{S^2}$ defined in \cite{Benini:2012ui, Doroud:2012xw}. Every Euclidean two-dimensional theory with $\cN=(2,2)$ supersymmetry and a conserved vector-like R-symmetry can be placed supersymmetrically on $S^2$ without twisting. This can be done by treating the stress tensor and the supercurrent as part of the $\cR$-multiplet \cite{Gates:1983nr, Komargodski:2010rb, Dumitrescu:2011iu}, and turning on an external scalar that couples to a scalar operator in that supermultiplet. This corresponds to adding certain scalar curvature couplings to the Lagrangian, which are controlled by the R-charges of the fields. If the theory has a Lagrangian description, one can define the path integral of the theory on $S^2$, which becomes a function of the parameters in the Lagrangian.

The sphere partition function of theories with vector and chiral multiplets was computed in \cite{Benini:2012ui, Doroud:2012xw} using localization techniques. In particular, it was shown that $Z_{S^2}$ can be written as a finite dimensional integral. It was further shown that for SQCD-like theories---such as our theories $\fA$ and $\fB$---the sphere partition functions can be rewritten as the sum of products of ``vortex partition functions.'' The latter formulation is convenient for exhibiting equalities of partition functions under the Seiberg-like duality being studied. In fact in \cite{Benini:2012ui}, the equality of the partition functions for the theories $\fA$ and $\fB$ was proven when $|N_f - N_a|> 1$. We reinterpret that equality, and examine the relations between the vortex partition functions when $|N_f - N_a|\leq 1$. We will soon see that this more refined study leads to interesting physical consequences.

To preserve supersymmetry on $S^2$, the imaginary parts of the twisted masses must be accompanied by external magnetic flavor fluxes $\fm_F$, $\tilde \fm_A$ proportional to them; it follows that those imaginary parts are quantized. It proves convenient to define the complex parameters
\bea
\label{def Sigma+-}
\S_{F+} &= is_F = i \re s_F + \frac{\m_F}2 \;,
\qquad\qquad&
\tS_{A+} &= i\tis_A = i\re \tis_A + \frac{\tm_A}2 \;, \\
\S_{F-} &= i\bar s_F =i \re s_F - \frac{\m_F}2 \;,
\qquad\qquad&
\tS_{A-} &= i\bar \tis_A = i\re \tis_A - \frac{\tm_A}2 \;,
\eea
using $\fm_F, \tilde \fm_A \in \bZ$ to denote the GNO quantized \cite{Goddard:1976qe} imaginary parts of the twisted masses.%
\footnote{Here and in the following, masses are expressed in units of the inverse sphere radius.}
The partition function is analytic in $\Sigma_{F\pm}$, $\tS_{A\pm}$, if we treat these parameters as independent. Moreover, we recall that the flavor symmetry is $S\big[ U(N_f) \times U(N_a)\big]$, since the diagonal $U(1)$ is gauged, and a shift of the masses along the diagonal $U(1)$ can be reabsorbed by a shift of $\Tr\sigma$ and the dynamical magnetic flux.

The partition function also depends on the complexified FI parameter $t$, that we express through $z$, and more generally on the full twisted superpotential $\widetilde \cW$. Finally, the sphere partition function depends on the choice of R-charges for the chiral multiplets, which we denote $r_F$ and $\tir_A$.%
\footnote{In general, to obtain the partition function of an IR fixed point we should use the superconformal R-charges at that point. If not unambiguously fixed by the superpotential, they can be found using $c$-extremization \cite{Benini:2012cz, Benini:2013cda}.}
We first present our calculation for the case of vanishing R-charges, and explain how to incorporate them afterwards.

The $S^2$ partition function \cite{Benini:2012ui, Doroud:2012xw} takes the form of a sum over diagonal quantized magnetic fluxes $\fm_I$ on the sphere, and an integral over the Cartan subalgebra $\sigma_I$ of the gauge group.%
\footnote{In the integral expression for the sphere partition function we use the integration variable $\sigma_I$, which should actually be identified with the real part of the complex field $\sigma_I$ of section \ref{sec: twisted chiral ring}. Hopefully, the context in which the symbol $\s_I$ is used makes clear what it indicates.}
It is convenient, in accordance with the notation in (\ref{def Sigma+-}), to define the combinations
\be
\label{def sigma+-}
\s_{I\pm} = i \s_I \pm \frac{\m_I}2 \;.
\ee
We also introduce the differences
\bea
\S^{I}_{J\pm} &= \s_{I\pm}-\s_{J\pm} \;, \qquad & \S^{I}_{A\pm} &= \s_{I\pm}-\tS_{A\pm} \;, \qquad & \S^{I}_{F\pm} &= \s_{I\pm}-\S_{F\pm} \;, \\
\S^{F_1}_{F_2\pm} &= \S_{F_1\pm}-\S_{F_2\pm} \;, \qquad & \S^{F}_{A\pm} &= \S_{F\pm}-\tS_{A\pm} \;, \qquad & \S^{A_1}_{A_2\pm} &= \tS_{A_1\pm}-\tS_{A_2\pm}  \;,
\eea
to condense our notation. Notice that we do not distinguish the various objects using tildes, but rather by their indices: we use $I,J,K$ for the gauge indices, $F$ to label fundamentals and $A$ to label antifundamentals. The $S^2$ partition function of a $U(N)$ gauge theory with $N_f$ fundamentals and $N_a$ antifundamentals can then be conveniently written as
\begin{multline}
\label{UNmother}
Z_{U(N)}^{N_f,N_a}\big( \S_{F\pm},\tS_{A\pm};z \big) = \frac1{N!} \sum_{\m_I \in \bZ^N} \int \prod_{I=1}^N \bigg[
\frac{d\s_I}{2\pi} \, \big( e^{i\pi(N_f-1)} z \big)^{\s_{I+}} \, \big( e^{-i\pi(N_f-1)} \bz \big)^{\s_{I-}} \bigg] \\
\times \prod_{I<J}^N \big( -\S^I_{J+} \S^I_{J-} \big) \;\cdot\; \prod_{I=1}^N \prod_{F=1}^{N_f} \frac{\Ga(-\S^I_{F+})}{\Ga (1+\S^I_{F-})} \prod_{A=1}^{N_a} \frac{\Ga(\S^I_{A+})}{\Ga (1-\S^I_{A-})} \;,
\end{multline}
where the holomorphic identification between the K\"ahler coordinate $z$ and the FI parameter $t=2\pi\xi + i\theta$ is on the sheet
\be
z = e^{i\pi(N_f-N)} \, e^{-t} \;.
\ee
The integrand in (\ref{UNmother}) contains an extra factor $(-1)^{(N-1)\sum\fm_I}$ with respect to \cite{Benini:2012ui}, which has been motivated in \cite{Hori:2013ika, Hori:2013ewa}.%
\footnote{In any case, the presence or absence of this extra sign factor does not modify the findings in this paper.}
The factors on the second line come from one-loop determinants of the fields around the Coulomb branch saddle point configurations: the first factor comes from the vector multiplet, while the others come from the chiral multiplets.

Charge conjugation acts on the parameters of the theory as:
\be
\label{cc}
N_f \,\leftrightarrow\, N_a \;, \qquad\quad
\S_F \,\leftrightarrow\, -\tS_A \;, \qquad\quad
t \,\leftrightarrow\, -t \;, \qquad\quad
z \,\leftrightarrow\, e^{i\pi(N_f-N_a)} z^{-1} \;.
\ee
It is easy to check that $Z_{S^2}$ is invariant under this action. For most of the rest of this section we assume that $N_f \geq N_a$ for the $U(N)$ theory $\fA$---if this is not the case, one can simply apply charge conjugation to both sides of the duality. When $N_f > N_a$, or $N_f = N_a$ and $\xi > 0$, the integral in \eq{UNmother} can be computed by closing the contours in the lower half-planes and picking up residues. Taking into account the sum over $\fm_I$, the poles are situated at
\be
\s_{I+} =\S_{F_I +} +n_{I+} \;, \qquad\qquad \s_{I-}=\S_{F_I -} +n_{I -}
\ee
and are parametrized by $F_I \in [N_f]$ and non-negative integers $n_{I\pm} \in \bZ_{\geq0}$. As explained in \cite{Benini:2012ui, Doroud:2012xw}, this integral can be written in the form
\be
\label{zs2aszvortsum}
Z_{U(N)}^{N_f,N_a} \big(\S_{F\pm}, \tS_{A\pm};z \big) = \sum_{\vF \in C(N,N_f)} Z_0^\vF \; Z^\vF_+ \; Z^\vF_- \;,
\ee
where $C(N,N_f)$ are ordered $N$-tuples of distinct integers in the range $[N_f]$, \ie{},
\be
1 \leq F_1 < F_2 < \ldots < F_N \leq N_f \;.
\ee
This expression has a natural interpretation as coming from Higgs branch localization of the path integral on the sphere. The contributing BPS configurations are characterized by $N$ out of the $N_f$ fundamental chiral multiplets being non-vanishing, and $\vec F$ labels such sectors. In each of them, the scalars $\sigma_I$ in the twisted chiral multiplets are fixed to equal and cancel the twisted masses $s_{F_I}$, so that the $N$ fundamentals can be non-vanishing. In each sector, there is a ``Higgs'' configuration where the fundamentals are constant, as well as other configurations with vortices at one pole and antivortices at the other pole of the sphere.%
\footnote{See \cite{Benini:2013yva} for a similar analysis in three-dimensional theories.}
Then, $Z_0^{\vec F}$ is the classical and one-loop contribution to the partition function, while $Z_+^\vF$ and $Z_-^\vF$ are the vortex and antivortex contributions, respectively.

We denote the complement of $\vF$ with respect to $[N_f]$ as $\vF^c$, such that
\be
\{ F^c_\Id \} = [N_f] \setminus \{ F_I \} \;.
\ee
We use primed letters $\Id, \Jd, \ldots \in [N_f-N]$ to label the gauge indices of the $U(N')=U(N_f-N)$ dual theory. Clearly, $\vec F^c \in C(N',N_f)$. Then $Z^\vF_0$ is given by
\be
\label{ZF0}
Z^\vF_0 = \prod_{I=1}^N \big( e^{i\pi(N_f-N)} z \big)^{\S_{F_I +}} \big( e^{-i\pi(N_f-N)} \bz \big)^{\S_{F_I -}} \; \prod_{\Id=1}^{N'} \frac{\Ga\big( -\S^{F_I}_{F^c_\Id +} \big)}{\Ga\big( 1+\S^{F_I}_{F^c_\Id -}\big)} \; \prod_{A=1}^{N_a} \frac{\Ga\big(\S^{F_I}_{A+}\big)}{ \Ga\big(1-\S^{F_I}_{A -}\big)} \;,
\ee
while the vortex partition functions can be written as
\be
Z^\vF_+ = Z_\text{v}^\vF \big( \S_{F+},\tS_{A+};(-1)^{N_f-N}z \big) \;, \qquad\qquad Z^\vF_- = Z_\text{v}^\vF \big( \S_{F-}, \tS_{A-}; (-1)^{N_a-N}\bz \big) \;,
\ee
in terms of the function
\be
\label{vpfexp}
Z_\text{v}^\vF \big( \S_F, \tS_A; q \big) = \sum_{n \geq 0} q^n \sum_{\sabs{(n_I)} = n} \; \prod_{I=1}^N \frac{ \prod_{A=1}^{N_a} \big( \S^{F_I}_A \big)_{n_I} }{ \prod_{J=1}^N \big( {-\S^{F_I}_{F_J} -n_I} \big)_{n_J} \, \prod_{\Jd=1}^{N'} \big( {-\S^{F_I}_{F^c_\Jd} -n_I} \big)_{n_I} } \;.
\ee
Here, $(n_I)$ are $N$-tuples of non-negative integers, we have defined the norm $\big| (n_I) \big| = \sum_I n_I$, while $(a)_n = \Gamma(a+n)/\Gamma(a)$ is the Pochhammer symbol.

We claim that the vortex partition function $Z_\text{v}^\vF \big( \S_{F}, \tS_{A}; q\big)$ satisfies the following relation:
\begin{multline}
\label{zvortrel}
Z_\text{v}^\vF \big( \S_F, \tS_A;q \big) = Z_\text{v}^\vFc \Big( \tfrac12 -\S_{F},\,  -\tfrac12 -\tS_{A};\,  (-1)^{N_a}q \Big) \\
\times
\begin{cases}
1 &\text{if $N_f \geq N_a +2$} \\[.2em]
\exp\Big((-1)^{N_f-N+1}q\Big) &\text{if $N_f = N_a +1$} \\[.4em]
\Big( 1+(-1)^{N_f-N}q \Big)^{\sum_A \tS_A - \sum_F \S_F + N_f-N} &\text{if $N_f = N_a$} \;.
\end{cases}
\end{multline}
Since $\vFc \in C(N',N_f)$, $Z_\text{v}^{\vec F^c}$ is the vortex partition function of a $U(N')$ theory with $N_f$ fundamentals and $N_a$ antifundamentals. The identity for the case $N_f \geq N_a +2$ was proven in \cite{Benini:2012ui}, using an integral representation \cite{Dimofte:2010tz} for the coefficients of the vortex partition function. For each $n\geq 0$, the coefficient of $q^n$ in \eq{vpfexp} can be expressed as a contour integral:
\begin{multline}
\sum_{\sabs{(n_I)}=n} \, \prod_I \frac{\prod_A \big( \S^{F_I}_A \big)_{n_I} }{ \prod_J \big( {-\S^{F_I}_{F_J} -n_I} \big)_{n_J} \, \prod_\Jd \big( {-\S^{F_I}_{F^c_\Jd} -n_I} \big)_{n_I} } \\
= \frac{(-1)^n}{n!} \int_{\cC} \prod_{\a=1}^n \frac{d \varphi_\a}{2\pi i} \, \frac{\prod_A (\varphi_\a -\tS_A) }{ \prod_I (\varphi_\a - \S_{F_I}) \, \prod_\Id (\Sigma_{F^c_{I'}} - \varphi_\a - 1)} \;\cdot\; \prod_{\a < \b}^n \frac{(\varphi_\a -\varphi_\b)^2 }{(\varphi_\a -\varphi_\b)^2 -1} \;.
\end{multline}
The contour $\cC$ is defined to encircle the codimension-$n$ poles such that for a partition $(n_I)$ of $n$ into $N$ non-negative parts,
\be
\{  \varphi_\a \}  = \bigcup_{I=1}^N \, \Big\{  \S_{F_I},\; \S_{F_I}+1,\ldots
,\; \S_{F_I}+ n_I-1 \Big\} \;.
\ee
A practical way of describing this contour is the following. Assume that all $\S_F$ satisfy
\be
0 < \re \S_F < 1 \;.
\ee
The contour can then be taken to be the product $\cC = \prod_{\a=1}^n \cC_\a$, where $\cC_\a$ is a contour in the $\varphi_\alpha$-plane winding counterclockwise along the boundary of the half-disk of infinite radius that lies inside the half-plane $\re \varphi_\a \geq 0$, and with diameter along the imaginary axis. When the integrand has no poles at infinity, each $\cC_\a$ can be ``flipped over" to become the boundary of the infinite half-disk lying inside $\re \varphi_\a \leq 0$. Then the integral becomes precisely the coefficient of the vortex partition function $Z_\text{v}^\vFc \big( \frac12 -\S_F,\, -\frac12 - \tilde\S_A;\, (-1)^{N_a}q \big)$. For $N_f \geq N_a +2$ indeed there are no poles at infinity, therefore
\bea
\label{NfNa2}
& \hspace{-1em} \sum_{\sabs{(n_I)}=n} \; \prod_I \frac{\prod_A \big(\S^{F_I}_A \big)_{n_I}}{\prod_J \big( {-\S^{F_I}_{F_J} -n_I} \big)_{n_J} \, \prod_\Jd \big( {-\S^{F_I}_{F^c_\Jd} - n_I} \big)_{n_I} } \qquad\qquad\qquad \text{for $N_f \geq N_a +2$} \\
&\qquad\quad =(-1)^{nN_a} \sum_{\sabs{(n_\Id)}=n} \; \prod_\Id \frac{\prod_A \big(1 - \S^{F^c_\Id}_A \big)_{n_\Id} }{ \prod_\Jd \big( \S^{F^c_\Id}_{F^c_\Jd} - n_\Id \big)_{n_\Jd} \, \prod_J \big( \S^{F^c_\Id}_{F_J} - n_\Id \big)_{n_\Id} } \;.
\eea
Recall that we use the indices $I, J \in [N]$ as gauge indices of theory $\fA$, while the primed indices $\Id, \Jd \in [N']$ are gauge indices of theory $\fB$. In particular, $(n_\Id)$ on the right-hand side are $N'$-tuples of non-negative integers.

When $|N_f -N_a| \leq 1$, however, there are ``poles at infinity," and the relation between the coefficients of the two vortex partition functions becomes more complicated. One way to determine it is to properly account for the behavior of the integrand at infinity. Another method is to start with the equalities for $N_f = N_a + 2$, and ``decouple'' one or two fundamental fields by taking their twisted masses to infinity. We take the latter approach to prove the equalities \eq{zvortrel} in appendices \ref{ap:NfNa10} and \ref{ap:forNfNa0}.

The classical and one-loop piece (\ref{ZF0}) of the partition function satisfies the identity
\begin{multline}
\label{z0rel}
Z^\vF_0 \big( \S_F,\tS_A;z \big)= |z|^{-N'} \, \prod_F \big( e^{i\pi N'}z \big)^{\S_{F+}} \big( e^{-i\pi N'}\bz \big)^{\S_{F-}} \,\cdot\, \prod_A (-1)^{N'( \tS_{A+} - \tS_{A-})} \\
\times \prod_{F,A} \frac{\Ga\big(\S^F_{A+}\big)}{\Ga\big(1-\S^F_{A-}\big)} \;\cdot\; Z^\vFc_0 \Big( \tfrac12 -\S_F,\; - \tfrac12 - \tS_A;\; e^{i\pi (N_f-N_a)}z \Big) \;.
\end{multline}
Here, $Z_0^{\vec F^c}$ represents the classical and one-loop piece for a $U(N')$ theory with $N_f$ fundamentals and $N_a$ antifundamentals, while the first factor on the second line is the one-loop determinant of the extra gauge singlets. The factors relating the two $Z_0$ functions are independent of the vortex sector $\vF$.

The relations (\ref{zvortrel}) and (\ref{z0rel}) imply a relation between the $S^2$ partition functions of a $U(N)$ theory with $N_f,N_a$ flavors and a $U(N')$ theory with $N_f,N_a$ flavors and $N_fN_a$ extra singlets. It is convenient to describe the latter theory in terms of charge conjugate fields, applying the map (\ref{cc}). We then find the following relation:
\begin{multline}
\label{UNUNprime}
Z_{U(N)}^{N_f,N_a} \big( \Sigma_{F\pm}, \tS_{A\pm}; z \big) = G \; |z|^{-N'} \, \prod_F \big( e^{i\pi N'}z \big)^{\Sigma_{F+}} \big( e^{-i\pi N'} \bar z\big)^{\Sigma_{F-}} \,\cdot\, \prod_A (-1)^{N'(\tS_{A+} - \tS_{A-})} \\
\times \prod_{F,A} \frac{\Gamma\big( \Sigma^F_{A+} \big)}{\Gamma\big( 1-\Sigma^F_{A-}\big)} \;\cdot\; Z_{U(N')}^{N_a,N_f} \Big( \tS_{A\pm} + \tfrac12,\, \Sigma_{F\pm} - \tfrac12 ;\, z^{-1} \Big) \;,
\end{multline}
where $G$ is the following function:
\be
\label{G function}
G = \begin{cases} 1 & \text{if $N_f \geq N_a + 2$} \\
e^{-2i\im z} & \text{if $N_f = N_a + 1$} \\
(1+z)^{\sum_A \tS_{A+} - \sum_F \Sigma_{F+} + N_f - N} (1+\bar z)^{\sum_A \tS_{A-} - \sum_F \Sigma_{F-} + N_f - N} &\text{if $N_f = N_a$}\;. \end{cases}
\ee
The second line of (\ref{UNUNprime}) is the partition function of a $U(N')$ theory with $N_a,N_f$ flavors as well as $N_fN_a$ extra singlets, which is the content of theory $\fB$. To understand the physical implications of the first line, it is convenient to split it into its absolute value and its phase. We can thus write  the relation in the compact form:
\be
\label{UNUNprime compact}
Z_{U(N)}^{N_f,N_a} \big( \Sigma_{F\pm}, \tS_{A\pm}; z \big) = f_\text{imp} \; f_\text{ctc} \; \prod_{F,A} \frac{\Gamma\big( \Sigma^F_{A+} \big)}{\Gamma\big( 1-\Sigma^F_{A-}\big)} \;\cdot\; Z_{U(N')}^{N_a,N_f} \Big( \tS_{A\pm} + \tfrac12,\, \Sigma_{F\pm} - \tfrac12 ;\, z^{-1} \Big) \;.
\ee
This equation, in fact, holds for the case $N_a > N_f$ as well, as can be checked upon charge conjugation. For general $N_f$ and $N_a$, the dual rank is $N' = \max(N_f,N_a)-N$ as we have already stated. We thereby reproduce the maps \eq{map ranks} of the duality between $\fA$ and $\fB$.

The function $f_\text{imp}$ is real positive:
\be
f_\text{imp} = |z|^{-N'} \times
\begin{cases}
|1+z|^{2(N_f-N)} & \text{if $N_f = N_a$} \\
1 &\text{otherwise} \;.
\end{cases}
\ee
This factor comes from a local counterterm---called the ``improvement Lagrangian'' in \cite{Closset:2014pda}---of theory $\fB$. This counterterm vanishes in the limit where the radius of the sphere is taken to infinity, and hence is not present on flat space. In general, the improvement Lagrangian is constructed out of a linear multiplet $\cJ$, and implements an improvement transformation of the $\cR$-multiplet of the theory such that the R-symmetry current is mixed with the conserved current in $\cJ$. A way of constructing $\cJ$ is to take $\cJ = \Omega + \overline\Omega$ for a twisted chiral multiplet $\Om$. In our case, $\Omega$ is a holomorphic function of $z$ (or $t$), which is a background multiplet. The resulting counterterm is then given precisely by $\log f_\text{imp}$---it is given by a function of the parameters of the theory with no dynamical fields involved. It is also independent of the twisted mass parameters of the theory, that can be promoted to dynamical fields when SQCD-like theories are embedded in quiver theories. For this reason, we can safely ignore $f_\text{imp}$ for the purposes of this paper.

When $N_f \neq N_a$, $f_\text{imp}$ can be absorbed by shifting the $\cR$-multiplet of the theory by a linear multiplet constructed from the gauge current. This means that the factor $f_\text{imp}$ can be absorbed by a shift of the R-charges of the fields of the theory by a multiple of the gauge charge. Meanwhile, when $N_f = N_a$, the irrelevance of the factor $f_\text{imp}$ has a geometric interpretation. It has been noticed in \cite{Jockers:2012dk} (see also \cite{Gomis:2012wy, Gerchkovitz:2014gta}) that the $S^2$ partition function produces the K\"ahler potential on the K\"ahler moduli space of the underlying Calabi-Yau manifold, whenever the gauge theory flows in the IR to a NLSM:
\be
Z_{S^2}(z,\bar z) = e^{-K_\text{K\"ahler}(z,\bar z)} \;.
\ee
The K\"ahler moduli $z,\bar z$ of the Calabi-Yau are controlled by the K\"ahler coordinates $z,\bar z$ in the gauge theory. Multiplication of $Z_{S^2}$ by the real positive function $f_\text{imp}(z)$ then implements a K\"ahler transformation of the K\"ahler potential, which does not affect the metric.

The function $f_\text{ctc}$ in \eq{UNUNprime compact} is a pure phase, given by
\bea
\label{f ctc}
&f_\text{ctc} = \\
&\begin{cases}
\prod\limits_F \big[ e^{i\pi N'}z \big]^{\Sigma_{F+}} \big[ e^{-i\pi N'} \bar z \big]^{\Sigma_{F-}} \prod\limits_A \big[ e^{i\pi N'}\big]^{\tS_{A+}} \big[ e^{-i\pi N'} \big]^{\tS_{A-}} & \text{$N_f \geq N_a + 2$} \\
\prod\limits_F \big[ e^{i\pi N'}z \big]^{\Sigma_{F+}} \big[ e^{-i\pi N'} \bar z \big]^{\Sigma_{F-}} \prod\limits_A \big[ e^{i\pi N'}\big]^{\tS_{A+}} \big[ e^{-i\pi N'} \big]^{\tS_{A-}} \; e^{-2i\im z} & \text{$N_f = N_a + 1$} \\
\prod\limits_F \Big[ \dfrac{e^{i\pi N'}z}{1+z} \Big]^{\Sigma_{F+}} \Big[ \dfrac{e^{-i\pi N'} \bar z}{1+\bar z} \Big]^{\Sigma_{F-}} \prod\limits_A \Big[ e^{i\pi N'} (1+z) \Big]^{\tS_{A+}} \Big[ e^{-i\pi N'} (1+\bar z) \Big]^{\tS_{A-}} & \text{$N_f = N_a$} \\
\prod\limits_F \big[ e^{i\pi N'} \big]^{\Sigma_{F+}} \big[ e^{-i\pi N'}  \big]^{\Sigma_{F-}} \prod\limits_A \big[ e^{i\pi (N_f-N)} z\big]^{\tS_{A+}} \big[ e^{-i\pi (N_f-N)} \bar z \big]^{\tS_{A-}}\; e^{2i\im z^{-1}} & \text{$N_f = N_a - 1$} \\
\prod\limits_F \big[ e^{i\pi N'} \big]^{\Sigma_{F+}} \big[ e^{-i\pi N'}  \big]^{\Sigma_{F-}} \prod\limits_A \big[ e^{i\pi (N_f -N)} z\big]^{\tS_{A+}} \big[ e^{-i\pi (N_f-N)} \bar z \big]^{\tS_{A-}} & \text{$N_f \leq N_a - 2$} \;.
\end{cases}
\eea
The cases $N_f \leq N_a$ can be obtained from the cases $N_f \geq N_a$ by charge conjugation. This factor also comes from a contact term of theory $\fB$, which does not vanish in the flat space limit: it, in fact, comes from a twisted superpotential that is a function of the background twisted chiral multiplets $z$ (or $t$), $s_F$ and $\tilde s_A$.%
\footnote{See \cite{Benini:2012ui, Doroud:2012xw} or the more general analysis in \cite{Closset:2014pda}.}
In the SQCD-like theory we can ignore the function $f_\text{ctc}$ as it cancels out when computing expectation values of operators.%
\footnote{Although $f_\text{ctc}$ does not affect the computation of correlators, it can be detected in the operator map between theory $\fA$ and $\fB$ if we define the operators through functional derivatives of the partition function. We see this in section \ref{sec: operator map}.}
Once, however, we gauge the flavor symmetries to construct more complicated theories---\eg{}, quiver gauge theories---the background fields $s_F$, $\tilde s_A$ become dynamical fields on the same footing as $\sigma_I$, and the full partition function becomes an integral/sum over $\Sigma_{F\pm}$, $\tS_{A\pm}$ on the same footing as $\sigma_{I\pm}$ in (\ref{UNmother}). Contact terms involving $s_F, \tilde s_A$ then become standard twisted superpotential terms. The theory would also have FI terms $t_f$, $t_a$ for the newly gauged symmetries, which appear in the partition function as terms
$$
\big( e^{i\pi(N_f-1)} \, e^{-t_f} \big)^{\Sigma_{F+}} \big( e^{-i\pi(N_f-1)} \, e^{-\bar t_f} \big)^{\Sigma_{F-}} \big( e^{i\pi(N_a-1)} \, e^{-t_a} \big)^{\tS_{A+}} \big( e^{-i\pi(N_a - 1)} \, e^{-\bar t_a} \big)^{\tS_{A-}} \;.
$$
Comparison with (\ref{f ctc}) reveals that $f_\text{ctc}$ can be physically interpreted as a nontrivial shift of the neighboring FI terms $t_f$, $t_a$ by a function of $t$ after the duality. This is the origin of the cluster algebra structure within quiver gauge theories, which is the subject of section \ref{s:CA}. By expanding at small $z$, \ie{}, in the large volume limit where $t \to\infty$, we see that in the case $N_f = N_a$ the shift involves instanton corrections.

We note that when $N_f = N_a \pm 1$, $f_\text{ctc}$ also gives rise to a twisted superpotential term linear in $z^{\pm1}$. This term does not involve dynamical fields, even after gauging the flavor symmetries. It does, however, play a role in the construction of the twisted chiral operator map, as we see in section \ref{sec: operator map}.

Let us finally explain the shift of the twisted masses inside $Z_{U(N')}^{N_a,N_f}$ in (\ref{UNUNprime compact}): it is due to the map of R-charges of chiral fields in the dual theories. The R-charges of chiral multiplets affect the Lagrangian on $S^2$; therefore, $Z_{S^2}$ depends on $r_F$, $\tilde r_A$ as well. The R-charges are incorporated in the integral localization formula (\ref{UNmother}) by making the replacements
\be
\label{replacements}
\S_{F \pm} \,\to\, \S_{F\pm} + \frac{r_F}2 \;, \qquad\qquad \tS_{A \pm} \,\to\, \tS_{A\pm} - \frac{\tir_A}2 \;.
\ee
Hence all equations from (\ref{UNmother}) to (\ref{G function}) are still valid with general R-charges. Making the R-charge dependence explicit, we arrive at
\begin{multline}
\label{UNUNprimewithR}
Z_{U(N)}^{N_f,N_a} \big( \Sigma_{F\pm}, \tS_{A\pm}, r_F, \tilde r_A; z \big) = \\
= f^{(r)}_\text{imp} \; f_\text{ctc} \; \prod_{F,A} \frac{\Gamma \big( \S^F_{A+} + \frac{r_F + \tir_A}2 \big) }{ \Gamma \big( 1-\S^F_{A-} - \frac{r_F + \tir_A}2 \big)} \;\cdot\; Z_{U(N')}^{N_a, N_f} \Big( \tS_{A\pm},\, \Sigma_{F\pm},\, 1- \tilde r_A,\, 1-r_F;\, z^{-1} \Big) \;.
\end{multline}
Here $f_\text{ctc}$ is exactly the same as in (\ref{f ctc}), while $f^{(r)}_\text{imp}$ is a deformation of $f_\text{imp}$ with dependence on the R-charges. Hence the map of parameters between theory $\fA$ and $\fB$ is still given as in (\ref{map ranks}), (\ref{map masses}), and (\ref{map FI}), while the map of R-charges is given by
\be
r'_{F'} = 1- \tilde r_{F'} \;,\qquad\qquad \tilde r'_{A'} = 1- r_{A'} \;.
\ee
From the gauge-singlet factor in (\ref{UNUNprimewithR}), we can deduce that the gauge singlet $M_{FA}$ has R-charge $r_F + \tilde r_A$, consistent with the chiral ring map (\ref{chiral ring map}). Moreover, the R-charges of the quarks in theory $\fB$ are such that the superpotential term $W_\text{dual} = \sum_{F',A'} \Tr \big( q'_{F'} M_{F'A'} \tilde q'_{A'} \big)$ in (\ref{suppot}) has R-charge 2, consistent with its presence in theory $\fB$.

\subsubsection{The operator map}
\label{sec: operator map}

We can check the presence of the contact terms implied by $f_\text{ctc}$ in (\ref{UNUNprime compact}) against the twisted chiral operator map between theory $\fA$ and $\fB$ encoded in (\ref{operator map}). Expanding equation (\ref{operator map}) and comparing the coefficient of the term $x^{\max(N_f,N_a) -1}$, we obtain the following linear relation between the operators with lower dimension:
\be
\label{trace map}
\Tr\sigma= - \Tr\sigma' + \begin{cases}
\sum_F s_F & \text{$N_f \geq N_a + 2$} \\
\sum_F s_F + i\mu\, z(\mu) & \text{$N_f = N_a + 1$} \\
\frac1{1+z}\sum_F s_F + \frac z{1+z} \sum_A \tilde s_A & \text{$N_f = N_a$} \\
\sum_A \tis_A + i\mu\, z(\mu)^{-1} & \text{$N_f = N_a - 1$} \\
\sum_A \tis_A & \text{$N_f \leq N_a - 2$} \;.
\end{cases}
\ee
For the cases $N_f = N_a\pm1$, we have reinstated the correct dimensions using the cutoff scale $\mu$, at which the K\"ahler coordinate $z(\mu)$ is defined.

We can alternatively extract the relation between $\Tr\sigma$ and $\Tr\sigma'$ from the equality of partition functions (\ref{operator map}). We can compute expectation values of twisted chiral operators inserted at the north pole of the sphere with localization by inserting their expectation value on the BPS configurations in the integral formula \eq{UNmother}. For $\Tr \s$ and $\Tr \s'$ this is implemented efficiently by taking a derivative:
\be
\langle \Tr\sigma \rangle = -i \, \frac{\partial \log Z^{N_f,N_a}_{U(N)}}{\partial \log z} \;, \qquad\qquad
\langle \Tr\sigma' \rangle = -i \, \frac{\partial \log Z^{N_a,N_f}_{U(N')}}{\partial \log z'} \;,
\ee
where the expectation values are on $S^2$. To compare with the flat space analysis, we should be careful to reinstate the correct dimensions using powers of the sphere radius $r$, and take the flat space limit $r \to \infty$. We should also recall that the complexified FI parameter runs according to the one-loop exact beta function $\beta = N_a - N_f$, therefore in the non-conformal cases it must be defined at some scale. The couplings at different scales are related by
\be
\frac{z(\mu_1)}{\mu_1^\b} = \frac{z(\mu_2)}{\mu_2^\b} \;.
\label{couplings and beta}
\ee
While the FI term in \eqref{trace map} is defined at some generic scale $\mu$, it is naturally defined at the scale $1/r$ on the sphere. Taking these subtleties into account, the equality (\ref{UNUNprime compact}) implies
\be
\label{trace vevs}
\vev{\Tr\sigma} = -\vev{\Tr\sigma'} + \begin{cases}
\sum_F s_F + \Big[ \frac{iN'}{2r} \Big] & \text{$N_f \geq N_a + 2$} \\
\sum_F s_F + \frac ir z\big( \frac1r \big) + \Big[ \frac{iN'}{2r} \Big] & \text{$N_f = N_a + 1$} \\
\frac1{1+z} \sum_F s_F + \frac z{1+z} \sum_A \tilde s_A  + \Big[ \frac{1-z}{1+z} \, \frac{iN'}{2r} \Big] & \text{$N_f = N_a$} \\
\sum_A \tis_A + \frac ir z\big(\frac1r\big)^{-1} + \Big[ \frac{iN'}{2r} \Big] & \text{$N_f = N_a - 1$} \\
\sum_A \tis_A + \Big[ \frac{iN'}{2r} \Big] & \text{$N_f \leq N_a - 2$} \;.
\end{cases}
\ee
The terms in brackets come from $f_\text{imp}$, and depend on how we mix the R-symmetry with the gauge symmetry; in any case, we see that they vanish in the $r\to\infty$ limit. All other correction terms come from $f_\text{ctc}$ and survive in the flat space limit, including the FI terms due to the relation \eq{couplings and beta}. We hence get a perfect matching between (\ref{trace map}) and (\ref{trace vevs}) in the flat space limit.

\section{Dualities of quiver gauge theories and cluster algebras}
\label{s:CA}

The dualities of $\cN=(2,2)$ SQCD-like theories can be directly applied to more complicated theories obtained by gauging (part of) the flavor symmetry. We can construct theories with more chiral and vector multiplets, in which the SQCD-like theory appears as a small ``block." Upon (partially) gauging the flavor symmetry, the twisted mass parameters $s_F$, $\tilde s_A$ get identified with the twisted chiral multiplets of some other gauge groups, which have their own FI interactions and so on.

An interesting class of theories is given by \emph{quiver gauge theories}, described in section \ref{ss:SDandCA}: they are gauge theories whose gauge/matter content is represented by a quiver diagram, which Seiberg-like dualities can act locally on. The main result of this section is the observation that these dualities precisely reproduce all elements of the \emph{cluster algebra} structure introduced by Fomin and Zelevinsky, which we briefly summarize in section \ref{ss:CA}. As we comment later on, this observation has potentially far-reaching consequences.

\subsection{Review of cluster algebras}
\label{ss:CA}

For completeness, we present a brief review of the basic definitions and facts about cluster algebras and mutations, following \cite{FominZ4}. More information can be found in the original papers \cite{FominZ1, FominZ2, BFominZ3, FominZ4} and other mathematical accounts \cite{FominZnotes, clusterportal}.

A cluster algebra $\cA$ of rank $n$ is a commutative ring with a unit and no zero divisors, defined as a subalgebra of an ambient field $\cF$ of rational functions in $n$ variables. The cluster algebra comes equipped with a distinguished set of generators called \emph{cluster variables}; this set is given by the (non-disjoint) union of a distinguished collection of overlapping algebraically independent $n$-subsets of $\cF$, called \emph{clusters}. The clusters are related to each other by special birational transformations: for each cluster $\mathbf{x}$ and every cluster variable $x\in\mathbf{x}$, there is another cluster $\mathbf{x}'$ obtained from $\mathbf{x}$ by replacing $x$ with a new cluster variable $x'$ determined by an \emph{exchange relation} of the form
$$
x\, x' = p^+ M^+ + p^- M^- \;,
$$
where $M^\pm$ are two monomials without common divisors in the $(n-1)$ variables $\mathbf{x} \setminus\{x\}$, while $p^\pm$ are coefficients valued in a \emph{semifield} $\bP$. Each exchange relation involves two different kinds of data: an \emph{exchange matrix} $B$ encoding the non-negative exponents in $M^\pm$, and the two \emph{coefficients} $p^\pm$. Furthermore, any two clusters can be obtained from each other by a sequence of exchanges of this kind.

In order to define a cluster algebra $\cA$, we first need to introduce its ground ring. The ground ring is constructed from a semifield $(\bP, \,\cdot\,, \oplus)$, \ie{}, an abelian multiplicative group endowed with an \emph{auxiliary addition} $\oplus$, which is commutative, associative and distributive, but not with a subtraction. The ground ring is taken to be the group ring $\bZ\bP$.%
\footnote{Given a ring $\bR$ and a group $G$, the group ring $\bR G$ consists of the formal sums
$$
\sum\nolimits_{y_\in G} c_y\, y \qquad\qquad\text{with } c_y \in \bR \;.
$$
If $\bR$ is a field, the group ring can be further thought of as a vector space over $\bR$ whose basis consists of the elements of $G$. A semifield $\bP$, as a multiplicative group, is automatically torsion-free (\cite{FominZ1} sec. 5), hence its group ring $\bZ\bP$ has no zero divisors.}

In this paper, we only consider tropical semifields. A \emph{tropical semifield} $\bP$ is an abelian group (with respect to multiplication) that is freely generated  by a finite set of labels (or formal variables) $u_1, \ldots, u_J$ of size $J$:
\be
\label{tropical semifield}
\bP = \Big\{ \prod\nolimits_{j=1}^J u_j^{a_j} \;\Big|\; a_j \in \bZ \Big\} \;.
\ee
Multiplication and the auxiliary addition $\oplus$ are defined to act on the elements as:
\be
\prod_j u_j^{a_j} \cdot \prod_j u_j^{b_j} = \prod_j u_j^{a_j + b_j} \;,\qquad\qquad \prod_j u_j^{a_j} \oplus \prod_j u_j^{b_j} = \prod_j u_j^{\max(a_j,b_j)} \;.
\ee
The group ring $\bZ\bP$ is then the ring of Laurent polynomials in the variables $u_j$. Notice that $1\oplus1=1$, so for $J=0$ the tropical semifield is trivial, \ie{},  it is a semifield with one element. In the next section, we see that the semifield with $J=1$ is relevant for studying the dualities of two-dimensional theories---the single label $u$ is identified with a ratio of renormalization scales in the field theory.

To define the core of the cluster algebra $\cA$, we need to define the exchange matrix and the coefficients for every possible exchange relation: in fact, they are determined in a very peculiar fashion. One starts by defining a \emph{seed} $(B, \why, \ex)$. The element $B = b_{ij}$ is an $n\times n$ skew-symmetric integer matrix.%
\footnote{This is not the most general type of cluster algebra. For example, the constraint that the matrix $b_{ij}$ should be skew-symmetric can be relaxed: one could have taken it to be skew-symmetrizable \cite{FominZ1}.}
The matrix $b_{ij}$ can be represented by a directed quiver diagram $B$ with $n$ nodes, where $b_{ij}$ counts, with sign, the number of arrows from node $i$ to node $j$ (\ie{}, in the notation of (\ref{[]+}), $[b_{ij}]_+$ is the number of arrows $i\to j$). Such a quiver diagram does not have any 1-cycles (an arrow that starts and ends at the same node) or oriented 2-cycles (a pair of arrows with opposite orientation connecting a given pair of nodes): we refer to quivers satisfying these conditions as \emph{cluster quivers} throughout this paper. The coefficient $n$-tuple $\why = (y_1,\dots,y_n)$ consists of coefficients $y_i \in \bP$, while the cluster $\ex=(x_1,\dots,x_n)$ is an $n$-tuple of cluster variables, each attached to the respective $i$-th node of the quiver.

Given a seed $(B,\why, \ex)$, for each node $k$ we define a \emph{mutation} $\mu_k$ which maps the seed to a new seed
\be
(B',\why',\ex') = \mu_k (B,\why,\ex) \;.
\ee
The mutation $\mu_k$ acts in the following way on the elements of the seed:
\bea
\label{cluster algebra transformations}
b_{ij}' &= \begin{cases} - b_{ij} & \text{if $i=k$ or $j=k$} \\ b_{ij} + \sign(b_{ik})\, [b_{ik} b_{kj}]_+ & \text{otherwise} \end{cases} \\[1em]
y_i' &= \begin{cases} y_k^{-1} & \text{if $i=k$} \\ y_i \, y_k^{[b_{ki}]_+} (y_k \oplus 1)^{-b_{ki}} & \text{otherwise} \end{cases} \\[1em]
x_i' &= \begin{cases} \dfrac1{x_k} \bigg( \dfrac{y_k}{y_k \oplus 1} \prod_j x_j^{[b_{jk}]_+} + \dfrac1{y_k \oplus 1} \prod_j x_j^{[-b_{jk}]_+} \bigg) & \text{if $i=k$} \\ x_i & \text{otherwise} \;. \end{cases}
\eea
Here we have defined
\be
\label{[]+}
[ x ]_+ = \max(x,0) \;, \qquad \qquad
\sign(x) = \begin{cases} 0 &\text{if $x=0$} \\ x/|x| &\text{otherwise} \;. \end{cases}
\ee
The mutation of the cluster $\ex$ in (\ref{cluster algebra transformations}) precisely implements an exchange relation $x_k \to x'_k$, while all other $(n-1)$ variables are left invariant. Therefore, the exchange matrix $B$ and the coefficient $n$-tuple $\why$ determine the $n$ possible exchange relations that involve $\ex$. On the other hand, $\mu_k$ also transforms $B$ and $\why$, and $(B',\why')$ in turn determine the $n$ possible exchange relations that involve $\ex'$, and so on. This gives rise to a convoluted dynamical system, and the full set of all possible exchange relations is determined by any one of the seeds. Notice that the mutations $\mu_k$ are involutions.

Finally, we define $\cX$ as the union of the clusters in all the seeds that can be generated by mutations from an ``initial seed" $(B_0,\why_0,\ex_0)$. Denoting the cluster variables of the initial seed as $\ex_0 = (x_1,\cdots,x_n)$, the ambient field $\cF$ of the cluster algebra is taken to be the field of rational functions of $(x_1,\cdots,x_n)$ with coefficients in $\bQ\bP$, \ie{},  $\cF = \bQ\bP(x_1, \ldots, x_n)$. The cluster algebra $\cA$ is the $\bZ\bP$-subalgebra of the ambient field $\cF$ generated by all cluster variables in $\cX$:
\be
\label{def cluster algebra}
\cA(B,\why,\ex) = \bZ\bP[\cX] \;.
\ee
We note that we could have used any one of the seeds obtained by mutating the initial seed to obtain an isomorphic definition of the cluster algebra, hence the notation $\cA = \cA(B,\why,\ex)$.

\

The matrix mutation of $B$ in (\ref{cluster algebra transformations}) can be realized graphically in the quiver representation. Upon acting with $\mu_k$, $B'$ is obtained from $B$ through the following three steps:
\begin{enumerate}[\; \; {Step} 1)]
\item For each ``path'' (a sequence of two arrows) $i \to k \to j$ passing through $k$, add an arrow $i \to j$.
\item Invert the direction of all arrows that start or end at $k$.
\item If, as a result of the manipulations of step one, two nodes $i$ and $j$ are connected by arrows in both directions, remove pairs of opposite arrows until the remaining arrows (if any) point in a unique direction.
\end{enumerate}
This procedure does not generate any 1-cycles or 2-cycles, thus a cluster quiver is mapped to a cluster quiver.

The action of $\mu_k$ in (\ref{cluster algebra transformations}) has an interesting hierarchical structure, which enables one to focus on how a subset of the data constituting the seeds transform under mutations, without worrying about the behavior of the other components. For instance, the action of the mutations on quivers does not depend on the coefficients or the clusters. Hence, the quiver (or matrix) mutation $B' = \mu_k(B)$ can be examined on its own. Also, the mutation rules of quivers and coefficients do not depend on the cluster variables, \ie{}, the so-called coefficient dynamics $(B',\why') = \mu_k(B,\why)$ can be studied independently. Another interesting restriction is to set all the coefficients of a seed to be trivial, \ie{}, $y_i = 1$; then the coefficients remain trivial in all other seeds and the cluster algebra structure reduces to $(B',\ex')= \mu_k(B,\ex)$.

We have now introduced all the elements of cluster algebra we need. For the rest of the section, we show that the Seiberg-like dualities of two-dimensional gauge theories based on cluster quivers elegantly realize those elements. Cluster algebras exhibit many interesting properties, such as total positivity, the Laurent phenomenon and the existence of a natural Poisson structure. Their relevance for two-dimensional physics certainly deserves further study.

\subsection{Dualities of quiver gauge theories}
\label{ss:SDandCA}

Let us consider two-dimensional $\cN=(2,2)$ gauge theories whose gauge group and matter content is based on a quiver diagram. We consider diagrams whose nodes are either circles or squares. Let $m$ be the number of circles, and $n$ the total number of nodes.%
\footnote{In the mathematical literature, given a cluster algebra of rank $m$ with coefficients in the tropical semifield in $J$ variables $\bP_J$, one can introduce an extended $m \times (m+J)$ exchange matrix $\tilde B$ which includes both $B$ and $\why$, in such a way that the extended matrix transforms as in the first line of (\ref{cluster algebra transformations}). In other words, one can represent the coefficients by the number of arrows connecting the $m$ ``gauge" nodes to $J$ extra ``flavor" nodes (as mutations cannot be taken on these nodes), or vice versa. Unfortunately, this description loses track of the arrows between flavor nodes, so we do not follow this route. We rather connect the field theory to a cluster algebra of rank $n$ with coefficients in $\bP_{J=1}$, allowing mutations only at the gauge nodes.}
The gauge group is a product of unitary factors
$$
U(N_1) \times \cdots \times U(N_m) \;,
$$
each associated to a circular node of the diagram, while the flavor group includes
$$
S \big[ U(N_{m+1}) \times \cdots \times U(N_n) \big] \;,
$$
where each factor is associated to a square node. The ``missing" $U(1)$ factor in the flavor group is gauged.%
\footnote{If there are no flavor nodes, then the diagonal $U(1)$ of the gauge group is decoupled and free. Moreover, the flavor group can be much larger: each independent loop brings an extra $U(1)$ factor, and arrows with multiplicity $|b_{ij}|$ bring an $SU\big( |b_{ij}| \big)$ factor, unless they are broken by superpotential interactions. The superpotential interactions, as we explain further in section \ref{ss:chiralring}, are also encoded in the R-charge assignments of the chiral fields.}
For each arrow from node $i$ to node $j$, there is a chiral multiplet in the bifundamental representation, which transforms as a fundamental of the node at the tail and an antifundamental of the node at the head. Notice that, from the point of view of the gauge group, these multiplets are either bifundamentals, (anti)fundamentals or singlets. We restrict our attention to theories whose quivers are cluster quivers, which are free of 1- and 2-cycles: they are represented by an $n\times n$ skew-symmetric matrix $B = b_{ij}$ such that $[b_{ij}]_+$ is the number of arrows pointing from $i$ to $j$.

The field theory attaches additional data to the nodes of the quiver diagram. First, each node $i \in [n]$ is assigned a non-negative integer $N_i$ representing the rank of the corresponding group. Each gauge node also has an associated complexified FI parameter
\be
t_i = 2\pi\xi_i + i\theta_i \;,
\ee
which is more conveniently expressed through the \emph{K\"ahler coordinate}
\be
\label{defz}
z_i = e^{i\pi \left( N_f(i) - N_i \right)} \; e^{-t_i} \;.
\ee
Here, $N_f(i)$ and $N_a(i)$ are the total number of fundamentals and antifundamentals with respect to the $i$-th node, and are given by
\be
N_f(i) = \sum\nolimits_j [b_{ij}]_+ \, N_j \;, \qquad\qquad N_a(i) = \sum\nolimits_j N_j \, [b_{ji}]_+ \;.
\ee
To make the notation more homogeneous, one can choose to introduce ``flavor FI parameters'' for the flavor nodes as well: these parameters do not have physical significance on flat space since they do not involve dynamical fields, but are useful for keeping track of the contact terms necessary to match the $S^2$ partition functions. In two dimensions, FI terms are classically marginal, but quantum mechanically they can have logarithmic running with the dynamical generation of a scale; with $\cN=(2,2)$ supersymmetry, the beta function is one-loop exact and is given by
\be
\beta_i = N_a(i) - N_f(i) = - \sum\nolimits_j b_{ij} N_j \;.
\ee

\begin{figure}[t]
\begin{center}
\includegraphics[width=.7\textwidth]{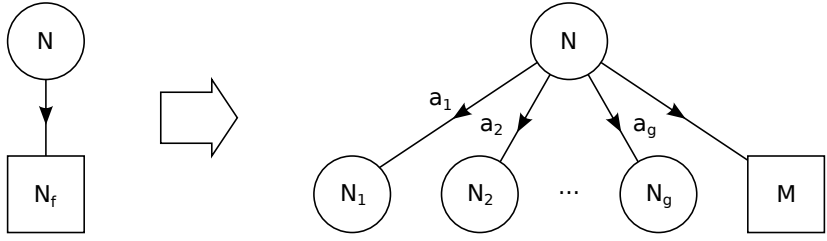}
\end{center}
\caption{Gauging the flavor symmetry. We promote a $U(N_1) \times \cdots \times U(N_g)$ subgroup of the $U(N_f)$ flavor symmetry to be a gauge symmetry, in such a way that $U(N_i)$ has embedding index $a_i$, with $\sum_i a_i N_i + M = N_f$. This leaves an unbroken $U(M)$ flavor symmetry, as well as a $\prod_i SU(a_i)$ symmetry, the latter of which is not represented by any node. These additional flavor symmetries can be broken by superpotential terms. The resulting quiver has $a_i$ arrows pointing to $U(N_i)$, and one arrow pointing to $U(M)$.
\label{fig: gauging}}
\end{figure}

Let us now apply the Seiberg-like duality of section \ref{sec: dualities} to the $k$-th gauge node in the quiver---the duality, obviously, cannot be applied to a flavor node. We can think of the full quiver as constructed by taking blocks of $U(N_k)$ gauge groups with $N_f(k)$ fundamentals and $N_a(k)$ antifundamentals, and identifying part or all of the flavor symmetries as gauge symmetries of other gauge groups of the theory. This process is depicted in figure \ref{fig: gauging}. Let us consider, for concreteness, a gauge node with $N_f$ fundamentals: we can embed $U(N_1) \times \cdots \times U(N_g) \times U(M)$ into $U(N_f)$ in such a way that the embedding index of $U(N_i)$ into $U(N_f)$ is $a_i$ and that of $U(M)$ is 1, as long as $M+ \sum_i a_i N_i = N_f$.%
\footnote{If a representation $\mathbf{r}$ of $G$ decomposes into $\oplus_i \mathbf{r}_i$ of $H$ under the embedding $H \subset G$ of Lie algebras, the Dynkin embedding index
\be
I_{H \hookrightarrow G} = \frac{\sum_i T(\mathbf{r}_i)}{T(\mathbf{r})}
\ee
is independent of $\mathbf{r}$. Here, $T(\mathbf{r})$ is the quadratic Casimir, normalized such that $T(\mathbf{n})=1$ for the fundamental representation $\mathbf{n}$ of $\mathfrak{su}(n)$.}
When we gauge $U(N_1), \ldots, U(N_g)$, the original arrow pointing to $U(N_f)$ breaks into $a_i$ arrows pointing to $U(N_i)$ for each $i$, and one arrow pointing to the flavor node $U(M)$. The gauge fields $U(N_i)$ will in general also be coupled to other chiral multiplets, therefore the nodes will be connected to the rest of the quiver. An example of this process is given in figure \ref{fig: UNtoGen}.

Based on what we already know about the dualities of SQCD-like theories, we can infer the action of the duality applied to the $U(N_k)$ node of the quiver theory. The action on the quiver diagram can be described in the following way. Let us, as in section \ref{sec: dualities}, refer to the theories before and after the duality as theories $\fA$ and $\fB$, respectively. Recalling the action of the duality on the SQCD-like theory depicted in figure \ref{fig: Seiberg duality}, we see that it realizes the first and second steps of the quiver mutation rule summarized at the end of section \ref{ss:CA}. The fields $M_{AF}$, the singlets with respect to $U(N_k)$ appearing in theory $\fB$, are the arrows added in step 1). The fact that fundamentals and antifundamentals of $U(N_k)$ get exchanged corresponds to step 2). What is missing in the SQCD-like example is step 3). Let us consider the case when an oriented 2-cycle is generated in theory $\fB$ upon the addition of the fields $M_{AF}$: we denote the two chiral multiplets forming the 2-cycle as $X_1$ and $X_2$, where $X_2$ is the newly added field. Suppose that in theory $\fB$ there is a quadratic superpotential term $W_\text{cycle} = X_1X_2$: then the fields $X_{1,2}$ are massive and can be integrated out. In other words, they disappear in pair at low energies. This mechanism would realize step 3). Going back to theory $\fA$, this implies that there must have been a 3-cycle $X_1$-$\tilde q$-$q$ involving some quarks of the $U(N_k)$ node, and a cubic superpotential $W_\text{cycle} = X_1 \tilde qq$: only then is the quadratic term $X_1 X_2$ present in theory $\fB$, as the operators $\tilde qq$ and $X_2$ are identified under the map (\ref{chiral ring map}).

For the rest of this section, we only consider theories for which this is always the case: we assume that the theories under consideration have a quiver and a superpotential such that, whenever the application of the Seiberg-like duality to a node generates oriented 2-cycles in the quiver, there are enough quadratic superpotential terms to make all fields involved massive. With this assumption, we always integrate such pairs out when defining the duality map. In fact, we assume something stronger: that this is the case for all possible sequences of dualities. This is a very nontrivial assumption, which we come back to in section \ref{sec: nondegenerateFTs}. We denote theories with this property as \emph{non-degenerate} quiver theories. Moreover, since the dualities considered here involve one more step than the Seiberg-like dualities of section \ref{sec: dualities}---namely the integrating out of massive pairs of chiral multiplets---we sometimes choose to refer to them as \emph{cluster dualities}.

Hence, under the assumption of non-degeneracy, the transformation of the quiver diagram under duality at $U(N_k)$ is the same as the matrix mutation $\mu_k$ of $b_{ij}$ described by the first equation in (\ref{cluster algebra transformations}). In particular, a cluster quiver is mapped to a cluster quiver.

Let us now examine how the ranks of the gauge groups transform. From equation (\ref{map ranks}), we see that only the rank of the dualized node $U(N_k)$ changes, therefore:
\be
\label{bN}
N_i' = \begin{cases} \max \Big( \sum_j [b_{kj}]_+ \, N_j ,~ \sum_j N_j \, [b_{jk}]_+ \Big) - N_k & \text{if $i=k$} \\
N_i & \text{otherwise} \;. \end{cases}
\ee
We can represent the ranks by elements of the tropical semifield $\bP$ in a single variable $u$ (introduced in (\ref{tropical semifield}), with $J=1$) by attaching a variable $\mathsf{n}_i = u^{N_i}$ to the $i$-th node of the quiver diagram. Then, the transformation rules for the ranks can be elegantly written in $\bP$ as
\be
\label{ranks mutation}
\mathsf{n}_i' = \begin{cases} \mathsf{n}_k^{-1} \Big( \prod_j \mathsf{n}_j^{[b_{jk}]_+} \oplus \prod_j \mathsf{n}_j^{[-b_{jk}]_+} \Big) & \text{if $i=k$} \\ \mathsf{n}_i & \text{otherwise} \;. \end{cases}
\ee
We can similarly introduce elements $y_i \in \bP$ associated to the beta functions of the nodes:
\be
y_i = u^{\beta_i} = u^{N_a(i) - N_f(i)} = \prod\nolimits_j \mathsf{n}_j^{b_{ji}} \;.
\ee
We can identify $y_i$ with the cluster algebra coefficients, as it is a simple exercise to show that they transform according to the rules given by the second equation of (\ref{cluster algebra transformations}).

We say that a quiver gauge theory is ``conformal'' if the FI beta functions $\beta_i$ of all gauge nodes vanish, \ie{}, $\beta_i = 0$.%
\footnote{The gauge theory is not conformal because the gauge coupling is dimensionful. The condition $\beta_i=0$, however, implies that the theory flows to a nontrivial fixed point in the IR. In this case, the FI parameters are marginal couplings that can be thought of as coordinates on the conformal manifold of the IR theory.}
In this case, $y_i = 1$ for all gauge nodes $i$. The transformation of $\mathbf{y}$ implies that conformal quivers are mapped to conformal quivers under the dualities, as it should.

\begin{figure}[t]
\centering\includegraphics[width=.35\textwidth]{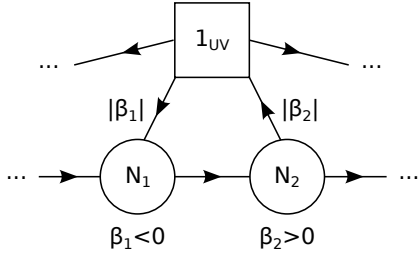}
\caption{\small The ``UV completion'' of a non-conformal quiver where we extend the theory by a $U(1)$ flavor node. The twisted mass with respect to this flavor symmetry sets the UV renormalization scale $u$. The new node is connected to the $i$-th node by $|\beta_i|$ arrows, in the direction such that the extended quiver is conformal.}
\label{f:Extended}
\end{figure}

Finally, let us study the transformation of the complexified FI parameters---or, equivalently, of the K\"ahler coordinates $z_i$---under cluster dualities. We already know from section \ref{sec: dualities} that dualization at the node $U(N_k)$ maps its K\"ahler coordinate from $z_k$ to $z_k^{-1}$. The contact terms represented by the function $f_\text{ctc}$ in (\ref{f ctc}) are responsible for the transformation of the K\"ahler coordinates of the neighboring nodes. Let us first consider this map in the case of conformal quivers. From (\ref{f ctc}) or, more precisely, from the analysis that we present at the end of this section, we deduce the map
\be
\label{zmapconformal}
z_i' = \begin{cases} z_k^{-1} & \text{if $i=k$} \\ z_i \, z_k^{[b_{ki}]_+} (z_k + 1)^{-b_{ki}} & \text{otherwise} \;. \end{cases}
\ee
These are precisely the transformation rules for the $\mathcal{X}$-coordinates of Fock and Goncharov \cite{FockGoncharov, Fock:2003xxy}.

In the non-conformal case, the complexified FI parameters run logarithmically along the RG flow, and must be defined at some scale $\mu$. An alternative way to define the theory is to embed the non-conformal quiver $B$ into a conformal UV quiver $B_\text{UV}$. The UV quiver is constructed by adding an extra $U(1)_\text{UV}$ flavor node, as shown in figure \ref{f:Extended}, and $|\beta_i|$ arrows connecting the node $U(N_i)$ to the node $U(1)_\text{UV}$ for each $i$ in such a way that the extended quiver is conformal. The original gauge theory can be recovered by turning on a large twisted mass $s_0$ for the $U(1)_\text{UV}$ flavor symmetry, and scaling the UV parameters $z_i^\text{UV}$ appropriately. The identification of parameters is
\be
\label{UVIR}
z_i^\text{UV} = u^{\beta_i} \, z_i(\mu) \;,
\ee
where $u = -is_0/\mu$ is the ratio of the UV and IR scales.%
\footnote{Our conventions are adapted to the $S^2$ partition function, which can be used to derive the precise maps under scale matching. We present the details of the computation with all relative signs in appendix \ref{ap:UVIR}.}
The IR limit $T'_\text{IR}$ of the UV theory $T'_\text{UV}$ dual to a given theory $T_\text{UV}$, is itself the dual to the IR limit $T_\text{IR}$ of $T_\text{UV}$. In other words, the diagram
\be
\begin{CD}
T_\text{UV}	@>{\text{cluster dual}}>>	T_\text{UV}' \\
@V\text{IR limit}VV	@VV{\text{IR limit}}V \\
T_\text{IR}	@>{\text{cluster dual}}>>	T_\text{IR}'
\end{CD}
\ee
commutes. Hence, upon identifying
\be
y_i = u^{\beta_i} = z_i^\text{UV} / z_i(\mu) \;,
\ee
we find an interpretation for the base $u$ of the tropical semifield $\bP$ as the ratio of the UV and IR scales. Since the UV theory is conformal by construction, its K\"ahler coordinates $z^\text{UV}_i$ transform as in (\ref{zmapconformal}). Thus, from $z_i(\mu) = y_i^{-1} z_i^\text{UV}$ and the transformation laws (\ref{cluster algebra transformations}) and (\ref{zmapconformal}), we obtain the transformation of the K\"ahler coordinates in the general case:
\be
\label{zrule}
z'_i = \begin{cases} z_k^{-1} & \text{if $i=k$} \\ z_i \, z_k^{[b_{ki}]_+} \bigg( \dfrac{y_k}{y_k \oplus 1} z_k + \dfrac{1}{y_k \oplus 1} \bigg)^{-b_{ki}} & \text{otherwise} \;. \end{cases}
\ee
While objects with this transformation law do not appear directly in the work of Fomin and Zelevinsky, they follow from the transformation of the cluster coordinates $x_i$ in (\ref{cluster algebra transformations}) if we define
\be
\label{zandx}
z_i = \prod\nolimits_j x_j^{b_{ji}} \;.
\ee
The variables $z_i$ are sometimes referred to as ``dual cluster variables."

The physical content of the formal expression (\ref{zrule}) is not immediately apparent. Recall, however, that its derivation still involves the UV construction with finite UV scale. The original quiver theory is obtained in the IR limit of this construction, that is, when $u\to\infty$. The transformation (\ref{zrule}) then reduces to
\be
\label{Kahler coord transform}
z'_i = \begin{cases} z_k^{-1} & \text{if $i=k$} \\
z_i \, z_k^{[-b_{ki}]_+} & \text{if $i \neq k$ when $\b_k >0$} \\
z_i \, z_k^{[b_{ki}]_+}(z_k+1)^{-b_{ki}} & \text{if $i \neq k$ when $\b_k =0$} \\
z_i \, z_k^{[b_{ki}]_+} & \text{if $i \neq k$ when $\b_k <0$} \;.
\end{cases}
\ee
In fact, taking the limit $u\to\infty$ in (\ref{zrule}) is the same as selecting only those terms in $z_i'$ whose coefficient is $1 = u^0 \in\bP$. Now, \eqref{Kahler coord transform} is precisely the transformation law that follows from the contact terms represented by $f_\text{ctc}$ in (\ref{f ctc}), derived from the $S^2$ partition function.

\

For the rest of this section, we check that the duality map (\ref{Kahler coord transform}) agrees with the $S^2$ partition function of the whole quiver theory. This is a direct consequence of our analysis of the SQCD-like theories in section \ref{sec: dualities}. To do so, let us consider the $S^2$ partition function of a quiver gauge theory characterized by the data $(B,\eN,\zee)$, where $B$ is a quiver, $\eN$ is an $n$-tuple of ranks and $\zee$ is an $m$-tuple of K\"ahler coordinates:
\begin{multline}
\label{ZBNz}
Z_{S^2} (B,\eN,\zee) = \prod_{i=1}^n \Bigg[ \prod_{I_i = 1}^{N_i} \, \sum_{\m_{i, I_i} \in \bZ} \, \int \frac{d\s_{i,I_i}}{2\pi} \,  \big( e^{i\pi(N_f(i)-1)} z_i \big)^{\s_{i,I_i +}} \big( e^{-i\pi(N_f(i)-1)} \bz_i)^{\s_{i,I_i -}} \Bigg] \\
\times \prod_{i=1}^n \Bigg[ \prod_{I_i <J_i}^{N_i} \big( {-\S^{I_i}_{{J_i}+} \S^{I_i}_{{J_i}-}} \big) \Bigg] \;
\prod_{i < j}^n \Bigg[ \prod_{\alpha_{ij}=1}^{|b_{ij}|} \, \prod_{I_i,J_j} \frac{\Ga \big( \frac12 r^{ij}_{\alpha_{ij}}-\sign(b_{ij})\S^{I_i}_{J_j+} \big) }{ \Ga \big( 1- \frac12 r^{ij}_{\alpha_{ij}}+\sign(b_{ij})\S^{I_i}_{J_j-} \big)} \Bigg] \;.
\end{multline}
Here, $\s_{I_i \pm}$ are the complex Coulomb branch parameters of the gauge group $U(N_i)$ as defined in (\ref{def sigma+-}), where the indices $I_i, J_i$ run over $[N_i]$. We have used the notation $\S^{I_i}_{J_j \pm}$ to denote the differences between Coulomb branch parameters:
\be
\S^{I_i}_{J_j \pm} = \s_{i,I_i\pm}-\s_{j,J_j\pm} \;.
\ee
The indices $\alpha_{ij}$ label the bifundamental matter between node $i$ and $j$ (the multiple arrows), and $r^{ij}_{\alpha_{ij}}$ are their R-charges. For notational simplicity we have assumed that all nodes of the quiver are gauge nodes, although the following analysis is valid for theories with flavor nodes as well. If a node $i$ is a flavor node, we can simply remove the summation over fluxes and the integration over the Coulomb branch, $\sum_\fm \int \frac{d\sigma}{2\pi}$, upon which the complex Coulomb branch parameters become components of the complex twisted masses of the flavor symmetries. The corresponding FI parameters then become formal parameters that keep track of the contact terms: we can simply choose to set $t=0$, although $t$ will then take some other values in dual descriptions, as described in section \ref{sec: partition function} for SQCD-like theories.

Now let us examine the duality at the node $k$. It is simple to observe that the integrand in (\ref{ZBNz}) can be factorized into a product of two factors $G_k$ and $G_k^c$, where $G_k$ depends on the parameters $\sigma_{k, I_k\pm}$ of the $k$-th node, while $G_k^c$ does not. One can first perform the integral/summation of $G_k$ over $\prod_{I_k} \sum_{\fm_{k,I_k}} \int d\sigma_{k,I_k}$---we call this integral $Z_k$---and then integrate/sum $G_k^c Z_k$ over the remaining parameters. In fact, $Z_k$ is the partition function of a $U(N_k)$ theory with $N_f(k)$ fundamentals and $N_a(k)$ antifundamentals, and we can readily apply the identity (\ref{UNUNprime compact}). The Coulomb branch parameters of the adjacent gauge nodes to $k$ should be plugged into the formula (\ref{UNUNprime compact}) as twisted masses, as (a subgroup of) the flavor symmetry has been gauged; notice that if the gauging is performed with an embedding index $I>1$, then $I$ twisted masses should be set equal to the same Coulomb branch parameter. Eventually, the various dual flavors and singlets can be recognized as bifundamentals of the dual theory, as exemplified in figure \ref{fig: UNtoGen}. The function $f_\text{imp}$ in (\ref{UNUNprime compact}) does not depend on integrated/summed parameters and can be pulled out of the integral/sum; the same is not true for $f_\text{ctc}$.

\begin{figure}[t]
\centering\includegraphics[width=.6\textwidth]{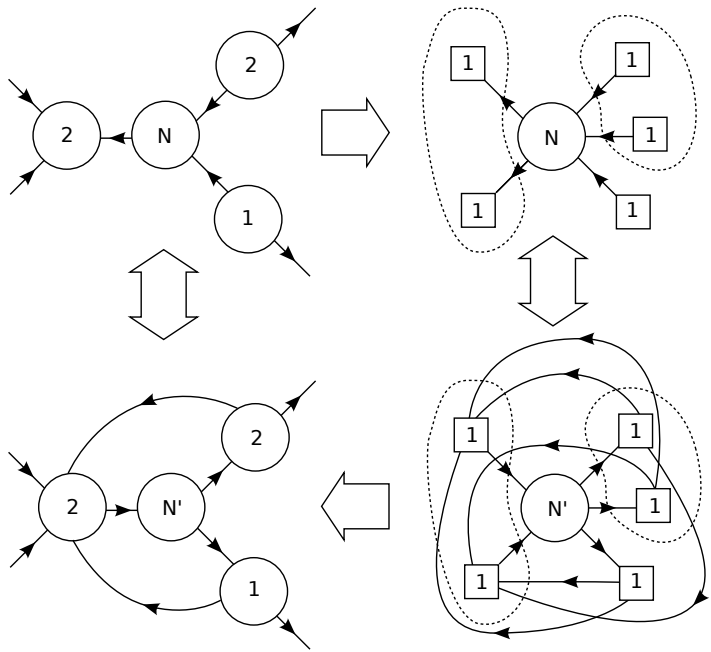}
\caption{Local picture of cluster duality. From the point of view of the dualized node $U(N)$, the Coulomb branch parameters of adjacent nodes play the role of twisted masses. The duality properties of $Z_{S^2}$ can be inferred by: 1) isolating the terms of the partition function participating in the duality and treating the Coulomb branch parameters as twisted masses (upper arrow); 2) dualizing (right arrow); 3) regrouping the twisted masses into Coulomb parameters of neighboring gauge groups (lower arrow). In this example all embedding indices are 1.
\label{fig: UNtoGen}}
\end{figure}

The duality at the partition function level encapsulates the map for the quiver and the ranks of the nodes stated before. In particular, this procedure reproduces the steps 1) and 2) in the mutation of the quiver $B$, as well as the transformation of ranks. It is less obvious how step 3) arises at the level of the partition function. Under the assumption of non-degeneracy, every 2-cycle generated on the $\fB$-side is accompanied by a quadratic superpotential term. This constrains the two chiral multiplets involved, $X_{1,2}$, which have opposite gauge and flavor charges, to have R-charges that sum up to 2. It follows that the product of their one-loop determinants is simply a flux-dependent sign:
\be
\prod_{I_i,J_j}  \frac{ \Ga \big( \S^{I_i}_{J_j+}+ \frac r2 \big) }{ \Ga \big( 1-\S^{I_i}_{J_j-} - \frac r2 \big)} \;\cdot\; \prod_{I_i,J_j} \frac{ \Ga\big( 1 -\S^{I_i}_{J_j+} - \frac r2) }{ \Ga \big( \S^{I_i}_{J_j-} + \frac r2 \big)} =(-1)^{N_j \sum_{I_i} \fm_{I_i} - N_i \sum_{J_j} \m_{J_j} } \;,
\ee
which is a consequence of gamma function identities and the fact that $\m_{I_i} = \S_{I_i+}-\S_{I_i-}$ is integral. The only trace of their existence is a shift of the theta angles of the gauge groups $i$ and $j$ by $N_j \pi$ and $N_i \pi$, respectively.

The map of the K\"ahler coordinates can be read off from the partition function from the factors that are exponentiated by the Coulomb branch parameters. These factors involve crucial contributions from $f_\text{ctc}$. For the dualized node we clearly have $z_k' = z_k^{-1}$ from (\ref{UNUNprime compact}). The map (\ref{Kahler coord transform}) for $i \neq k$ can be derived from the following expressions:
\be
\label{compcoeff}
\frac{ e^{i\pi (N_f(j)'-1)} z_j' }{ e^{i\pi (N_f(j)-1)} z_j} = \begin{cases}
\big( e^{i\pi N_k'} z_k \big)^{[b_{kj}]_+} \big( e^{i\pi N_k'} \big)^{-[-b_{kj}]_+} \prod\limits_{i \neq k} e^{i\pi N_i a_{ij}} & N_f > N_a \\
\Big[ \dfrac{e^{i\pi N_k'} z_k }{ 1+ z_k } \Big]^{[b_{kj}]_+} \Big[ e^{i\pi N_k'} (1+z_k) \Big]^{-[-b_{kj}]_+} \prod\limits_{i \neq k} e^{i\pi N_i a_{ij}} & N_f = N_a  \\
\big( {e^{i\pi N_k'}} \big)^{[b_{kj}]_+} \big( e^{i\pi(N_f(k)-N_k)}z_k \big)^{-[-b_{kj}]_+} \prod\limits_{i \neq k} e^{i\pi N_i a_{ij}} &
N_f < N_a \,.
\end{cases}
\ee
The first two terms on the right-hand side of these equations come from $f_\text{ctc}$; see equation \eq{f ctc}. The last terms come from integrating out pairs of massive fields with quadratic superpotential couplings, and $a_{ij}$ denotes the number of ``annihilated'' 2-cycles between the nodes $i$ and $j$. By isolating the $z_k$ dependence on the right-hand side of this equation, one obtains the map (\ref{Kahler coord transform}) up to shifts of the theta angles by multiples of $\pi$. The cancellation of these shifts can be shown by using some non-trivial relations among the variables involved, which we work out in appendix \ref{ap:thetacancel}.

\subsection{The twisted chiral ring}
\label{ss:CAofTCR}

The operator map of the twisted chiral ring under cluster dualities can be deduced from the operator map studied in section \ref{sec: twisted chiral ring}. This is described efficiently using the $Q$-polynomial introduced in that section. In a quiver theory, we have a $Q$-polynomial for each gauge group:
\be
Q_i (x) = \det(x-\sigma_i) \;,
\ee
where $\s_i$ is the lowest component of the adjoint twisted chiral multiplet of $U(N_i)$. Consider the cluster duality with respect to node $k$, and denote the $Q$-polynomials of the dual theory as $Q_i' (x)= \det(x-\sigma_i')$. Similarly, we define $Q$-polynomials for the flavor nodes:
\be
Q_f (x)= \prod_{F=1}^{N_f} (x-s_{f,F}) \;,
\ee
where $s_{f,F}$ are the twisted masses lying in the Cartan subalgebra of the $U(N_f)$ flavor symmetry. Note that while the coefficients of the gauge $Q$-polynomials are operators, those of the flavor nodes are constant parameters.

Since the duality does not modify the twisted chiral operators of nodes other than $k$,
\be
Q_i'(x) =  Q_i(x) \qquad\qquad \text{for $i \neq k$} \;.
\ee
Meanwhile, upon repeating the argument of section \ref{sec: twisted chiral ring}, the map for $Q_k$ can be derived from the twisted chiral ring relation
\be
\label{QTCR}
\prod\nolimits_j Q_j (x)^{[b_{kj}]_+} + i^{N_a(k)-N_f(k)} z_k \prod\nolimits_j Q_j (x)^{[-b_{kj}]_+} = C_k (z_k) \, Q_k (x) \, Q_k' (x) \;,
\ee
which is the same as \eqref{TCR}. The function $C_j(z)$ is defined as
\be
C_j (z) = \begin{cases}
1 & \text{when $N_f(j) > N_a(j)$} \\
1+ z & \text{when $N_f(j)=N_a(j)$} \\
i^{N_a(j)-N_f(j)}z & \text{when $N_f(j) < N_a (j)$} \;.
\end{cases}
\ee
We can rewrite the operator map as
\be
\label{Qprime}
Q'_k(x) = \frac{ i^{N_a(k)-N_f(k)} z_k \prod_j Q_j (x)^{[b_{jk}]_+} + \prod_j Q_j (x)^{[-b_{jk}]_+}}{ C_k(z_k) \, Q_k(x)} \;,
\ee
which resembles a cluster algebra exchange relation. In fact, we can make the connection more precise. Assuming that there exist cluster variables $x_i$ such that $z_i = \prod_j x_j^{b_{ji}}$, we find that they transform as%
\footnote{If the exchange matrix $b_{ij}$ is invertible, clearly $x_i = \prod_j z_j^{b_{ji}^{-1}}$. The matrix, however, in general is not invertible: \eg{}, when the number of nodes is odd, the exchange matrix is never invertible. When this is the case, one can introduce one or more auxiliary flavor nodes suitably connected to the quiver, with rank zero but with an FI term: these extra nodes can be used to make $b_{ij}$ effectively invertible.}
\be
\label{modclust}
x_i' = \begin{cases}
x_k^{-1} \prod_j x_j^{[b_{jk}]_+} & \text{if $i=k$ and $\b_k>0$} \\
x_k^{-1} \Big( \prod_j x_j^{[b_{jk}]_+} + \prod_j x_j^{[-b_{jk}]_+} \Big) & \text{if $i=k$ and $\b_k=0$} \\
x_k^{-1} \prod_j x_j^{[-b_{jk}]_+} & \text{if $i=k$ and $\b_k<0$} \\
x_i & \text{if $i\neq k$} \;. \end{cases}
\ee
under cluster dualities. These rules can be derived by taking the $u\to\infty$ limit of the standard transformation of cluster variables in (\ref{cluster algebra transformations}). One can verify that
\be
x_k'= i^{N_f(k) - N'_k - N_k} \, C_k(z_k) \, x_k^{-1} \, \prod\nolimits_j x_j^{[-b_{jk}]_+}
\ee
where $N_f(k) - N'_k - N_k = -[\beta_k]_+$. We now define the ``dressed $Q$-polynomials":
\be
\cQ_j (x) =  i^{N_k} x_j ~\!Q_j (x) \,.
\ee
The auxiliary variable $x$ should not be confused with the cluster variables $x_i$. The relation \eq{Qprime} can then be written as
\be
\cQ'_k (x) = \frac{\prod_j \cQ_j(x)^{[b_{jk}]_+} + \prod_j \cQ_j(x)^{[-b_{jk}]_+} }{\cQ_k(x)} \;.
\ee
For $i \neq k$, $\cQ_i(x)$ remains invariant, therefore
\be
\label{Q map}
\cQ'_i(x) = \begin{cases}
\cQ_k(x)^{-1} \Big( \prod_j \cQ_j(x)^{[b_{jk}]_+} + \prod_j \cQ_j(x)^{[-b_{jk}]_+} \Big) & \text{if $i=k$} \\
\cQ_i(x) &\text{otherwise} \;.
\end{cases}
\ee
This is exactly the transformation law of cluster variables with trivial coefficients. While the operator relations are nicely summarized by (\ref{Qprime}), some work must be done to recover the actual map of operators implied in this equation. As in the case of SQCD-like theories, the map between the operators of theories $\fA$ and $\fB$ can be obtained by expanding the equation \eq{QTCR} in the variable $x$ and comparing the coefficients.

\

The coefficients of the $Q$-polynomials generate the twisted chiral ring, but are subjected to the relations (\ref{QTCR}). Let us rewrite them as
\be
\label{vacuum equations}
\cG_i ( B, \eN,\zee;\Qyoo) \;\equiv\; \prod\nolimits_j Q_j (x)^{[b_{ij}]_+} + i^{N_a(i)-N_f(i)} z_i \prod\nolimits_j Q_j (x)^{[-b_{ij}]_+} \;=\; C_i (z_i)  \, Q_i (x) \, T_i (x) \;,
\ee
where $\Qyoo = (Q_i)$ denotes the array of $Q$-polynomials. These equations are to be solved for monic polynomials $Q_i(x)$ of degree $N_i$ and $T_i(x)$ of degree $\max\big( N_f(i), N_a(i) \big) - N_i$, where $i$ runs over the gauge nodes. Unfortunately, in general this is not the whole set of relations that define the twisted chiral ring. Meanwhile, in the dual theory $\fB$ the $Q$-polynomials $Q'_i(x)$ satisfy
\be
\label{vacuum equations B}
\cG_i ( B',\eN',\zee';\Qyoo' ) = C_i' (z_i')  \, Q_i' (x) \, T_i' (x) \,.
\ee
It is natural to expect that the twisted chiral rings of theories $\fA$ and $\fB$ are equivalent under \eq{Qprime}; while we do not have a proof, we present some evidence.

A solution $V = (Q_i)$ to the equations \eq{vacuum equations}, \ie{}, an array of polynomials, represents a Coulomb branch configuration $\s_i$ that is a critical point of the effective twisted superpotential $\widetilde W_\text{eff}$. Such a configuration can be safely interpreted as a Coulomb branch vacuum of the theory if all chiral multiplets and off-diagonal vector multiplets are massive around it \cite{Witten:1993yc}, because then the approximation used to derive the vacuum is self-consistent. Therefore, $V$ is surely a vacuum if all $Q_i(x)$ have non-degenerate roots and all pairs $Q_i(x)$, $Q_j(x)$ at adjacent nodes ($b_{ij} \neq 0$) do not have any common roots. We call such solutions, \emph{strong solutions}. On the other hand, if any one of the conditions is not met, a more refined analysis is necessary to determine whether such solutions represent true vacua of the theory: we call these \emph{weak solutions}.

We show that every strong solution $V$ to the equations (\ref{vacuum equations}) of theory $\fA$ are mapped to solutions
\be
V' = \mu_k (V) \,\equiv\, \big(\, \mu_k (Q_i) \, \big)
\ee
to the equations (\ref{vacuum equations B}) of theory $\fB$ under \eq{Qprime}. By construction, $V'$ solves \eq{vacuum equations B} at node $k$. Also, since $\cG_i$ remain unaltered for nodes $i$ that are not adjacent to $k$, $V'$ solves (\ref{vacuum equations B}) for all $i$ such that $b_{ik}=0$.

It remains to show that $V'$ solves (\ref{vacuum equations B}) for $i$ adjacent to $k$. For $i\neq k$, the polynomial $Q_i$ remains invariant under duality and hence is non-degenerate. Then, the equation (\ref{vacuum equations B}) at node $i$ is equivalent to the condition that every root of $Q'_i(x)$ is also a root of $\cG_i(B', \eN', \zee';V')$, \ie{}, we need to show that given an $\alpha$ such that $Q_i'(\alpha) = Q_i(\alpha) = 0$, it implies $\cG_i ( B',\eN',\zee';V') \big|_{x=\alpha}=0$. The equation (\ref{vacuum equations}) on the $\fA$ side already implies that
\be
\cG_i ( B,\eN,\zee;V) \big|_{x=\alpha}=0 \;.
\ee
It is useful to define the rational functions
\be
\cZ_j (x) = i^{N_a(j)-N_f(j)} z_j \prod\nolimits_l Q_l (x)^{b_{lj}}
=  \prod\nolimits_l \cQ_l (x)^{b_{lj}} \,.
\ee
From the last identity, we see that they transform precisely as the conformal K\"ahler coordinates $z_j$ do under cluster dualities, due to the transformation rules \eq{Q map} of the dressed $Q$-polynomials. Let us assume that $b_{ik} >0$. There are two cases we need to consider:
\begin{enumerate}[\, 1)]
\item When $Q_j(\alpha) \neq 0$ for all $b_{jk}<0$, the rational function $\cZ_k (x)$ has a zero at $x=\alpha$, \ie{},
\be
\label{cZk}
\cZ_k (\alpha) =0 \;.
\ee
Also, since
\be
\cG_i (B,\eN,\zee;V) \big|_{x=\alpha} = \prod\nolimits_j Q_j(\alpha)^{[-b_{ji}]_+} \big( 1 + \cZ_i (\alpha) \big) =0 \;,
\ee
and since $Q_j(\alpha) \neq 0$ for all nodes $j$ adjacent to $i$ because $V$ is a strong solution, it follows that
\be
\label{cZi}
\cZ_i (\alpha) = -1 \;.
\ee
Upon dualizing, it is simple to verify that
\be
\cG_i ( B',\eN',\zee';V') = \prod\nolimits_{j \neq k} Q_j(x)^{[-b_{ji}']_+} \Big(1 + \cZ_i (x) \big (1+\cZ_k (x) \big)^{b_{ik}} \Big) \;.
\ee
Equations \eq{cZk} and \eq{cZi} hence imply
\be
\cG_i ( B',\eN',\zee';V')\big|_{x=\alpha} =0 \,,
\ee
as desired.

\item When there exists a node $j$ with $b_{jk}<0$ such that $Q_j (\alpha)=0$, we find that
\be
\cG_k (B,\eN,\zee;V)\big|_{x=\alpha} =0 \;.
\ee
Since we have assumed that $V$ is a strong solution, $i$ and $j$ are not adjacent. Meanwhile, the vacuum equation at node $k$ implies that
\be
C_k (z_k) \, Q_k (\alpha) \, Q_k'(\alpha) =\cG_k (B,\eN,\zee;V)\big|_{x=\alpha} =0 \;.
\ee
Since we have assumed that $V$ is a strong solution, $Q_k (\alpha) \neq 0$: hence, $Q_k'(\alpha) =0$. In the dual theory $\fB$, $b'_{ik}=-b_{ik}<0$, while $b'_{ij} = b_{ij}+\sign (b_{ik}) [b_{ik} b_{kj}]_+ = b_{ik} b_{kj} >0$. Therefore, it follows that
\be
\cG_i ( B',\eN',\zee';V') \big|_{x=\alpha} = Q_k'(\alpha)^{b'_{ki}}\prod_{l \neq k} Q_l (\alpha)^{[b_{li}']_+} + Q_j(\alpha)^{-b'_{ji}}\prod_{l \neq j} Q_l (\alpha)^{[-b_{li}']_+}  = 0 \;,
\ee
as desired.
\end{enumerate}
An analogous procedure shows that $V'$ solves the equations at nodes $i$ with $b_{ik} <0$.

Some comments are in order. First, in case 2) presented above, a strong solution $V$ of $\fA$ is mapped to a weak solution $V'$ of $\fB$. Second, the assumption of a solution being strong plays a crucial role in the proof. In fact, it is possible to construct examples where the map \eq{Q map} of a weak solution does not produce a solution of $\fB$. A better understanding of weak solutions, hopefully, would lead to a proof of the equivalence of twisted chiral rings of cluster dual theories.%
\footnote{Let us give an example of what might happen. Consider $U(N)$ SQCD with $N_f = N+1$, $N_a = 1$ and tune $\tilde s_1 = s_1$. The equation in theory $\fA$ is $(x-s_1) \big( \prod_{F=2}^N (x-s_F) - i^N z \big) = Q(x) Q'(x)$, with $\deg Q = N$, $\deg Q'=1$. There is a strong solution where $Q(x)$ is not divisible by $(x-s_1)$ and $N$ weak solutions where it is. In fact, the weak solutions are not Coulomb branch vacua, but nevertheless represent vacua: they include a non-compact Higgs branch direction, and are non-normalizable. In the dual theory, the strong solution is mapped to a weak solution $Q'(x) = (x-s_1)$: this is, in fact, a massive vacuum on the Higgs branch, thanks to the superpotential $W_\text{dual}$. The weak solutions of $\fA$ are mapped to solutions of $\fB$ which are strong for what concerns the gauge node; there is, however, a single gauge singlet in each of these vacua that represents a non-compact flat direction. We would like to thank Kentaro Hori for helping us understand this example.}

\subsection{Superpotentials and single-trace chiral operators}
\label{ss:chiralring}

The map of superpotential terms under the dualities of quiver gauge theories is a straightforward generalization of the one for SQCD-like theories, presented in section \ref{sec: chiral ring}. Let us consider taking the cluster dual of a quiver theory at the $k$-th gauge node. We denote the bifundamental fields connecting two nodes $i \ra j$ as $q^{ij}_{\a}$ with $\a \in \big[ [b_{ij}]_+ \big]$. Let us assume that the superpotential of theory $\fA$, \ie{}, the theory before dualization, is given by
\be
W = W_0 \big( q^{ik}_\a q^{kj}_\b ,\, X_k \big) + W_1 (X_k ) \;,
\ee
where we have made use of the fact that $q^{ik}_{\a}$ and $q^{kj}_{\b}$ can only appear in the $U(N_k)$-invariant combinations $q^{ik}_\a q^{kj}_\b$.  We have denoted the set of all other fields, uncharged under node $k$, as $X_k$.

In theory $\fB$, the combinations of fields $q^{ik}_\alpha q^{kj}_\beta$ are substituted with the bifundamentals $M^{ij}_{\alpha\beta}$, with $(\alpha,\beta) \in \big[ [b_{ik}]_+ \big] \times \big[ [b_{kj}]_+ \big]$, which are singlets of $U(N_k)$. Also, the dual quarks $q'^{ki}_{\a}$ and $q'^{jk}_{\b}$ are added to the theory. The superpotential is then given by
\be
\widehat W' = \sum_{\a, \b}  \Tr \big( M^{ij}_{\a\b} \: q'^{jk}_\b \: q'^{ki}_\a \big) + W_0 \big( M^{ij}_{\a\b},\, X_k \big) + W_1 (X_k) \;.
\ee
We have put a hat on this superpotential, obtained by the prescription in (\ref{suppot complete}), to emphasize that we need to take a further step to complete the duality in the case of non-degenerate quiver theories: we need to integrate out bifundamentals that have a quadratic coupling in the superpotential---this is step 3) described at the end of section \ref{ss:CA}. Whenever a term $M^{ij}_\gamma q^{ji}_{\gamma'} \,\in\, W_0 \big( M^{ij}_{\a\b},\, X_k \big)$ is present in the superpotential, we solve the constraints
\be
\frac{ \p \widehat W' }{\p M^{ij}_\gamma} = \frac{\p \widehat W'}{\p  q^{ji}_{\gamma'}} = 0 \;,
\ee
and replace $M^{ij}_\gamma$ and $q^{ji}_{\gamma'}$ in $\widehat W'$ accordingly to obtain the superpotential $W'$ of theory $\fB$. This procedure is well-defined for superpotentials with single-trace as well as multi-trace terms. Note, however, that this map preserves the single-tracedness of the superpotential, \ie{}, if we start with a single-trace superpotential $W$, the procedure generates a single-trace superpotential $W'$ of the dual theory.

The transformation rule for the superpotential is consistent with the sphere partition function, paralleling the discussion of section \ref{sec: partition function}. The partition function does not depend on the chiral ring deformations---\ie{}, on the coefficients in the superpotential---but it does depend on the R-charges of chiral multiplets. Denoting by $r^{ij}_\alpha$ the R-charge of $q^{ij}_\a$, \eq{UNUNprimewithR} implies that the R-charges of $q'^{jk}_\beta$ and $q'^{ki}_\alpha$ are $(1-r^{kj}_\beta)$ and $(1-r^{ik}_\a)$ respectively, while $M^{ij}_{\alpha\beta}$ have R-charges $r^{ik}_\alpha + r^{kj}_\beta$ (up to mixing with the gauge charges). This is consistent with identifying the fields $M^{ij}_{\alpha\beta}$ with the mesonic operators $q^{ik}_\alpha q^{kj}_\beta$, as well as with the presence of the superpotential terms $\Tr ( M^{ij}_{\a\b} q'^{jk}_\b q'^{ki}_\a )$ which have R-charge 2. Moreover, we have seen that the procedure of integrating out pairs of massive chiral fields with quadratic superpotential terms is consistent with the partition function, as the one-loop determinants of such fields in the Coulomb branch integral cancel out up to a shift of theta angles.

We conclude by summarizing some results of \cite{Derksen1} with regard to the mapping of the single-trace chiral ring under dualities. This has been discussed in the mathematics literature, from the point of view of the path algebra and the Jacobian algebra of quiver diagrams with potentials, in \cite{Derksen1, Derksen2}. A result of \cite{Derksen1} is that the ``deformation spaces" of two quivers with potentials are isomorphic under cluster mutations (Proposition 6.9). The cluster mutation rules of quivers and potentials are given precisely by the dualities we have described (Definition 5.5). The deformation space (Definition 5.7) is defined as the space of directed loops in the quiver modulo constraints obtained from the potential, endowed with some algebraic structure. This is precisely the single-trace chiral ring: it consists of all single-trace chiral operators $\Tr \big( q^{i_1 i_2}_{\a_1} ~\!q^{i_2 i_3}_{\a_2}~\! \cdots ~\!q^{i_n i_1}_{\a_n} \big)$, which can be thought of as directed loops in the quiver, subject to the set of constraints $\p W / \p q^{ij}_\a =0$. The result of \cite{Derksen1} is restrictive in the sense that only single-trace operators---which are neutral under flavor charges---are allowed to be present in the potential. It would be interesting to see if their result can be extended to more general potentials and operators.

\subsection{Non-degenerate quiver theories}
\label{sec: nondegenerateFTs}

Let us come back to the question whether a given quiver gauge theory with a superpotential is such that all 2-cycles that are generated for any arbitrary chain of Seiberg-like dualities are cancelled---in other words, whether every dual theory obtained by successive applications of Seiberg-like dualities has enough quadratic terms in the superpotential that, after integrating out the chiral multiplets involved, no oriented 2-cycles are left. It is relatively easy to write down the condition for which this happens after a single Seiberg-like duality, as it turns out to be a condition on the cubic terms in the superpotential. It is, however, quite nontrivial to make sure that this happens indefinitely, \ie{}, after an arbitrary sequence of cluster dualities.

Following the mathematics literature, we call a quiver gauge theory with a superpotential \emph{non-degenerate} if it has the property described above, and we call an R-symmetry preserving superpotential that does the job a \emph{non-degenerate graded potential}. We also say that a quiver diagram is non-degenerate when the corresponding quiver gauge theory admits a non-degenerate graded potential.

Non-degeneracy is a crucial property, if we wish to make contact between Seiberg-like dualities of 2d quiver gauge theories and cluster algebras. This is because, if the multiplets forming a 2-cycle cannot be integrated out in some duality frame, cluster duality is not applicable to that particular dual theory, as the quiver diagram is not of cluster type. If we proceed to perform a Seiberg-like duality on this theory with respect to one of the two nodes along the 2-cycle, the quiver obtained develops a 1-cycle on the other node: an adjoint chiral multiplet. We cannot apply a further Seiberg-like duality on the node with a 1-cycle, since such a duality is not defined in this case (although a duality of a different nature might exist). Thus the mathematics of cluster algebras have limited applicability for degenerate theories.

It is possible to construct examples of degenerate quivers, \ie{}, quivers that do not admit any non-degenerate graded potentials (see \eg{}, remark 4.41 in \cite{2013arXiv1307.3379D}), and therefore non-degeneracy is a nontrivial condition.%
\footnote{While there is no obstruction to applying mutations indefinitely to a degenerate quiver from the point of view of cluster algebra, the cluster algebra structure obtained in this way cannot be reproduced by a 2d gauge theory.}
Unfortunately, it is an open problem to classify non-degenerate quivers and non-degenerate theories. There are, however, many classes of quivers which are known to admit non-degenerate superpotentials, a list of which can be found in \cite{2013arXiv1307.3379D}. For instance, all acyclic quivers with vanishing superpotential are non-degenerate; then, all quivers that can be made acyclic by mutations are non-degenerate, and their non-degenerate superpotential can be constructed via cluster dualities.

An intriguing entry in the list is the class of quivers that are dual to ideal triangulations of surfaces with marked points \cite{Labardini, Geiss}. Such quivers and their mutations have frequently appeared in the study of class $\mathcal{S}$ theories \cite{Gaiotto:2009hg, Alim:2011ae, Alim:2011kw}, which are a class of four-dimensional $\cN=2$ theories obtained by wrapping M5-branes on punctured Riemann surfaces \cite{Gaiotto:2009we} (see also \cite{Benini:2009gi, Gaiotto:2009hg}). It would be interesting to understand whether there is a connection between the two-dimensional gauge theories associated to these quivers and the corresponding four-dimensional theories.

\section{Some geometric implications}
\label{s:applications}

The dualities of two-dimensional $\cN=(2,2)$ quiver gauge theories, possibly with superpotentials, have some interesting geometric implications. This is because the gauge theory can flow in the IR to a non-linear sigma model (NLSM) on a subvariety of a K\"ahler quotient manifold. The complexified FI terms, or equivalently the K\"ahler coordinates $z_i$, are coordinates on the ``extended K\"ahler moduli space" of the manifold. In the case when the IR theory is conformal, the moduli space extends to include those of manifolds related by flops or other transitions, which are different large volume limits of the same conformal field theory. The cluster algebra mutations give rise to different gauged linear sigma model (GLSM) descriptions of the same geometry, as well as to special birational changes of coordinates on the K\"ahler moduli space.

As hinted above, the theories that flow to a fixed point in the IR are particularly interesting: in this case, the IR geometry is a (possibly non-compact) CY manifold, and the K\"ahler coordinates parametrize its K\"ahler moduli space. The sphere partition function $Z_{S^2}$ can be used to compute the quantum corrected metric on this moduli space \cite{Jockers:2012dk, Gerchkovitz:2014gta}---in fact, it produces the K\"ahler potential of the Zamolodchikov metric in the CFT. Our result implies that this metric is the same in all different GLSM descriptions as long as we transform the coordinates according to cluster algebra mutation rules.%
\footnote{In the non-conformal case, the mutation of a node with $|N_f - N_a|=1$ leads to a correction of the K\"ahler potential that should be taken into account.}

\begin{figure}[t]
\centering\includegraphics[width=12cm]{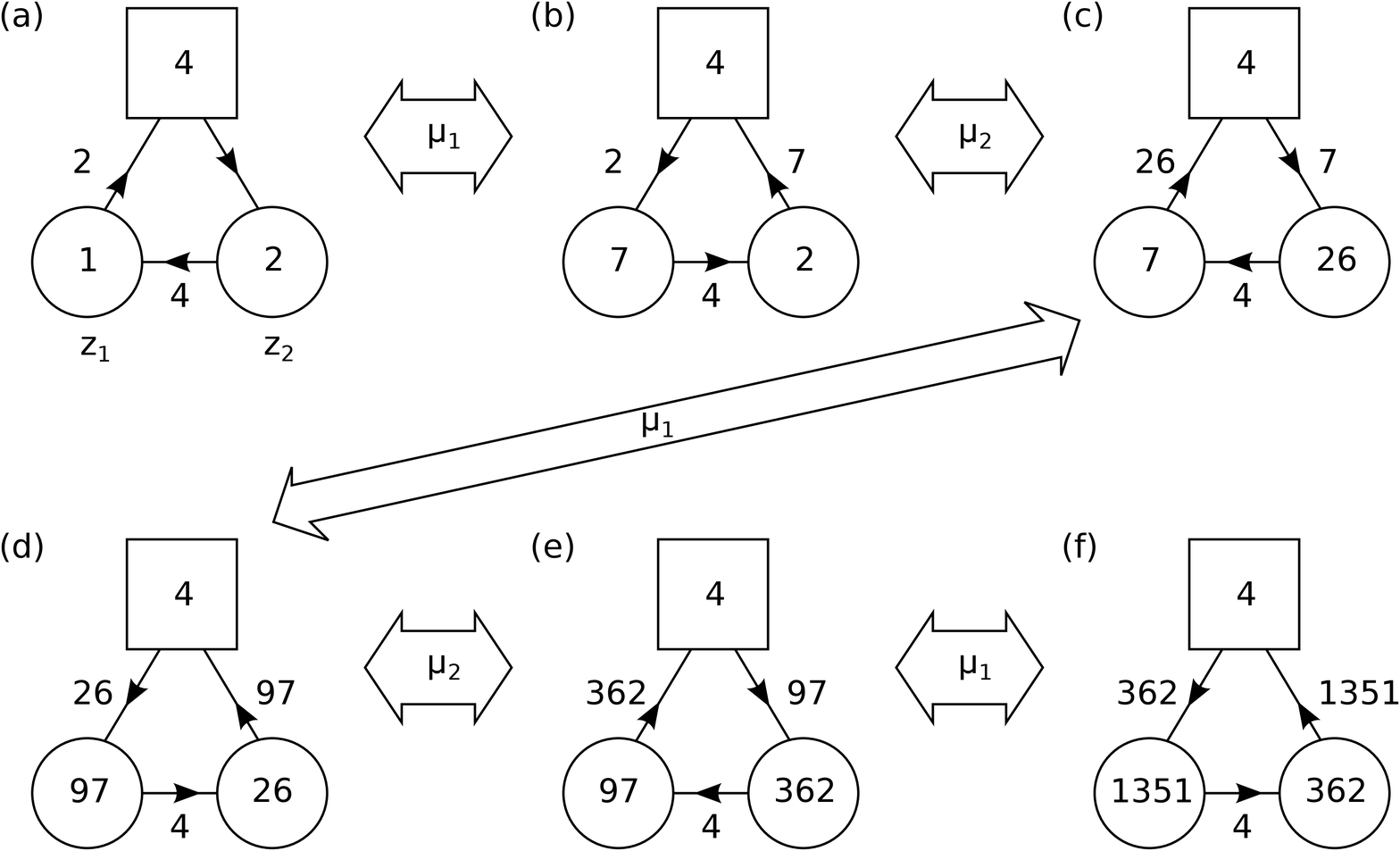}
\caption{Cluster mutations on the quiver diagram for the Gulliksen-Neg\r{a}rd Calabi-Yau threefold. The numbers next to the arrows denote the multiplicity. Figure (a) is the GLSM model proposed in \cite{Jockers:2012zr}, while the other ones are obtained via mutation on the left ($\mu_1$) or the right ($\mu_2$) node.
\label{f:GN}}
\end{figure}

The cluster algebra structure of the K\"ahler moduli space implies many interesting properties. For instance, it implies that there is a natural (possibly degenerate) Poisson structure \cite{gekhtman2003cluster} on the moduli space, defined by
\be
\{z_i, z_j \} = z_i z_j b_{ij} \;,
\ee
which can be extended to Laurent polynomials by the Leibniz rule and to analytic functions by continuity. It follows that
\be
\{f,g\} = \sum_{k,l} \frac{\partial f}{\partial z_k} \frac{\partial g}{\partial z_l} z_k z_l b_{kl} \;,\qquad\qquad \{\log z_i, \log z_j \} = b_{ij} \;.
\ee
This equation implies that $\log z_i$ can be interpreted as a set of ``canonical coordinates," although $b_{ij}$ can be degenerate. This Poisson structure is compatible with the cluster algebra structure, in the sense that the Poisson bracket is invariant under mutation of $b_{ij}$ and $z_i$ (\ref{cluster algebra transformations})-(\ref{zrule}): $\{ z_i'(z), z_j'(z) \}_{b_{kl}} = z_i' z_j' b_{ij}'$. Therefore, mutations can be thought of as canonical transformations.

There are other properties of the K\"ahler moduli space following from cluster algebra, which we do not know how to take full advantage of. One is total positivity: there is a canonical way of restricting the coordinates $z_i$ to $\bR_{>0}$, in the sense that performing the restriction in any frame would yield the same submanifold of the moduli space, as is apparent from (\ref{zmapconformal}). From the field theory point of view, this means that there is a canonical choice of theta angles.

One could be worried that the CY geometries that arise from quiver gauge theories with R-symmetric superpotentials might be uninteresting. For example, one can wonder whether there are compact Calabi-Yau manifolds in the list. The answer is yes: an illustrative example is the Gulliksen-Neg\r{a}rd Calabi-Yau threefold \cite{Gulliksen}, which is studied from the GLSM point of view in \cite{Jockers:2012zr,Jockers:2012dk}. This CY threefold can be engineered by a $U(1) \times U(2)$ gauge theory whose matter content consists of eight fundamentals of $U(1)$, four fundamentals of $U(2)$ and four bifundamentals, coupled by a cubic superpotential. The Hodge numbers $(h^{1,1}, h^{2,1})= (2,34)$ imply that there are two K\"ahler parameters $z_{1,2}$ identified with the K\"ahler coordinates in the GLSM. The quiver is depicted in figure \ref{f:GN} (a). We thus learn that the K\"ahler moduli space of this CY has a cluster algebra structure. In fact, in this example, we obtain infinitely many distinct GLSM dual descriptions by successive applications of cluster dualities. We have depicted some of these GLSMs in figure \ref{f:GN}, and listed their K\"ahler coordinates as functions of the original ones in table \ref{t:GN}: after a few mutations, they soon become quite intricate.

\begin{table}[t]
\center
\begin{tabular}{ | c | c | c | c |}
\hline
  & mutations & \multicolumn{2}{|c|}{K\"ahler Coordinates} \\ \hline\hline
  (a) & $\cdot$ & $z_1$ & $z_2$  \\ \hline
  (b) & $\mu_1$ & $z_1^{-1}$ & $ z_2 (1+z_1)^4$  \\ \hline
  (c) & $\mu_2 \mu_1$ & $ z_1^{-1} \big( 1+z_2(1+z_1)^4 \big)^{4 \rule{0pt}{.6em}} $ & $z_2^{-1}(1+z_1)^{-4}$ \\ \hline
  (d) & $\mu_1 \mu_2 \mu_1$ & $ z_1 \big( 1+z_2(1+z_1)^4 \big)^{-4 \rule{0pt}{.6em}}$ &
  $ z_2^{-1}(1+z_1)^{-4} \Big( 1 + z_1^{-1} \big( 1+z_2 (1+z_1)^4 \big)^4 \Big)^{4 \rule{0pt}{.6em}}$ \\ \hline
\end{tabular}
\caption{Map of K\"ahler coordinates of the Gulliksen-Neg\r{a}rd CY threefold under cluster dualities.
\label{t:GN}}
\end{table}

The cluster algebra structure has also interesting implications for the topology of the moduli space of certain Calabi-Yau manifolds. A cluster algebra is said to be of \emph{finite mutation type} if there is a finite number of exchange matrices in all seeds: only a finite number of quiver diagrams appear upon applying arbitrary sequences of mutations. Then mutations define an automorphism of $\bC^m$ (where $m$ is the number of gauge nodes), since certain sequences of mutations map a quiver diagram to itself, with a non-trivial transformation of the cluster variables and K\"ahler coordinates $z_i$: then the actual moduli space of K\"ahler parameters of the field theory is the quotient of $\bC^m$ by the automorphism group, and the group defines a tessellation of $\bC^m$. For example, quiver diagrams constructed from ideal triangulations of Riemann surfaces with punctures are of this nature. It would be interesting to see if this structure has any physical consequence, for the phase structure and singular points of the moduli space, and so on.

Finally, let us suggest an application of quantum field theory to the theory of cluster algebras. As described in section \ref{sec: dualities}, if a gauge node in a quiver has more colors than flavors, \ie{}, if
$$
\max\big( N_f(i), N_a(i) \big) < N_i
$$
for some gauge node $i$, then the theory breaks supersymmetry. An attempt to apply the cluster mutation $\mu_i$ would generate a node with negative rank. Then, given a quiver with ranks which can be represented by elements $\mathsf{n}_i \in \bP$ transforming as in (\ref{ranks mutation}) under mutations, it is an interesting combinatorial problem to determine whether the supersymmetry of the theory is broken in this way, \ie{}, whether there is a duality frame where some node has negative rank. ``Supersymmetry breaking" as defined above is a mutation-invariant property, that, to our knowledge, has not been defined or utilized in the mathematical literature. We note that a necessary condition for the defined supersymmetry breaking is that
$$
Z_{S^2} = 0 \;.
$$
This is a mutation invariant condition! We thus see that the sphere partition function provides a useful handle on the mathematical problem of discerning whether a quiver with given ranks can develop any negative rank nodes upon mutations.

\section{$\cN = (2,2)^*$ dualities}
\label{sec: 22star}

We conclude by discussing some dualities, similar to the Seiberg-like (or cluster) dualities we considered so far, of $\cN=(2,2)^*$ quiver gauge theories. In the case of quivers with a single gauge node, these dualities have been analyzed by Nekrasov and Shatashvili \cite{Nekrasov:2009uh, Nekrasov:2009ui} in the context of the ``Gauge/Bethe correspondence," and were called ``Grassmannian dualities"; in fact, the avatars of such dualities are also known in the integrable systems literature \cite{Gromov:2007ky, Bazhanov:2010ts}. Here, we discuss how these dualities are extended to general quivers, as well as the map of parameters and the equality of $S^2$ partition functions under the dualities.%
\footnote{We mention that some dualities between certain quivers have been considered in \cite{Orlando:2010uu}. They can be understood as examples of the Seiberg-like dualities of section \ref{sec: dualities}.}
We stress that the $\cN=(2,2)^*$ dualities are \emph{not} special cases of the Seiberg-like dualities of section \ref{sec: dualities}, although geometrically they are tightly related. In particular, the latter can be embedded into the former in some limit, as we explain later on. A more thorough analysis of this relation will appear elsewhere.

\begin{figure}[t]
\centering\includegraphics[width=10cm]{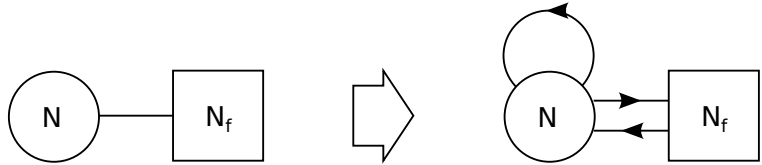}
\caption{Left: diagram representing an $\cN=(4,4)$ quiver gauge theory. Notice that lines are unoriented since they represent hypermultiplets. Right: quiver diagram of the same theory, thought of as an $\cN=(2,2)$ theory. Each gauge node has an adjoint chiral field, while bifundamental hypermultiplets break into pairs of chiral multiplets.
\label{fig: 44to22star}}
\end{figure}

The $\cN=(2,2)^*$ theories can be obtained from $\cN=(4,4)$ gauge theories by softly breaking the supersymmetry by turning on a twisted mass for a particular R-symmetry of the theory. So let us first present some relevant facts about quiver gauge theories with $\cN=(4,4)$ supersymmetry and gauge group $U(N_1) \times \cdots \times U(N_m)$. The matter content is organized into bifundamental hypermultiplets of the gauge and flavor factors, and can be represented by a diagram with gauge (circles) and flavor (squares) nodes connected by unoriented lines, as in figure \ref{fig: 44to22star} on the left. Letting $n$ be the total number of gauge and flavor nodes, we can introduce an $n\times n$ symmetric matrix $C = c_{ij}$ whose non-negative entries are the number of bifundamental hypermultiplets between node $i$ and $j$. We set the diagonal entries of $C$ to vanish. The $\cN=(4,4)$ vector multiplet decomposes into an $\cN=(2,2)$ vector multiplet and an adjoint chiral multiplet, which we denote $\Phi_i$, for each node $i$; a hypermultiplet charged under the nodes $i$ and $j$ decomposes into two chiral multiplets $q_{ij}$ and $q_{ji}$ in conjugate representations. The corresponding quiver diagram is depicted in figure \ref{fig: 44to22star} on the right. For each gauge node $i$, the terms $\sum_\a \Tr \big( q^\a_{ji}~\!\Phi_i~\! q^\a_{ij} \big)$ are added to the superpotential, where $\a$ parametrizes all the hypermultiplets charged under $U(N_i)$. Since no other superpotential is allowed by $\cN=(4,4)$ supersymmetry (with the exception of complex mass terms), the diagram completely specifies the theory.

These theories have a global $U(1)_\Phi$ symmetry, part of the R-symmetry, under which the adjoint scalars $\Phi_i$ have charge $-1$ while the bifundamental chiral multiplets have charge $\frac12$. The $\cN=(2,2)^*$ theories are obtained by turning on a complex twisted mass $s_\Phi$ associated to this symmetry, thereby breaking $(4,4)$ to $(2,2)$ supersymmetry. We can introduce the parameters
\be
v_+ = i s_\Phi \;, \qquad\qquad v_- = i\bar s_\Phi \;,
\ee
which are useful in expressing the sphere partition function and the twisted chiral ring relations.

The $\cN=(2,2)^*$ theories enjoy dualities defined with respect to a gauge node $k$. They leave the (quiver) diagram and the $U(1)_\Phi$ charges invariant, while modifying the gauge group ranks simply as
\be
\label{22starrankmap}
N_i'= \begin{cases} \sum_j c_{kj} N_j  - N_k \qquad &\text{if $i=k$} \\ N_i &\text{otherwise} \;.
\end{cases}
\ee
The exponentiated FI terms, that we denote by
\be
\zeta_i = e^{-t_i} \;,
\ee
to avoid confusion with the parameters $z_i$ (which include suitable signs), are mapped according to the rules:
\be
\label{22starFImap}
\zeta_i'= \begin{cases} \zeta_k^{-1} &\text{if $i=k$} \\ \zeta_i \prod_j \zeta_j^{c_{ij}} \quad &\text{otherwise} \;. \end{cases}
\ee
In the case of a single gauge node, the basic $\cN=(2,2)^*$ duality is depicted in figure \ref{fig: 22starbasic}. Notice that these dualities do not have a direct interpretation in cluster algebra.

\begin{figure}[t]
\centering\includegraphics[width=10cm]{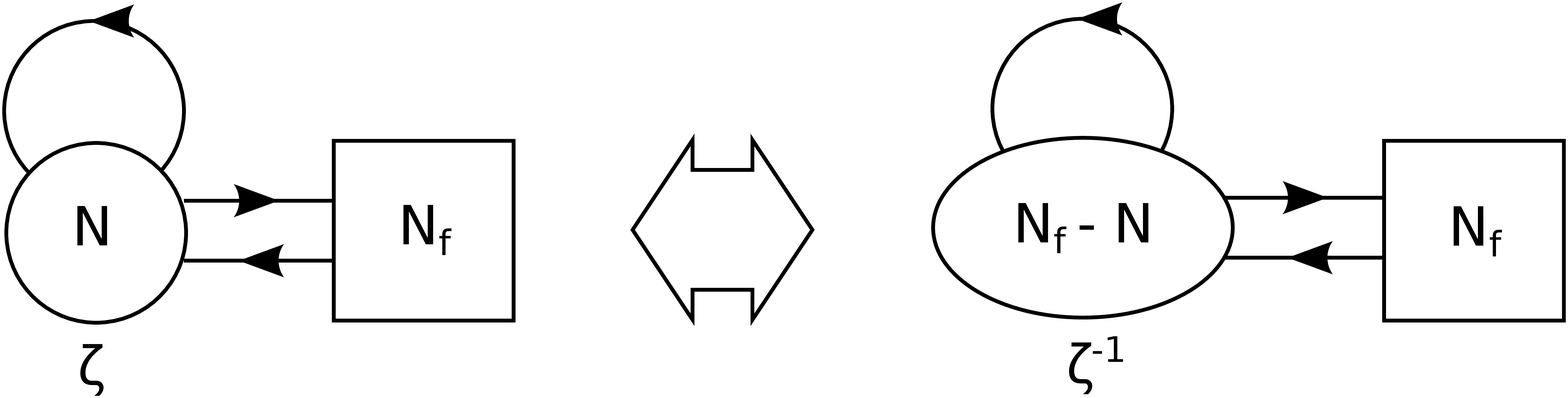}
\caption{The basic $\cN=(2,2)^*$ duality between a pair of single-node gauge theories with $N_f$ flavors. The K\"ahler coordinates $\zeta_i$ are indicated.
\label{fig: 22starbasic}}
\end{figure}

In the case of the single-node theory in figure \ref{fig: 44to22star} and \ref{fig: 22starbasic}, \ie{} $\cN=(4,4)$ or $\cN=(2,2)^*$ SQCD, the duality has a simple geometric interpretation. At low energies the GLSM flows%
\footnote{More precisely, $\cN=(4,4)$ SQCD leads to two different CFTs \cite{Witten:1997yu}, one on the Higgs branch which is the hyper-K\"ahler NLSM we are talking about here, and one on the Coulomb branch.}
to the NLSM of $T^* \text{Gr}(N,N_f)$, the total space of the cotangent bundle to the Grassmannian, possibly deformed by equivariant parameters for the global symmetry $U(1)_\Phi \times SU(N_f)$. The duality simply expresses the identity of spaces $T^* \text{Gr}(N,N_f) = T^* \text{Gr}(N_f - N, N_f)$. The $\cN=(2,2)^*$ dualities also have an interpretation in terms of integrable systems \cite{Nekrasov:2009uh, Nekrasov:2009ui}, and the map of twisted chiral ring operators can be obtained from there \cite{Bazhanov:2010ts}.

Let us consider the duality for $\cN=(2,2)^*$ SQCD as in figure \ref{fig: 22starbasic}. We denote the $N_f$ complexified twisted mass parameters associated to the flavors as $s_F$ and then introduce $\S_{F\pm}$ as in (\ref{def Sigma+-}), possibly shifted by $\Sigma_{F\pm} \to \Sigma_{F\pm} + \frac{r_F}2$ according to the R-charges $r_F$ of the fundamentals. We also redefine
\be
v_+ = i \re s_\Phi + \frac{\fm_\Phi}2 + \frac{r_\Phi}2 - 1 \;,\qquad\qquad v_- = i \re s_\Phi - \frac{\fm_\Phi}2 + \frac{r_\Phi}2 - 1
\ee
where $\fm_\Phi$ is the quantized magnetic flux associated to $U(1)_\Phi$ equal to the imaginary part of the twisted mass, and $r_\Phi$ is the R-charge of $\Phi$. Notice that, with respect to the notation of section \ref{sec: partition function}, $\Sigma_{\Phi\pm} \equiv 1+v_\pm$. When the theory flows to a fixed point, $v_\pm = 0$ corresponds to the Higgs branch CFT while $v_\pm = - 1$ corresponds to the Coulomb branch CFT. The sphere partition function is computed by the integral:
\begin{multline}
Z^{(4,4)}_{U(N)} (\S_{F\pm},v_\pm; \zeta) = \frac1{N!} \sum_{\m_I  \in \mathbb{Z}^N} \int \prod_{I=1}^N \frac{d\s_I}{2\pi} \, \big( e^{i\pi(N-1)} \zeta \big)^{\s_{I+}} \big( e^{-i\pi(N-1)} \bar \zeta \big)^{\s_{I-}} \;\cdot\; \prod_{I<J} \big( {-\S^I_{J+} \S^I_{J-}} \big) \\
\times \prod_{I,J}  \frac{\Ga \big( 1-\S^I_{J+} + v_+ \big)}{\Ga \big( \S^I_{J-} - v_- \big)} \;\cdot\; \prod_{I, F} \frac{\Ga \big( {-\S^I_{F+} - \frac{v_+}2} \big) ~ \Ga \big( \S^I_{F+} - \frac{v_+}2 \big) }{ \Ga \big( 1+\S^I_{F-} + \frac{v_-}2 \big) ~ \Ga \big( 1-\S^I_{F-} + \frac{v_-}2 \big) } \;.
\end{multline}
We have assigned appropriate R-charges to the fields compatible with the superpotential; different R-charge assignments simply correspond to improvement transformations. The integral is invariant under the charge-conjugation map:
\be
Z^{(4,4)}_{U(N)} (\S_{F\pm}, v_\pm; \zeta) = Z^{(4,4)}_{U(N)} ( - \S_{F\pm}, v_\pm; \zeta^{-1}) \;.
\ee
The integral can be written in terms of vortex partition functions:
\be
\label{22stardecomp}
Z^{(4,4)}_{U(N)} (\S_{F\pm},v_\pm;\zeta) = \sum_{\vF \in C(N,N_f)} \cZ_0^\vF \, \cZ_+^\vF \, \cZ_-^\vF \;.
\ee
Just as in the $\cN=(2,2)$ case, the vector $\vF \in C(N,N_f)$ labels vortex sectors; $\cZ^\vF_0$ encodes the contribution of the classical and one-loop piece to the partition function in a given sector:
\be
\cZ^\vF_0 = (-1)^{\frac{N(N+1)}2 \fm_\Phi} \, \zeta^{\sum\limits_I \S_{F_I+} - \frac N2 v_+} \; \bar \zeta^{\sum\limits_I \S_{F_I-} - \frac N2 v_-} \, \prod_{I,\Id} \frac{ \Ga \big( {-\S^{F_I}_{F^c_\Id+} } \big) ~ \Ga \big( \S^{F_I}_{F^c_\Id+} - v_+ \big) }{ \Ga\big( 1+\S^{F_I}_{F^c_\Id-} \big) ~ \Ga\big( 1-\S^{F_I}_{F^c_\Id-} + v_- \big) } \;.
\label{Z022star}
\ee
Then $\cZ^\vF_\pm$ are the vortex partition functions:
\be
\cZ^\vF_+ = \cZ^\vF \big( \S_{F+},v_+;\, \zeta \big) \;, \qquad\qquad \cZ^\vF_-= \cZ^\vF \big( \S_{F-},v_-;\, \bar \zeta \big) \;,
\ee
with \cite{Benini:2012ui}
\be
\label{22starVPF}
\cZ^\vF \big( \S_F,v;\zeta \big) = \sum_{n\geq 0} \zeta^n \sum_{\sabs{(n_I)}=n} \prod_I \frac{ \prod_J \big( \S^{F_I}_{F_J} - v + n_I - n_J \big)_{n_J} \, \prod_\Jd \big( \S^{F_I}_{F^c_\Jd} - v \big)_{n_I} }{ \prod_J \big( {-\S^{F_I}_{F_J}-n_I} \big)_{n_J} \, \prod_\Jd \big( {-\S^{F_I}_{F^c_\Jd}-n_I} \big)_{n_I} } \;,
\ee
where $(n_I)$ are $N$-tuples of non-negative integers.

Notice the special cases:
\be
\cZ^\vF(\Sigma_F,0;\zeta) = 1 \;,\qquad\qquad \cZ^\vF(\Sigma_F,-1;\zeta) = \Big( 1 - (-1)^{N_f} \zeta \Big)^{-N} \;.
\ee
For $v_\pm = 0$ (Higgs branch CFT) the vortex partition function is trivial, as there are no instanton corrections, while
\be
\cZ_0^\vF(\Sigma_{F\pm}, 0; \zeta) = \zeta^{\sum\limits_I \S_{F_I+}} \, \bar \zeta^{\sum\limits_I \S_{F_I-}} \prod_{I,I'} \frac{(-1)^{\fm_{F_I} - \fm_{F_{I'}}} }{ \big( {-\Sigma^{F_I}_{F^c_{I'}+} \Sigma^{F_I}_{F^c_{I'}-} } \big)} \;.
\ee
For $v_\pm = -1$ we find that
\be
\label{zvf when z=-1}
\cZ_0^\vF(\Sigma_{F\pm}, -1; \zeta) = \zeta^{\sum\limits_I \S_{F_I+} + \frac N2} \, \bar \zeta^{\sum\limits_I \S_{F_I-} + \frac N2} \prod_{I,I'} (-1)^{\fm_{F_I} - \fm_{F_{I'}}} \;,
\ee
while the vortex partition function is independent of the masses.%
\footnote{Notice that the original integral expression for the sphere partition function is ill-defined in this case, and for $U(1)$ formally gives $\Gamma(0) \delta\big( \log(-1)^{N_f}\zeta \big) \delta\big( \log(-1)^{N_f}\bar\zeta \big)$.}

It turns out that the $\cN=(2,2)^*$ vortex partition function satisfies the following identity:
\be
\label{22star VPF identity}
\cZ^\vF (\S_F,v;\zeta) = \Big( 1-(-1)^{N_f} \zeta \Big)^{(2N-N_f)v} \; \cZ^\vFc (-\S_F,v;\zeta) \;,
\ee
motivated below.%
\footnote{This identity has also been noticed in \cite{Honda:2013uca}.}
The one-loop piece satisfies
\begin{multline}
\cZ_0^\vF \big( \Sigma_{F\pm}, v_\pm;\zeta\big) = (-1)^{\frac{(2N-N_f)(N_f+1)}2 \fm_\Phi} \; \zeta^{\sum_F \Sigma_{F+} + \frac12(N_f-2N)v_+} \; \bar\zeta^{\sum_F \Sigma_{F-} + \frac12 (N_f - 2N)v_-} \\
\times \cZ_0^{\vF^c} \big( { -\Sigma_{F\pm}, v_\pm ; \zeta} \big) \;,
\end{multline}
which follows from the definition \eq{Z022star}. Exploiting the charge conjugation map, we arrive at the following identity of partition functions:
\begin{multline}
\label{22starZdual}
Z^{(4,4)}_{U(N)} (\S_{F\pm},v_\pm;\zeta) = \\
= (-1)^{\frac{(2N-N_f)(N_f+1)}2 \fm_\Phi} \; \Big( 1-(-1)^{N_f} \zeta \Big)^{(2N-N_f)v_+} \; \Big( 1-(-1)^{N_f} \bar\zeta \Big)^{(2N-N_f) v_-} \\
\times \zeta^{\sum\limits_F \S_{F+} + \frac12 (N_f-2N)v_+} \;\; \bar\zeta^{\sum\limits_F \S_{F-} + \frac12(N_f-2N) v_-} \;\; Z^{(4,4)}_{U(N_f-N)} (\S_{F\pm},v_\pm ; \zeta^{-1}) \;.
\end{multline}
This equality confirms the duality between the $U(N)$ and $U(N_f-N)$ theories; moreover, it detects the contact terms responsible for the map of FI parameters of neighboring nodes. Keeping only the factors that have a dependence on the twisted masses associated to the flavor symmetry, we can write
\be
Z^{(4,4)}_{U(N)} \big( \S_{F_\pm} ,v_\pm ;\zeta) \;\simeq\; \zeta^{\sum\limits_F \S_{F+}} \; \bar\zeta^{\sum\limits_F \S_{F-}} \; Z^{(4,4)}_{U(N_f-N)} \big( \S_{F\pm},v_\pm;\zeta^{-1} \big) \;.
\ee
As in the $\cN=(2,2)$ case, this relation can be used to derive the duality maps of general $\cN=(2,2)^*$ quiver theories.

Although we do not have a full analytic proof of the vortex partition function identity \eqref{22star VPF identity}, we have verified it up to high order for a range of values of $N_f$ and $N$. Before we explain how this relation is motivated, we note that it should be possible to prove it by taking a similar strategy as in the $\cN=(2,2)$ case. The coefficient of $\zeta^n$ in (\ref{22starVPF}) can be given an integral representation:
\begin{multline}
\nn
\frac{(-1)^n}{n!} \int_\cC \prod_{\a=1}^n \frac{d \varphi_\a}{2\pi i} \, \frac{ \prod_I \big( {- \varphi_\a + \S_{F_I} - v - 1} \big) \, \prod_\Id \big( \varphi_\a - \S_{F^c_\Id} - v \big) }{ \prod_I \big( \varphi_\a - \S_{F_I} \big) \, \prod_\Id \big( {- \varphi_\a+\S_{F^c_\Id}  - 1} \big) } \\
\times \prod_{\a < \b}^n \frac{ (\varphi_\a -\varphi_\b)^2 }{ (\varphi_\a -\varphi_\b)^2 -1} \;\cdot\; \prod_{\a \leq \b}^n \frac{ (\varphi_\a -\varphi_\b)^2 - v^2 }{ (\varphi_\a -\varphi_\b)^2 - (v+1)^2} \;.
\end{multline}
The contour is taken to encircle the poles of the integrand situated at $(\varphi_\a)$ that satisfy
\be
\{  \varphi_\a \}  = \bigcup_{I=1}^N \, \Big\{ \S_{F_I},\; \S_{F_I}+1,\ldots,\; \S_{F_I}+ n_I-1 \Big\}
\ee
for some non-negative $(n_I)$ such that $\sum_I n_I = n$. A convenient way to construct such a contour is to prescribe $0 < \re \Sigma_{F_I} < 1$ and $1 + \re v < 0$, and then take $\cC$ to be a product of contours winding counterclockwise along the imaginary axis and closed in the right half-plane for each of the $\varphi_\alpha$. Any other range of parameters is obtained by analytic continuation. By ``flipping'' this contour in $\mathbb{C}^n$ one gets a contour that encircles the poles at $(\varphi_\a)$ with
\be
\{  \varphi_\a \}  = \bigcup_{I'=1}^{N_f-N} \, \Big\{ \S_{F^c_\Id} - 1,\; \S_{F^c_\Id}-2,\ldots,\; \S_{F^c_\Id} - n'_\Id \Big\}
\ee
for non-negative $(n'_{I'})$ such that $\sum_{I'} n'_{I'} = n$; the generating function for these coefficients is in fact given by $\cZ^\vFc ( -\S_{F},v;\zeta)$. One, however, also gets contributions from poles at infinity, which should be properly accounted for in order to relate $\cZ^\vF (\S_F,v;\zeta)$ to $\cZ^\vFc (-\S_F,v;\zeta)$.

To get the full relation,%
\footnote{Another piece of information comes from the case $N = N_f=1$, which can be easily resummed: $\cZ(\Sigma_F,v;\zeta) = (1+\zeta)^v$. This suggests the form of the general relation.}
we can exploit the operator map between the twisted chiral rings of the dual theories. The twisted chiral ring of the original theory can be derived in the usual way by evaluating the effective twisted superpotential on the Coulomb branch, leading to the vacuum equations:
\be
\label{22star TCR}
P\big( x - \tfrac{iv}2 \Big) \, Q\big( x + iv\big) + q\, P\big( x + \tfrac{iv}2 \big)\, Q\big( x - iv \big) = (1+q)\, Q(x) \, T(x) \;,
\ee
where
\be
P(x) = \prod_{F=1}^{N_f} (x-s_F) \;,\qquad Q(x) = \det (x-\sigma) \;,\qquad q = (-1)^{N_f}\zeta \;,\qquad \deg T(x) = N_f \;.
\ee
The equation has to be solved for $Q(x)$ and a monic polynomial $T(x)$ of degree $N_f$. It can be proven with elementary methods \cite{Gromov:2007ky} that, given a $Q(x)$ that solves the equation above, there exists a unique monic polynomial $Q'(x)$ of degree $N_f - N$ such that
\be
\label{22star operator map}
P(x) = \frac1{1-q} \Big[ Q\big( x + \tfrac{iv}2 \big)\, Q'\big( x - \tfrac{iv}2 \big) - q\, Q\big( x - \tfrac{iv}2 \big) \, Q'\big( x + \tfrac{iv}2 \big) \Big] \;.
\ee
The converse is also true: given $Q(x)$ and $Q'(x)$ that solve this equation, $Q(x)$ solves (\ref{22star TCR}) with
\be
T(x) = \frac1{1-q^2} \Big[ Q\big( x + iv \big)\, Q'\big( x - iv \big) - q^2 \, Q\big( x -iv\big) \, Q'\big( x + iv \big) \Big] \;.
\ee
Therefore, we see that (\ref{22star TCR}) and (\ref{22star operator map}) are in fact equivalent equations. On the other hand, since $Q'(x)$ solves the twisted chiral ring equation of the dual theory with the same $T(x)$, we can identify $Q'(x) = \det(x-\sigma')$ as the $Q$-polynomial of the dual theory, while $T'(x) = T(x)$ remains invariant. Then, the equation \eqref{22star operator map} gives the operator map between the twisted chiral rings of the dual theories. In particular, the coefficient of $x^{N_f-1}$ yields the relation between the trace operators:
\be
\label{22startrace}
\Tr \s = -\Tr\sigma' + \sum_F {s_{F}} -  \frac{iv}2 (N_f-2N) \frac{1+(-1)^{N_f} \zeta }{ 1-(-1)^{N_f} \zeta} \;.
\ee
Now, we can compute expectation values on the sphere with localization following section \ref{sec: operator map}. For linear operators, we have seen that the twisted chiral ring relations (including the operator map) have analogs on the sphere.%
\footnote{For higher order operators this is not the case: this is because all operators have to be inserted at the same point---the pole of $S^2$---and contact terms with interesting physical meaning appear.}
Therefore, recalling that $\zeta'=\zeta^{-1}$, we infer the following relations:
\bea
\frac{\p \log Z^{(4,4)}_{U(N)} }{\p \log \zeta} &= \sum_F \Sigma_{F+} + \frac{v_+}2 (N_f-2N) \frac{1+(-1)^{N_f} \zeta}{ 1-(-1)^{N_f} \zeta} - \frac{\p \log Z^{(4,4)}_{U(N_f-N)} }{ \p \log \zeta} \\
\frac{\p \log Z^{(4,4)}_{U(N)} }{\p \log \bar\zeta} &= \sum_F \Sigma_{F-} + \frac{v_-}2 (N_f-2N) \frac{1+(-1)^{N_f} \bar\zeta}{ 1-(-1)^{N_f} \bar\zeta} - \frac{\p \log Z^{(4,4)}_{U(N_f-N)} }{ \p \log \bar\zeta} \;.
\eea
By integrating these equations, we can reproduce all the $\zeta$-dependent factors in (\ref{22starZdual}), which suggests the relation (\ref{22star VPF identity}).

\

We conclude this section with an observation relating the $\cN=(2,2)^*$ theories to the $\cN=(2,2)$ theories we have studied for most of the paper. The $\cN=(2,2)^*$ dualities do not give rise to a cluster algebra structure. One can, however, consider a ``decoupling limit" in which an $\cN=(2,2)^*$ quiver theory flows to an $\cN=(2,2)$ quiver theory. The decoupling limit consists of taking the adjoint mass $s_\Phi$ to be very large, and at the same time scaling the other flavor twisted masses in such a way that one of the two chirals in each hypermultiplet remains light. The UV FI parameters should also be scaled appropriately, if we want to keep the IR couplings finite. In the limit, the Coulomb branch vacua of the dual theory split into groups: roughly, some vacua remain around the origin while others ``fly" to infinity. To each well-separated group of vacua, one can associate an effective IR description, which is an $\cN=(2,2)$ quiver gauge theory.

\begin{figure}[!t]
\centering\includegraphics[width=13cm]{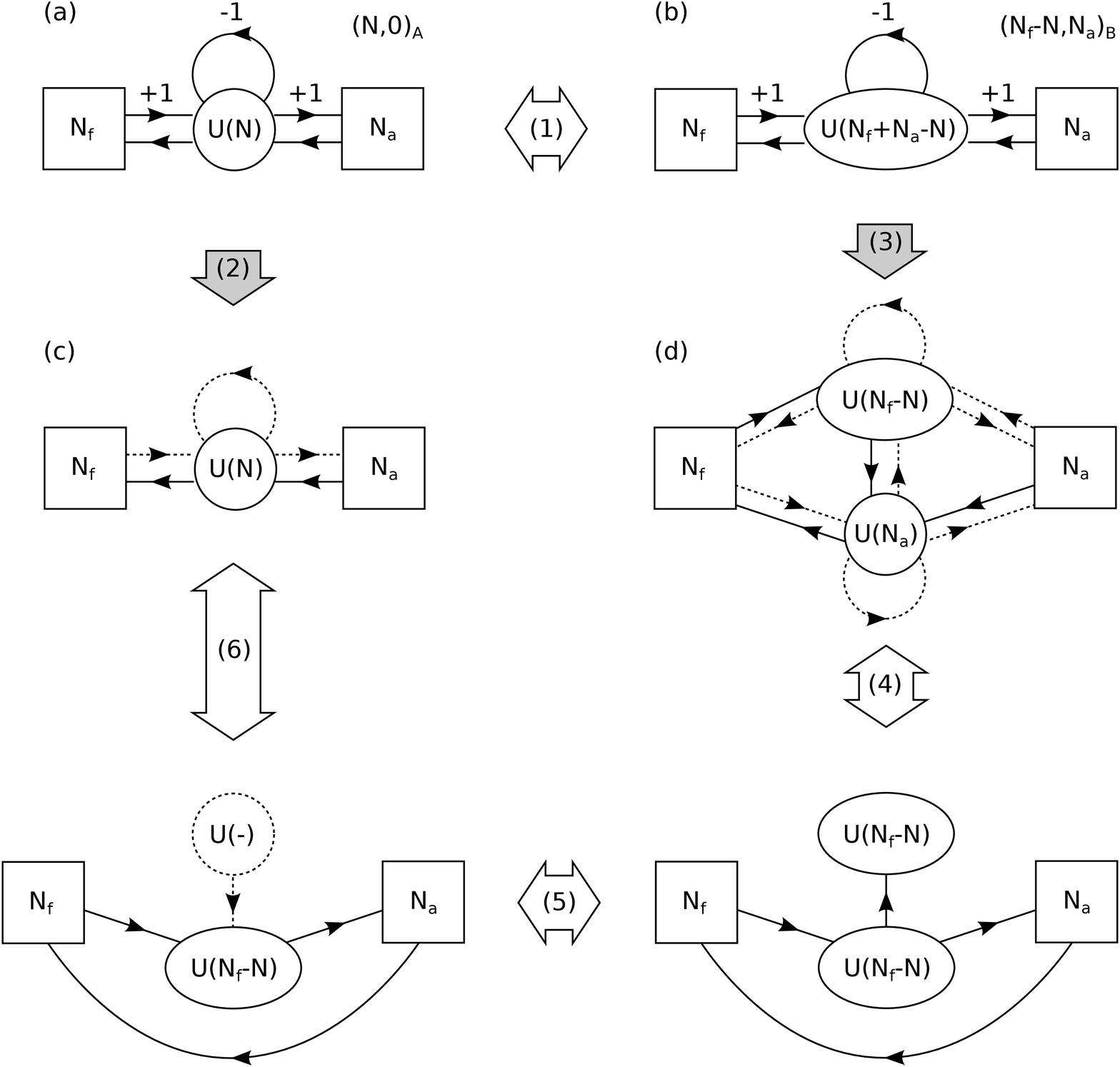}
\caption{\small Various dualities of quiver theories. In (a) we consider the $(N,0)_\fA$ vacuum sector of the $\cN=(2,2)^*$ $U(N)$ theory. This is dual to the $(N_f-N,N_a)_\fB$ sector of the dual $U(N_f+N_a-N)$ theory (b). The effective theory of (a) in the large $s_\text{d}$ limit is the $\cN=(2,2)$ theory (c) with gauge group $U(N)$. The dotted lines denote fields that become very massive, and hence are integrated out. Meanwhile, the effective theory in the decoupling limit of the dual theory is the $U(N_f-N) \times U(N_a)$ theory (d). The theories (c) and (d) are related by a series of cluster dualities (4), (5) and (6).
\label{f:22starweb}}
\end{figure}

As a simple example, we can consider an $\cN=(2,2)^*$ $U(N)$ theory with $N_f+N_a$ flavors where $N_f\geq N_a \geq N$, which we call theory $\fA$. To take the decoupling limit, we consider a $U(1)_\text{d}$ subgroup of the flavor symmetry characterized by the charges depicted on the quiver diagram in figure \ref{f:22starweb} (a), and assign a large twisted mass parameter $s_\text{d}$ to it---this breaks the flavor symmetry as $SU(N_f+N_a) \to S\big[ U(N_f) \times U(N_a) \big]$. Let us denote the flavor masses associated to the flavor groups as $s_F$ and $\tilde s_A$. We also scale the FI parameter as $\zeta = s_\text{d}^{N_a-N_f} z$ keeping $z$ fixed. The theory has $\binom{N_f+N_a}{N}$ distinct vacua on the Coulomb branch. In the $s_\text{d} \to \infty$ limit, $\binom{N_f}{N}$ of them remain around the origin (the $\sigma_I$ are of order $s_F, \tilde s_A$) while the other ones fly to infinity along with $s_\text{d}$. We call the vacua that stay around the origin, $(N,0)_\fA$ vacua. The effective theory describing the $(N,0)_\fA$ vacua is the $\cN=(2,2)$ quiver theory in figure \ref{fig: Seiberg duality} (arrow (2) in figure \ref{f:22starweb}), with K\"ahler coordinate $z$.

On the other hand, theory $\fA$ is dual to theory $\fB$ in figure \ref{f:22starweb} (b) by the $\cN=(2,2)^*$ duality of arrow (1): theory $\fB$ is a $U(N_f + N_a - N)$ gauge theory with $N_f + N_a$ flavors and $\zeta'=\zeta^{-1}$. In the decoupling limit, theory $\fB$ also experiences a separation of vacua into groups: the duality, though, maps the $(N,0)_\fA$ vacua of theory $\fA$ to vacua in theory $\fB$ where $\sigma'$ has $N_f-N$ eigenvalues of order $s_\text{d}$ and $N_a$ of order $s_F, \tilde s_A$---we call them $(N_f-N,N_a)_\fB$ vacua. The effective theory for them is an $\cN=(2,2)$ $U(N_f - N) \times U(N_a)$ quiver theory whose diagram is in (d). It turns out that the two low energy theories are related by a chain of cluster dualities: arrows (4)-(5)-(6) in figure \ref{f:22starweb}.

\section{Future directions}
\label{s:future}

In this paper, we have studied certain dualities of two-dimensional $\cN=(2,2)$ supersymmetric quiver gauge theories, and have shown that they realize all elements of cluster algebras: the quiver diagram transforms according to quiver mutations, the FI beta functions transform as the coefficients, the exponentiated complexified FI terms transform as the dual cluster variables, and the dressed $Q$-polynomials transform as cluster variables. To arrive at this result, we have analyzed the sphere partition functions $Z_{S^2}$ of the theories, and have proven that they exactly coincide when the parameters are transformed correctly.

Since 2d $\cN=(2,2)$ theories turn out to be related to various branches of mathematics, one can then claim the existence of a cluster algebra structure in other contexts. We have commented in section \ref{s:applications} that GLSMs can flow in the IR to NLSMs of compact or non-compact Calabi-Yau manifolds: this implies that the extended K\"ahler moduli spaces of such manifolds that arise from non-degenerate quiver gauge theories have a cluster algebra structure and a mutation-invariant quantum metric. It would be interesting to understand if such a structure can be directly derived from the corresponding Picard-Fuchs equations in a mirror model.

Another interesting problem is to see if the $\cN=(2,2)$ quiver gauge theories have any direct connection to quantum integrable systems such as spin chains \cite{Nekrasov:2009uh, Nekrasov:2009ui}. If so, our results would then imply the existence of a cluster algebra structure in those integrable systems. The spin chains associated to the $\cN=(2,2)^*$ gauge theories presented in \cite{Nekrasov:2009uh, Nekrasov:2009ui} are explicitly known. We have shown that the $S^2$ partition functions of those theories are indeed invariant under the ``Grassmannian'' or ``particle-hole'' duality (and its generalizations to quiver theories) presented in \cite{Nekrasov:2009uh, Nekrasov:2009ui}. It would then be interesting to find the quantum integrable systems corresponding to the $\cN=(2,2)$ quiver gauge theories. As mentioned in section \ref{sec: 22star}, one can at least embed the $\cN=(2,2)$ theories into  $\cN=(2,2)^*$ theories as special decoupling limits, and the web of dualities turns out to be self-consistent in such a limit (more details will appear elsewhere). It remains to see what are, if any, the integrable systems that arise in the decoupling limit.

Finally, an interesting direction to explore would be to ``uplift'' the two-dimensional $\cN=(2,2)$ quiver gauge theories to three-dimensional $\cN=2$ theories, possibly on a circle. Such theories may well be related to a $q$-deformed version of cluster algebra. In fact, the Seiberg-like dualities considered in this paper have a direct three-dimensional analog proposed in \cite{Benini:2011mf} and explored in \cite{Closset:2012eq, Xie:2013lya}, with an interesting M-theory realization \cite{Benini:2011cma, Closset:2012ep}. We leave this matter for the future.

\acknowledgments{%
We thank Mohammed Abouzaid, Cyril Closset, Stefano Cremonesi, Mike Douglas, Davide Gaiotto, Chris Herzog, Kentaro Hori, Peter Koroteev, Vijay Kumar, Bruno Le Floch, Sungjay Lee, Dave Morrison, Nikita Nekrasov, Martin Ro\v{c}ek and Dan Xie for useful discussions and comments on the draft. We would especially like to thank Vijay Kumar for inspiring us to work on this project.
FB and PZ thank the KITP, where part of this work was carried out, for its warm hospitality.
DP thanks the Perimeter Institute for Theoretical Physics, the Center for Theoretical Physics at MIT, and the Carg\`ese Summer Institute for their hospitality during the completion of this work.
This research was supported in part by DOE grant DE-FG02-92ER-40697, and in part by the National Science Foundation under Grant No. NSF PHY11-25915.
}

\appendix

\section{Vortex partition function equalities for $|N_f -N_a| \leq 1$}
\label{ap:NfNa10}

In this appendix, we derive the vortex partition function relations \eq{zvortrel} for $N_f = N_a+1$ and $N_f =N_a$. The relations for $N_f \geq N_a+2$ have been derived in \cite{Benini:2012ui} by the method described in section \ref{sec: partition function}. The cases $N_f - N_a = 1$ or $0$ can be derived from the case $N_f =N_a+2$ by considering a limit where one or two of the fundamental twisted masses become very large.

Consider first the equation \eq{NfNa2} for $N_f =N_a+2$:
\begin{multline}
\nonumber
\sum_{\sabs{(n_I)}=n} \; \prod_I \frac{\prod_A \big( \S^{F_I}_A \big)_{n_I}}{\prod_J \big( {-\S^{F_I}_{F_J} -n_I} \big)_{n_J} \, \prod_\Jd \big( {-\S^{F_I}_{F^c_\Jd} - n_I} \big)_{n_I} } = \\
=(-1)^{nN_a} \sum_{\sabs{(n_\Id)}=n} \; \prod_\Id \frac{\prod_A \big(1 - \S^{F^c_\Id}_A \big)_{n_\Id} }{ \prod_\Jd \big( \S^{F^c_\Id}_{F^c_\Jd} - n_\Id \big)_{n_\Jd} \, \prod_J \big( \S^{F^c_\Id}_{F_J} - n_\Id \big)_{n_\Id} } \;.
\end{multline}
Let us denote
\be
x=\S_{F^c_{N_f-N}}=\S_{F^c_{N_a-N+2}} \;,
\ee
and restrict the indices $\Id, \Jd$  to take values in $[N_a-N+1]$ as opposed to $[N_a-N+2]$. This relation can then be written as
\bea
\label{NfNa1from2}
&\sum_{ \sabs{(n_I)}=n} \prod_I \frac{\prod_A \big( \S^{F_I}_A \big)_{n_I} }{ \prod_J \big( {-\S^{F_I}_{F_J}-n_I} \big)_{n_J} \, \prod_{\Jd} \big( {-\S^{F_I}_{F^c_\Jd}-n_I} \big)_{n_I}} \;\cdot\;
\prod_I \frac1{\big( x-\S_{F_I}-n_I \big)_{n_I}} = \\
&\qquad\qquad = (-1)^{nN_a} \sum_{\sabs{(n_\Id)} + p = n} \prod_\Id \frac{ \prod_A \big( \S^A_{F^c_\Id}+1 \big)_{n_\Id} }{ \prod_\Jd \big(\S^{F^c_\Id}_{F^c_\Jd} -n_\Id \big)_{n_\Jd} \, \prod_J \big( \S^{F^c_\Id}_{F_J} -n_\Id \big)_{n_\Id} } \\
&\qquad\qquad\qquad\times \frac{(-1)^p \prod_A  \big( -x+\tS_A+1 \big)_p }{p! \, \prod_{\Jd} \Big[ \big( x-\S_{F^c_\Jd}-p \big)_{n_\Jd} \big( {-x+\S_{F^c_\Jd}-n_\Jd} \big)_p \Big] \prod_J \big( x-\S_{F_J}-p \big)_p} \;,
\eea
where the summation on the right-hand side is over the $(N_a-N+1)$-tuples $(n_{I'})$ and the non-negative integer $p$. Since the left- and right-hand sides of this equation are both rational functions of $x$, we can expand them as a Laurent series in $y=x^{-1}$. Since
\be
\prod_I  \frac1{(x-\S_{F_I}-n_I)_{n_I}}= y^n + \cO(y^{n+1})
\ee
and
\be
\frac{ (-1)^p \, \prod_A \big( {-x+\tS_A+1} \big)_p }{
p! \, \prod_{\Jd} \Big[ \big( x-\S_{F^c_\Jd}-p \big)_{n_\Jd} \big( {-x+\S_{F^c_\Jd}-n_\Jd} \big)_p \Big]
\prod_J \big( x-\S_{F_j}-p \big)_p} = \frac{(-1)^{-pN} }{p!} y^n + \cO(y^{n+1}) \;,
\ee
we obtain the relation
\begin{multline}
\label{NfNa1}
\sum_{\sabs{(n_I)}=n} \prod_I \frac{ \prod_A  \big( \S^{F_I}_A \big)_{n_I} }{
\prod_J \big( {-\S^{F_I}_{F_J} -n_I} \big)_{n_J} \, \prod_\Jd \big( {-\S^{F_I}_{F^c_\Jd} -n_I} \big)_{n_I} } = \\
= \sum_{p=0}^n \frac{ (-1)^{p(N+N_a)+(n-p)N_a} }{p!} \sum_{\sabs{(n_\Id)}=n-p}  \prod_\Id \frac{ \prod_A \big( \S^{A}_{F^c_\Id}+1 \big)_{n_\Id} }{
\prod_\Jd \big( \S^{F^c_\Id}_{F^c_\Jd} -n_\Id \big)_{n_\Jd} \, \prod_J \big( \S^{F^c_\Id}_{F_J} -n_\Id \big)_{n_\Id} } \;,
\end{multline}
by taking the coefficient of $y^n$ in the Laurent expansion of \eq{NfNa1from2}. This implies the relation \eq{zvortrel} for $N_f = N_a+1$. We have extracted the leading order behavior of \eq{NfNa2} in the limit that one of the twisted masses $x$ becomes very large, \ie{}, in a decoupling limit of the theory with $N_f =N_a+2$.

Starting from equation \eq{NfNa1}, we can pursue a similar strategy and get the vortex partition function equality for $N_f = N_a$. As before, we denote
\be
x=\S_{F^c_{N_a-N+1}} \;,
\ee
and restrict the indices $\Id, \Jd$  to take values in $[N_a-N]$. Without loss of generality, we can further assume that $F^c_{N_a-N+1} = N_a+1$, so that the index $F$ enumerating fundamental fields takes values in $[N_a]$. Equation \eq{NfNa1} can then be written as
\bea
\label{NfNa0from1}
&\hspace{-1em} \sum_{\sabs{(n_I)}=n} \prod_I \frac{ \prod_A \big( \S^{F_I}_A \big)_{n_I} }{
\prod_J \big( {-\S^{F_I}_{F_J}-n_I} \big)_{n_J} \, \prod_{\Jd} \big( {-\S^{F_I}_{F^c_\Jd}-n_I} \big)_{n_I}} \;\cdot\;
\prod_I \frac1{\big( x-\S_{F_I}-n_I \big)_{n_I}} = \\
&\quad = \sum_{\sabs{(n_\Id)} +m = n} (-1)^{(n-m)N_a} \prod_\Id \frac{ \prod_A \big( \S^A_{F^c_\Id}+1 \big)_{n_\Id} }{
\prod_\Jd \big( \S^{F^c_\Id}_{F^c_\Jd} -n_\Id \big)_{n_\Jd} \, \prod_J \big( \S^{F^c_\Id}_{F_J} -n_\Id \big)_{n_\Id} } \\
&\quad\quad \times \frac{(-1)^{m(N_a-N)} }{m!} \sum_{q=0}^m (-1)^q \dbinom{m}{q} \prod_\Jd \frac{ x-\S_{F^c_\Jd}+n_\Jd-q }{ \big( x-\S_{F^c_\Jd} \big)_{n_\Jd+1}} \: \frac{\prod_A \big( {-x+\tS_A +1} \big)_q }{ \prod_F \big( {-x+\S_F+1} \big)_q} \;,
\eea
where we have used the relation
\be
1+ \frac{m-l}a = \frac{(a+1)_m\, (-a+1)_l}{(a-l)_m \, (-a-m)_l} \qquad\qquad \text{for $m,l\in\bZ$} \;.
\ee
Taking the Laurent expansion with respect to $y=x^{-1}$, the left-hand side of \eq{NfNa0from1} is given by
\be
\label{lhsof01}
\text{LHS}(\ref{NfNa0from1}) = \sum_{\sabs{(n_I)}=n} \prod_I \frac{ \prod_A \big( \S^{F_I}_A \big)_{n_I} }{ \prod_J \big( {-\S^{F_I}_{F_J}-n_I} \big)_{n_J} \, \prod_{\Jd} \big( {-\S^{F_I}_{F^c_\Jd}-n_I} \big)_{n_I}} \, y^n + \cO( y^{n+1} ) \;.
\ee
The right-hand side is more involved. Defining
\begin{multline}
\label{def Qm}
Q_m \big( y,T_\Jd,\tS_A,\S_F \big) = \sum_{q=0}^m (-1)^q \dbinom{m}{q} \prod_\Jd \big( 1+(T_\Jd-q)y \big) \\
\times \prod_{A=1}^{N_a} \prod_{k=1}^q \big( 1-(\tS_A+k)y \big) \, \prod_{F=1}^{N_f} \prod_{k=q+1}^m \big( 1-(\S_F+k)y \big) \;,
\end{multline}
we find that the third line of \eq{NfNa0from1} is given by
\be
\label{rhsof01}
\text{3$^\text{rd}$ line of (\ref{NfNa0from1})} = \frac{(-1)^{m(N_a-N)} }{m!} \, \frac{ y^{n-m} \, Q_m \big( y, -\S_{F^c_\Jd}+n_\Jd, \tS_A, \S_F \big) }{ \prod_\Jd \prod_{k=0}^{n_\Jd} \big( 1-(\S_{F^c_\Jd}-k)y \big) \, \prod_F \prod_{k=1}^{m} \big( 1-(\S_{F}+k)y \big) } \;.
\ee
We show in appendix \ref{ap:forNfNa0} that
\be
\label{qlead}
Q_m \big( y,T_\Jd,\tS_A,\S_F \big) = \frac{ \Big( \sum_A \tS_A-\sum_F \S_F  + N' \Big)! }{ \Big( \sum_A \tS_A- \sum_F \S_F+ N'-m \Big)! } \, y^m + \cO(y^{m+1}) \;,
\ee
where $N' = N_a -N$.  From \eq{NfNa0from1}, \eq{lhsof01}, \eq{rhsof01} and \eq{qlead} we arrive at the equality:
\bea
\label{NfNa0}
&\sum_{\sabs{(n_I)}=n} \prod_I \frac{ \prod_A \big( \S^{F_I}_A \big)_{n_I} }{ \prod_J \big( {-\S^{F_I}_{F_J}-n_I} \big)_{n_J} \, \prod_{\Jd} \big( {-\S^{F_I}_{F^c_\Jd}-n_I} \big)_{n_I}} = \\
&\qquad = \sum_{\sabs{(n_\Id)} + m = n} (-1)^{m(N_a-N)} \, \dbinom{\sum_A \tS_A - \sum_F \S_F + N_a-N}{m} \\
&\qquad\qquad \times  (-1)^{(n-m)N_a} \prod_\Id \frac{\prod_A \big( \S^{A}_{F^c_\Id}+1 \big)_{n_\Id} }{ \prod_\Jd \big( \S^{F^c_\Id}_{F^c_\Jd} -n_\Id \big)_{n_\Jd} \, \prod_J \big( \S^{F^c_\Id}_{F_J} -n_\Id \big)_{n_\Id} } \;.
\eea
This implies the relation \eq{zvortrel} for $N_f=N_a$.

\section{Some useful identities}
\label{ap:forNfNa0}

In this appendix we prove that the polynomial $Q_m \big( y,T_\Jd,\tS_A,\S_F \big)$ defined in (\ref{def Qm}) with $N_f = N_a$ satisfies \eq{qlead}:
\be
\label{appbmain}
Q_m \big( y,T_\Jd,\tS_A,\S_F \big) = \frac{( S_a -S_f + N')! }{ (S_a -S_f + N' -m)!} \, y^m + \cO(y^{m+1}) \;,
\ee
where we have defined $S_f = \sum_{F=1}^{N_f} \S_F$ and $S_a = \sum_{A=1}^{N_a} \tS_A$. We present a brute force proof of this identity, consisting mainly of rearrangements of sums and basic identities involving binomial coefficients. In particular we extensively use the following identities:
\be
\sum_{\sum_ i m_i = n} \prod_i \dbinom{a_i}{m_i} = \dbinom{\sum_i a_i}{n} \;, \qquad\qquad
\dbinom{m}{q} \dbinom{q}{r} = \dbinom{m}{r} \dbinom{m-r}{q-r} \;.
\ee

To prove \eq{appbmain} in full generality, it is useful to prove the identities for simpler values of the arguments. Let us denote $Q_m$ with empty $(T_\Jd)$ as
\be
\label{Qm.}
Q_m \big( y, \tS_A,\S_F \big) = \sum_{q=0}^m (-1)^q \dbinom{m}{q} \prod_A \prod_{k=1}^q \big( 1-(\tS_A+k)y \big) \, \prod_F \prod_{k=q+1}^m \big( 1-(\S_F+k)y \big) \;,
\ee
and $Q_m$ with all $T_\Jd =0$ as $Q_m\big( y, 0^{N'}, \ts_A, \S_F \big)$. We first prove the following
\begin{proposition} \label{prop.}
The $y$-expansion of $Q_m(y, \tS_A ,\S_F)$ is given by
\be
Q_m \big( y, \tS_A ,\S_F \big) = \frac{(S_a-S_f)! }{ (S_a-S_f-m)!} \, y^m + \cO(y^{m+1}) \;.
\label{prop1}
\ee
\end{proposition}
\begin{proof}
We proceed by induction on $m$. Equation \eq{prop1} can be checked explicitly for $m=0,1$. Let us now assume that \eq{prop1} holds for $m \leq M-1$. Using the symmetric functions $e_p (x_1 ,\ldots,x_k)$ defined as
\be
\prod_{i=1}^k (1+x_i t) = \sum_{p=0}^k e_p (x_1, \ldots, x_k) \, t^p \;,
\ee
we can express the coefficients of $Q_m(y, \tS_A ,\S_F) = \sum_{r=0}^{mN_a} C_{m,r} (\tS_A,\S_F) (-y)^r$ as
\be
C_{m,r} (\tS_A,\S_F) = \sum_{|\vec t| + |\vec r| = r} \prod_{i=1}^m \, i^{r_i} \, \dbinom{N_a-t_i}{r_i} \;\cdot\; \sum_{q=0}^m (-1)^q \dbinom{m}{q} \prod_{j=1}^q e_{t_j}(\tS_A) \prod_{l=q+1}^m e_{t_l}(\S_F) \;.
\ee
Here $t_i$ and $r_i$ are $m$-tuples of non-negative integers, and we use the notation $|\vec s| = \sum_{i=1}^m s_i$ for $m$-tuples $\vec s$. For $m \leq M-1$, the assumption that \eq{prop1} holds true implies
\be
C_{m,r}(\tS_A,\S_F) =0 \qquad\text{for $r <m$} \;,\qquad\qquad C_{m,m}(\tS_A,\S_F) = (-1)^m \frac{(S_a-S_f)! }{ (S_a -S_f -m)!} \;.
\ee

We can observe from the definition \eq{Qm.} that
\bea
Q_M(y, \tS_A,\S_F)
& =  Q_{M-1} (y, \tS_A,\S_F) \prod_F \big( 1-(\S_F+M)y \big) \\
&\quad - Q_{M-1} (y, \tS_A+1,\S_F+1) \prod_A \big( 1-(\tS_A+1)y \big) \\
&= (-1)^M \, C_{M,M}(\tS_A,\S_F) \, y^M + \cO(y^{M+1}) \;,
\eea
due to the inductive assumption. It hence follows that $C_{M,r}=0$ for $r <M$.
The inductive assumption further implies that
\be
\label{CasdiffC}
C_{M,M}(\tS_A,\S_F) = (-1)^M \, \Big( S_a -S_f-(M-1)N_a \Big) \, \frac{ (S_a-S_f)! }{ (S_a-S_f-M+1)!} + \Delta C \;,
\ee
where we have defined $\Delta C = C_{M-1,M}(\tS_A,\S_F) - C_{M-1,M}(\tS_A+1,\S_F+1)$.
Using the relation $e_p \big( x_1 +1 ,\ldots, x_k +1 \big) = \sum_{r=0}^p \binom{k-p+r}{k-p} \, e_{p-r} (x_1,\ldots,x_k)$, we can write
\begin{multline}
\Delta C = -\sum_{|\vec t| + |\vec r| = M}
\left[ \prod_{i=1}^{M-1} i^{r_i} \dbinom{N_a-t_i}{r_i} \right]
\sum_{\substack{0\leq r_i' \leq t_i \\ |\vec r~\!\!'| \geq 1} }
\prod_{i=1}^{M-1} \dbinom{N_a-t_i+r_i'}{N_a-t_i} \\
\times \sum_{q=0}^{M-1} (-1)^q \dbinom{M-1}{q}
\prod_{i=1}^q e_{t_i-r_i'}(\tS_A)
\prod_{i=q+1}^{M-1} e_{t_i-r_i'}(\S_F) \;,
\end{multline}
where $t_i$, $r_i$ and $r_i'$ are $(M-1)$-tuples of non-negative integers. After the change of variables $R = |\vec r~\!'|$, $t_i' = t_i -r_i'$, this equation can be re-written as
\begin{multline}
\label{diffC}
\Delta C = -\sum_{R=1}^M \, \sum_{|\vec t'| + |\vec r| =M-R}
\left[ \prod_{i=1}^{M-1} i^{r_i} \right]
\sum_{|\vec r~\!\!'| = R}
\left[ \prod_{i=1}^{M-1} \dbinom{N_a-t_i'-r_i'}{r_i}
\dbinom{N_a-t_i'}{N_a -t_i'-r_i'} \right] \\
\times \sum_{q=0}^{M-1} (-1)^q \dbinom{M-1}{q}
\prod_{i=1}^q e_{t_i'}(\tS_A)
\prod_{i=q+1}^{M-1} e_{t_i'}(\S_F) \;.
\end{multline}
A simple computation shows that
\be
\sum_{|\vec r~\!\!'| = R} \prod_{i=1}^{M-1} \dbinom{N_a-t_i'-r_i'}{r_i} \dbinom{N_a-t_i'}{N_a -t_i'-r_i'}
= \dbinom{(M-1)N_a-M+R}{R} \prod_{i=1}^{M-1} \dbinom{N_a-t_i'}{r_i} \;.
\ee
Hence equation \eq{diffC} becomes
\bea
\Delta C &= -\sum_{R=1}^M \dbinom{(M-1)N_a-M+R}{R} \sum_{|\vec t'| + |\vec r| =M-R}
\prod_{i=1}^{M-1} i^{r_i} \dbinom{N_a-t_i'}{r_i} \\
&\quad \times \sum_{q=0}^{M-1} (-1)^q \dbinom{M-1}{q} \,
\prod_{i=1}^q e_{t_i'}(\tS_A) \, \prod_{i=q+1}^{M-1} e_{t_i'}(\S_F) \\
&= -\sum_{R=1}^M \dbinom{(M-1)N_a-M+R}{R} \, C_{M-1,M-R}(\tS_A,\S_F) \\
&= - (-1)^{M-1} (M-1)(N_a-1) \, \frac{(S_a -S_f)! }{ (S_a-S_f -M+1)!} \;,
\eea
by the inductive assumption. Plugging this expression into equation \eq{CasdiffC}, we find that
\be
C_{M,M}(\tS_A,\S_F) = (-1)^M \, \frac{(S_a-S_f)! }{ (S_a -S_f-M)!} \;,
\ee
which concludes our proof.
\end{proof}

The $y$-expansion of $Q_m (y, T_\Jd,\tS_A,\S_F)$ with non-empty $T_\Jd$ then follow as  corollaries.

\begin{corollary}\label{cor0}
The $y$-expansion of $Q_m \big( y, 0^{N'},\tS_A,\S_F \big)$ is given by
\be
Q_m \big( y, 0^{N'},\tS_A,\S_F \big) = \frac{(S_a-S_f + N')! }{ (S_a-S_f + N'-m)!} \, y^m + \cO(y^{m+1}) \;.
\ee
\end{corollary}
\begin{proof}
Let us first note that for integer $s$, $q^s = \sum_{r=0}^s f_{r,s} \binom{q}{r}$ for some combinatorial factors $f_{r,s}$ independent of $q$. In particular, $f_{s,s}=s!$. Using this fact, we see that
\be
\dbinom{m}{q} \, (1-qy)^{N'} = \sum_{s=0}^{N'} \dbinom{N'}{s} \dbinom{m}{q} \, q^s (-y)^s = \sum_{s=0}^{N'} \sum_{r=0}^s \dbinom{N'}{s} \dbinom{m}{r} \dbinom{m-r}{q-r} f_{r,s}\, (-y)^s \;.
\ee
Then $Q_m \big( y, 0^{N'},\tS_A,\S_F \big)$ can be written as
\bea
& Q_m (y, 0^{N'},\tS_A,\S_F) = \sum_{s=0}^{N'} \sum_{r=0}^s \dbinom{N'}{s} \dbinom{m}{r} \, f_{r,s} \, y^s (-1)^{s} \prod_A \prod_{k=1}^r \big( 1-(\tS_A+k)y \big) \\
&\qquad \times \sum_{q'=0}^{m-r} (-1)^{q'+r}\dbinom{m-r}{q'}
\prod_A \prod_{k=1}^{q'} \big( 1-(\tS_A+r+k)y \big) \prod_F \prod_{k=q'+1}^{m-r} \big( 1-(\tS_A+r+k)y \big) \\
&=\sum_{s=0}^{N'} \sum_{r=0}^s \dbinom{N'}{s} \dbinom{m}{r} \, f_{r,s} y^s (-1)^{r+s}
\prod_A \prod_{k=1}^r \big( 1-(\tS_A+k)y \big) \;\cdot\; Q_{m-r} \big( y, \tS_A+r, \S_F+r \big) \;.
\eea
The variable $q'$ of this equation is related to the variable $q$ of the defining equation (\ref{def Qm}) by $q' =q-r$. Applying proposition \ref{prop.} to $Q_{m-r} (y,  \tS_A+r, \S_F+r)$, it follows that
\begin{align}
\!\! Q_m \big( y, 0^{N'},\tS_A,\S_F \big) &= \sum_{s=0}^{N'} \sum_{r=0}^s \dbinom{N'}{s} \dbinom{m}{r} \frac{(S_a -S_f)! }{(S_a-S_f-m+r)!} f_{r,s} (-1)^{r+s} y^{m+s-r} + \cO (y^{m+1}) \nn \\
&= \sum_{s=0}^{N'} \dbinom{N'}{s} \dbinom{m}{s} \frac{(S_a -S_f)! \, s! }{ (S_a-S_f-m+s)!} \, y^{m} + \cO (y^{m+1}) \nn \\
&= \frac{(S_a -S_f+ N')! }{ (S_a-S_f+ N'-m)!} \, y^m +\cO (y^{m+1}) \;,
\end{align}
where we have used the fact that $f_{s,s}=s!$ in going from the first to the second line.
\end{proof}

\begin{corollary}\label{corfinal}
The $y$-expansion of $Q_m \big( y,T_\Jd,\tS_A,\S_F \big)$ is given by equation \eq{appbmain}.
\end{corollary}
\begin{proof}
From the relation $\prod_{\Jd = 1}^{N'} \big( 1+(T_\Jd-p)y \big) = \sum_{s=0}^{N'} e_s(T_\Jd) \, y^s \, (1-py)^{N'-s}$, it follows that
\bea
Q_m \big( y,T_\Jd,\tS_A,\S_F \big)
&= \sum_{s=0}^{N'} y^s \, e_s(T_\Jd) \, Q_m \big( y,0^{N'-s},\tS_A,\S_F \big) \\
&= \sum_{s=0}^{N'} e_s(T_\Jd) \, \frac{(S_a-S_f+ N'-s)! }{ (S_a-S_f+ N'-s-m)!} \, y^{m+s} + \cO(y^{m+1}) \\
&= \frac{ (S_a-S_f+ N')! }{ (S_a-S_f+ N'-m)!} \, y^m + \cO(y^{m+1}) \;,
\eea
since $e_0(T_\Jd)=1$ for any $N'$-tuple $T_{\Jd}$.
\end{proof}

\section{Conformal embedding of non-conformal theories}
\label{ap:UVIR}

We showed in section \ref{ss:SDandCA} that a quiver theory whose nodes are not all conformal can be thought of as the IR limit of a conformal UV theory. The latter can be constructed by adding a flavor node $U(1)_\text{UV}$ and bifundamental matter charged under this node and the gauge nodes to offset the non-conformality of each gauge group. Taking the twisted mass $s_0$ associated to the flavor node to infinity, and scaling the UV couplings appropriately, one recovers the IR theory. We introduce the parameter $u$:
\be
\label{defu}
\mu\, u = -is_0 = - s_{0+} \;,\qquad\qquad \mu\, \bar u = i\bar s_0 = s_{0-} \;,
\ee
where $\mu$ is the renormalization scale at which the IR couplings are defined.
We want to show that the sphere partition function behaves consistently with the scale matching formula
\be
\label{UVIR2}
z^\text{UV}_i = u^{\beta_i} z_i(\mu) \;,
\ee
where $z^\text{UV}_i$ and $z_i(\mu)$ are the K\"ahler parameters of the UV and IR theory respectively, $\mu = 1/r$ is fixed by the radius $r$ of the sphere, and $\beta_i = N_a(i) - N_f(i)$ is the FI beta function for the $i$-th node in the IR theory.

We do not need to deal with the full quiver to check \eq{UVIR2}: since we integrate out some (anti)fundamentals separately at each node, we can simply consider a single-node $U(N)$ theory with flavors. Let us first consider the case with $N_f$ fundamentals and $N_a (<N_f)$ antifundamentals in the IR theory. The UV theory has $N_f^\text{UV} = N_a^\text{UV} = N_f$ flavors. The additional antifundamental matter has charge 1 with respect to the $U(1)_\text{UV}$ flavor node. The partition function of the UV theory is then given by the following Coulomb branch integral:
\begin{multline}
\label{ZUV}
Z^\text{UV}_{U(N)} \big( \S_{F\pm},\tS_{A\pm},u;z^\text{UV} \big) = \frac1{N!} \sum_{\m_I \in \mathbb{Z}^N} \int \prod_{I=1}^N \frac{d\s_I}{2\pi} \big( e^{i\pi(N^\text{UV}_f-1)} z^\text{UV} \big)^{\s_{I+}}
\big(e^{-i\pi(N^\text{UV}_f-1)} \bz^\text{UV} \big)^{\s_{I-}} \\
\times \prod_{I<J}^N \big( {-\S^I_{J+} \S^I_{J-}} \big) \times
\prod_{I=1}^N \prod_{F=1}^{N_f} \frac{\Ga(-\S^I_{F+})}{\Ga (1+\S^I_{F-})} \prod_{A=1}^{N_a} \frac{\Ga(\S^I_{A+})}{\Ga (1-\S^I_{A-})}
\left[ \frac{\Ga(\s_{I+} - s_{0+}) }{ \Ga (1-\s_{I-} +s_{0-}) } \right]^{-\b} \;,
\end{multline}
where $\b = (N_a-N_f)$. Note that, as usual, all masses are expressed in units of $1/r$. Using the asymptotic expansion of $\Ga(z)$ for large $z$,%
\footnote{From Stirling's approximation, valid for $|z| \to \infty$ and $|\arg z| < \pi - \epsilon$, we get:
$$
\Gamma(z) = \sqrt\frac{2\pi}z \, \Big( \frac ze \Big)^z \big( 1 + \cO(z^{-1}) \big) \qquad\Rightarrow\qquad \Gamma(a+z) = \sqrt{2\pi}\, e^{-z} z^{a+z - \frac12} \big( 1 + \cO(z^{-1}) \big) \;.
$$
}
we have
\be
\label{GammaLimit}
\frac{\Ga(x-s_{0+})}{\Ga(1-y + s_{0-})} = \frac{\Ga(x+u)}{\Ga(1-y +\bar u)} = \left( \frac ue \right)^u
\left( \frac{\bar u}e \right)^{-\bar u} u^{x-1/2} \, \bar u^{y-1/2} \Big( 1 + \cO\big( |u|^{-1} \big)  \Big)
\ee
for fixed $x$ and $y$ when $s_0$---or more precisely, the real part of $s_0$---is large. Hence, in the limit of large $s_0$, $Z^\text{UV}_{U(N)}$ can be written as
\begin{multline}
\label{ZUVscaling}
Z^\text{UV}_{U(N)} \Big{/} \left( e^{-u+\bar u} u^{u-1/2} \bar u ^{-\bar u-1/2}\right)^{-\b N} = {} \\
{} = \frac1{N!} \sum_{\m_I \in \mathbb{Z}^N} \int \prod_{I=1}^N
\frac{d\s_I}{2\pi} \, \big( e^{i\pi(N^\text{UV}_f-1)} z^\text{UV} u^{-\b} \big)^{\s_{I+}}
\big( e^{-i\pi(N^\text{UV}_f-1)} \bz^\text{UV} \bar u^{-\b} \big)^{\s_{I-}} \\
\times \prod_{I<J} \big( {-\S^I_{J+} \S^I_{J-}} \big) \times \prod_I \prod_F \frac{\Ga(-\S^I_{F+})}{\Ga (1+\S^I_{F-})} \prod_A \frac{\Ga(\S^I_{A+})}{\Ga (1-\S^I_{A-})} \quad + \quad \cO \big( |u|^{-1} \big) \;.
\end{multline}
The leading term on the right-hand side is the partition function of a theory with $N_f$ fundamentals and $N_a$ antifundamentals. The relative factor on the left-hand side does not depend on any of the twisted masses of the IR theory and can be neglected.%
\footnote{If the RG flow were between two fixed points, this relative factor would contain information about the difference of the two central charges.}
We can thus write:
\be
\label{ZUV final}
Z^\text{UV}_{U(N)} \big( \S_{F\pm},\tS_{A\pm},u;z^\text{UV} \big) \cong Z^{N_f,N_a}_{U(N)} \big( \S_{F\pm},\tS_{A\pm};\, z^\text{UV}u^{-\b} \big) + \cO\big( |u|^{-1} \big) \;,
\ee
hence the identification $z= u^{-\b} z^\text{UV}$.

One might be concerned that applying the limit \eq{GammaLimit} to the integral \eq{ZUV} is not justified, since $\s_{I\pm}$ are values that are summed/integrated over and can become arbitrarily large, so let us clarify this point. For $N_f > N_a$ we have $\beta < 0$, therefore if we send $u\to\infty$ with $z(\mu)$ fixed, we let $z_\text{UV} \to 0$ and $t \to +\infty$. We can then close the integrals in \eq{ZUV} in the lower complex half-planes, and pick the residues at the poles of the gamma functions related to the fundamentals. Since poles are located at $\s_{I\pm} = \S_{F_I} + n_{I\pm}$, whose positions are independent of $s_0$, taking the large $s_0$ limit in this formulation is legitimate. In the limit, the vortex partition functions of the UV theory (in the $N_f$ vacua) asymptote to those of the IR theory.

All the statements go through in the case $N_a > N_f$. In this case we add fundamentals to obtain the UV theory (\ie{}, $N_f^\text{UV} = N_a^\text{UV} = N_a$) and its partition function is given by
\begin{multline}
\label{ZUV2}
Z^\text{UV}_{U(N)} \big( \S_{F\pm},\tS_{A\pm},u; z^\text{UV} \big) = \frac1{N!} \sum_{\m_I \in \mathbb{Z}^N} \int \prod_{I=1}^N \frac{d\s_I}{2\pi} \, \big( e^{i\pi(N^\text{UV}_f-1)} z^\text{UV} \big)^{\s_{I+}}
\big( e^{-i\pi(N^\text{UV}_f-1)} \bz^\text{UV} \big)^{\s_{I-}} \\
\times \prod_{I<J} \big( {-\S^I_{J+} \S^I_{J-}} \big) \prod_I \prod_F \frac{\Ga(-\S^I_{F+})}{\Ga (1 + \S^I_{F-})} \prod_A \frac{\Ga(\S^I_{A+})}{\Ga (1-\S^I_{A-})} \left[ \frac{\Ga(-\s_{I+} -u)}{\Ga (1+\s_{I-} - \bar u)} \right]^\b \;.
\end{multline}
The ratio of gamma functions involving $u$ now has the large $|u|$ limit
\be
\frac{\Ga(-x-u)}{\Ga(1+y -\bar u)} = \left( \frac{e^{i\pi}u}e \right)^{-u} \left( \frac{e^{-i\pi}\bar u}e \right)^{\bar u} (e^{i\pi}u)^{-x-1/2} \, (e^{-i\pi}\bar u)^{-y-1/2} \Big( 1 + \cO\big( |u|^{-1} \big)  \Big)
\ee
for fixed $x$ and $y$. Noticing that $e^{i\pi(N_f^\text{UV}-\b+1)} = e^{i\pi(N_f-1)}$, one finally arrives at the same relation as in \eq{ZUV final}. To justify the limit in this case, we notice that $\beta>0$ therefore $z_\text{UV} \to \infty$ and $t \to -\infty$ as $u \to \infty$. Then, the Coulomb branch integrals can be closed in the upper half-planes, and is reduced to a sum over the residues at the poles from the antifundamentals. As in the case when $N_f > N_a$, the position of these poles do not depend on $s_0$; the large $s_0$ limit can be taken at the level of vortex partition functions.

\section{Cancellation of theta-angle shifts}
\label{ap:thetacancel}

We show that the extra phases in equation \eq{compcoeff} from ``annihilations'' cancel out. In other words, upon rewriting \eq{compcoeff} as
\be
\label{compcoeff2}
z_j' = \begin{cases}
e^{i\pi \varphi_j} \, z_j z_k^{[b_{kj}]_+} &\text{if $N_f > N_a$}  \\
e^{i\pi \varphi_j} \, z_j \Big( \dfrac{ z_k }{1+ z_k } \Big)^{[b_{kj}]_+}
(1+z_k)^{-[-b_{kj}]_+} &\text{if $N_f = N_a$} \\
e^{i\pi \varphi_j} \, z_j z_k^{-[-b_{kj}]_+} &\text{if $N_f < N_a$}
\end{cases}
\ee
for $j \neq k$, we show that $\varphi_j \equiv 0 \mod{2}$. Throughout this section we use the symbol ``$\equiv$" to denote equivalence modulo $2$. The following lemma turns out to be useful:

\begin{lemma}\label{aij}
Let us consider a mutation of a quiver with respect to node $k$, and denote the number of annihilations that occur between nodes $i$ and $j$ upon mutation as $a_{ij} = a_{ji}$. Then
\be
a_{ij} \equiv [b_{ij}']_+ + [b_{ij}]_+ +[b_{ik}]_+ [b_{kj}]_+ \,.
\label{aijlemma}
\ee
\end{lemma}
\begin{proof}
Recall that $b_{ij}' = b_{ij} + \sign(b_{ik}) [b_{ik}b_{kj}]_+$. The number of annihilations is zero if the two summands have the same sign, and equals the minimum of their absolute values if they have opposite sign, namely:
\be
a_{ij} = \big[ -\sign(b_{ij}b_{ik}) \big]_+ \, \min\big( |b_{ij}|,\, [b_{ik}b_{kj}]_+ \big) \;.
\ee
The values of $a_{ij}$ modulo 2 are given in table \ref{t:aij}, depending on the signs of $b_{ik}$, $b_{kj}$, $b_{ij}$ and $b_{ij}'$. By explicit comparison with the values of the right-hand side of \eq{aijlemma}, also tabulated in table \ref{t:aij}, we find that equation \eq{aijlemma} is true.
\end{proof}

\begin{table}[t]
\center
\begin{tabular}{ | c | c | c| c || c || c | c| c || c| }
\hline
$b_{ik}$ & $b_{kj}$ & $b_{ij}$ & $b_{ij}'$ &
$a_{ij}$ & $[b_{ij}']_+$ & $[b_{ij}]_+$ & $[b_{ik}]_+ [b_{kj}]_+$ &
$[b_{ij}']_+\!\! +\! [b_{ij}]_+\!\! +\! [b_{ik}]_+ [b_{kj}]_+$    \\ \hline\hline
$+$ & $+$ & $+$ & $+$ & $0$ & $b_{ij}\!+\!b_{ik}b_{kj}$ & $b_{ij}$ & $b_{ik} b_{kj}$ &
$2b_{ij}+2b_{ik} b_{kj}$ \\ \hline
$+$ & $+$ & $-$ & $+$ & $b_{ij}$ & $b_{ij}\!+\!b_{ik}b_{kj}$ & $0$ & $b_{ik} b_{kj}$ &
$b_{ij}+2b_{ik} b_{kj}$ \\ \hline
$+$ & $+$ & $-$ & $-$ & $b_{ik}b_{kj}$ & $0$ & $0$ & $b_{ik} b_{kj}$ &
$b_{ik} b_{kj}$ \\ \hline
$+$ & $-$ & $\pm$ & $\pm$ & $0$ & $[b_{ij}]_+$ & $[b_{ij}]_+$ & $0$ & $2[b_{ij}]_+$ \\ \hline
$-$ & $+$ & $\pm$ & $\pm$ & $0$ & $[b_{ij}]_+$ & $[b_{ij}]_+$ & $0$ & $2[b_{ij}]_+$ \\ \hline
$-$ & $-$ & $+$ & $+$ & $b_{ik}b_{kj}$ & $b_{ij}\!-\!b_{ik}b_{kj}$ & $b_{ij}$ & $0$ &
$2b_{ij}-b_{ik}b_{kj}$ \\ \hline
$-$ & $-$ & $+$ & $-$ & $b_{ij}$ & $0$ & $b_{ij}$ & $0$ & $b_{ij}$ \\ \hline
$-$ & $-$ & $-$ & $-$ & $0$ & $0$ & $0$ & $0$ & $0$ \\ \hline
\end{tabular}
\caption{\small The value of $a_{ij}$ and other functions of the adjacency matrix for all possible combinations of signs of $b_{ik}$, $b_{kj}$, $b_{ij}$ and $b_{ij}'$. Note that $a_{ij}$ is reported modulo $2$.
\label{t:aij}}
\end{table}

\begin{proposition}\label{propphase}
The phases $\varphi_j$ in equation \eq{compcoeff2} satisfy $\varphi_j \equiv 0$.
\end{proposition}
\begin{proof}
When $N_f \geq N_a$:
\bea
\varphi_j &\equiv \big( N_k-N_f(k) \big)[b_{kj}]_+ + \big( N_k-N_f(k)\big)[-b_{kj}]_+ + N_f(j)'+N_f(j) + \sum_{i\, ( \neq k)} N_i a_{ij} \;.
\eea
Using the identities
\bea
N_f (j) &= N_k[b_{jk}]_+ + \sum\nolimits_{i\,( \neq k)} N_i [b_{ji}]_+ \;, &
N_f (k) &= \sum\nolimits_{i\,(\neq k)} N_i [b_{ki}]_+ \;, \\
N_f (j)' &= \big( N_f(k)-N_k \big) [-b_{jk}]_+ + \sum\nolimits_{i\,(\neq k)} N_i [b'_{ji}]_+ \;,
\eea
we find that
\be
\varphi_j \equiv
\sum\nolimits_{i\,(\neq k)} N_i \Big( [b_{ji}']_+ + [b_{ji}]_+ + [b_{jk}]_+ [b_{ki}]_+ \Big)
+ \sum\nolimits_{i\,(\neq k)} N_i a_{ij} \equiv 0 \;,
\ee
where we have used $[b_{kj}]_+ - [-b_{kj}]_+ = b_{kj}$ and lemma \ref{aij}.

When $N_f < N_a$:
\be
\varphi_j \equiv \big( N_k-N_a(k) \big)[b_{kj}]_+ + \big( N_k-N_f(k) \big)[-b_{kj}]_+ + N_f(j)' + N_f(j) + \sum_{i \neq k} N_i a_{ij} \;,
\ee
and by a similar reasoning, we conclude $\varphi_j \equiv 0$.
\end{proof}

%
%

{
\bibliographystyle{JHEP}
\bibliography{2dClusterAlgebra}

\providecommand{\href}[2]{#2}\begingroup\raggedright\begin{thebibliography}{10}

\bibitem{Seiberg:1994pq}
N.~Seiberg, {\it {Electric - magnetic duality in supersymmetric nonAbelian
  gauge theories}},  {\em Nucl.Phys.} {\bf B435} (1995) 129--146,
  [\href{http://xxx.lanl.gov/abs/hep-th/9411149}{{\tt hep-th/9411149}}].

\bibitem{FominZ1}
S.~Fomin and A.~Zelevinsky, {\it Cluster algebras. {I}. {F}oundations},  {\em
  J. Amer. Math. Soc.} {\bf 15} (2002), no.~2 497--529,
  [\href{http://xxx.lanl.gov/abs/math/0104151}{{\tt math/0104151}}].

\bibitem{FominZ2}
S.~Fomin and A.~Zelevinsky, {\it Cluster algebras. {II}. {F}inite type
  classification},  {\em Invent. Math.} {\bf 154} (2003), no.~1 63--121,
  [\href{http://xxx.lanl.gov/abs/math/0208229}{{\tt math/0208229}}].

\bibitem{BFominZ3}
A.~Berenstein, S.~Fomin, and A.~Zelevinsky, {\it Cluster algebras. {III}.
  {U}pper bounds and double {B}ruhat cells},  {\em Duke Math. J.} {\bf 126}
  (2005), no.~1 1--52, [\href{http://xxx.lanl.gov/abs/math/0305434}{{\tt
  math/0305434}}].

\bibitem{FominZ4}
S.~Fomin and A.~Zelevinsky, {\it Cluster algebras. {IV}. {C}oefficients},  {\em
  Compos. Math.} {\bf 143} (2007), no.~1 112--164,
  [\href{http://xxx.lanl.gov/abs/math/0602259}{{\tt math/0602259}}].

\bibitem{ZelevinskyWCM}
A.~Zelevinsky, ``{Cluster Algebras via Quivers with Potentials}.'' {Research
  lecture at the Worldwide Center of Mathematics}, 2009.
\newblock {\url{http://www.youtube.com/watch?v=NPJmKoO4WJA}}.

\bibitem{Berenstein:2002fi}
D.~Berenstein and M.~R. Douglas, {\it {Seiberg duality for quiver gauge
  theories}},  \href{http://xxx.lanl.gov/abs/hep-th/0207027}{{\tt
  hep-th/0207027}}.

\bibitem{gekhtman2005}
M.~Gekhtman, M.~Shapiro, and A.~Vainshtein, {\it {Cluster algebras and
  Weil-Petersson forms}},  {\em Duke Math. J.} {\bf 127} (2005), no.~2
  291--311, [\href{http://xxx.lanl.gov/abs/math/0309138}{{\tt math/0309138}}].

\bibitem{FockGoncharov}
V.~Fock and A.~Goncharov, {\it Moduli spaces of local systems and higher
  {T}eichm\"uller theory},  {\em Publ. Math. Inst. Hautes \'Etudes Sci.}
  (2006), no.~103 1--211, [\href{http://xxx.lanl.gov/abs/math/0311149}{{\tt
  math/0311149}}].

\bibitem{Fock:2003xxy}
V.~Fock and A.~Goncharov, {\it {Cluster ensembles, quantization and the
  dilogarithm}},  \href{http://xxx.lanl.gov/abs/math/0311245}{{\tt
  math/0311245}}.

\bibitem{Kontsevich:2008}
M.~Kontsevich and Y.~Soibelman, {\it {Stability structures, motivic
  Donaldson-Thomas invariants and cluster transformations}},
  \href{http://xxx.lanl.gov/abs/0811.2435}{{\tt arXiv:0811.2435}}.

\bibitem{Kontsevich:2009}
M.~Kontsevich and Y.~Soibelman, {\it Motivic {D}onaldson-{T}homas invariants:
  summary of results},  in {\em Mirror symmetry and tropical geometry},
  vol.~527 of {\em Contemp. Math.}, pp.~55--89.
\newblock Amer. Math. Soc., Providence, RI, 2010.
\newblock \href{http://xxx.lanl.gov/abs/0910.4315}{{\tt arXiv:0910.4315}}.

\bibitem{FominZ:Y}
S.~Fomin and A.~Zelevinsky, {\it {$Y$}-systems and generalized associahedra},
  {\em Ann. of Math. (2)} {\bf 158} (2003), no.~3 977--1018.

\bibitem{Gaiotto:2010be}
D.~Gaiotto, G.~W. Moore, and A.~Neitzke, {\it {Framed BPS States}},
  \href{http://xxx.lanl.gov/abs/1006.0146}{{\tt arXiv:1006.0146}}.

\bibitem{Alim:2011ae}
M.~Alim, S.~Cecotti, C.~Cordova, S.~Espahbodi, A.~Rastogi, et~al., {\it {BPS
  Quivers and Spectra of Complete N=2 Quantum Field Theories}},  {\em
  Commun.Math.Phys.} {\bf 323} (2013) 1185--1227,
  [\href{http://xxx.lanl.gov/abs/1109.4941}{{\tt arXiv:1109.4941}}].

\bibitem{Alim:2011kw}
M.~Alim, S.~Cecotti, C.~Cordova, S.~Espahbodi, A.~Rastogi, et~al., {\it {N=2
  Quantum Field Theories and Their BPS Quivers}},
  \href{http://xxx.lanl.gov/abs/1112.3984}{{\tt arXiv:1112.3984}}.

\bibitem{Xie:2012gd}
D.~Xie, {\it {BPS spectrum, wall crossing and quantum dilogarithm identity}},
  \href{http://xxx.lanl.gov/abs/1211.7071}{{\tt arXiv:1211.7071}}.

\bibitem{Cecotti:2014zga}
S.~Cecotti and M.~Del~Zotto, {\it {$Y$ systems, $Q$ systems, and 4D
  $\mathcal{N}=2$ supersymmetric QFT}},
  \href{http://xxx.lanl.gov/abs/1403.7613}{{\tt arXiv:1403.7613}}.

\bibitem{Xie:2012mr}
D.~Xie and M.~Yamazaki, {\it {Network and Seiberg Duality}},  {\em JHEP} {\bf
  1209} (2012) 036, [\href{http://xxx.lanl.gov/abs/1207.0811}{{\tt
  arXiv:1207.0811}}].

\bibitem{Heckman:2012jh}
J.~J. Heckman, C.~Vafa, D.~Xie, and M.~Yamazaki, {\it {String Theory Origin of
  Bipartite SCFTs}},  {\em JHEP} {\bf 1305} (2013) 148,
  [\href{http://xxx.lanl.gov/abs/1211.4587}{{\tt arXiv:1211.4587}}].

\bibitem{Franco:2014nca}
S.~Franco, D.~Galloni, and A.~Mariotti, {\it {Bipartite Field Theories, Cluster
  Algebras and the Grassmannian}},
  \href{http://xxx.lanl.gov/abs/1404.3752}{{\tt arXiv:1404.3752}}.

\bibitem{Terashima:2013fg}
Y.~Terashima and M.~Yamazaki, {\it {3d N=2 Theories from Cluster Algebras}},
  {\em PTEP} {\bf 023} (2014) B01,
  [\href{http://xxx.lanl.gov/abs/1301.5902}{{\tt arXiv:1301.5902}}].

\bibitem{Dimofte:2013iv}
T.~Dimofte, M.~Gabella, and A.~B. Goncharov, {\it {K-Decompositions and 3d
  Gauge Theories}},  \href{http://xxx.lanl.gov/abs/1301.0192}{{\tt
  arXiv:1301.0192}}.

\bibitem{Xie:2013lca}
D.~Xie, {\it {Higher laminations, webs and N=2 line operators}},
  \href{http://xxx.lanl.gov/abs/1304.2390}{{\tt arXiv:1304.2390}}.

\bibitem{Cordova:2013bza}
C.~Cordova and A.~Neitzke, {\it {Line Defects, Tropicalization, and
  Multi-Centered Quiver Quantum Mechanics}},
  \href{http://xxx.lanl.gov/abs/1308.6829}{{\tt arXiv:1308.6829}}.

\bibitem{ArkaniHamed:2012nw}
N.~Arkani-Hamed, J.~L. Bourjaily, F.~Cachazo, A.~B. Goncharov, A.~Postnikov,
  et~al., {\it {Scattering Amplitudes and the Positive Grassmannian}},
  \href{http://xxx.lanl.gov/abs/1212.5605}{{\tt arXiv:1212.5605}}.

\bibitem{Golden:2013xva}
J.~Golden, A.~B. Goncharov, M.~Spradlin, C.~Vergu, and A.~Volovich, {\it
  {Motivic Amplitudes and Cluster Coordinates}},  {\em JHEP} {\bf 1401} (2014)
  091, [\href{http://xxx.lanl.gov/abs/1305.1617}{{\tt arXiv:1305.1617}}].

\bibitem{clusterportal}
S.~Fomin, ``{Cluster Algebras Portal}.''
\newblock {\url{http://www.math.lsa.umich.edu/~fomin/cluster.html}}.

\bibitem{Benini:2012ui}
F.~Benini and S.~Cremonesi, {\it {Partition functions of $\mathcal{N}=(2,2)$
  gauge theories on $S^2$ and vortices}},
  \href{http://xxx.lanl.gov/abs/1206.2356}{{\tt arXiv:1206.2356}}.

\bibitem{Hanany:1997vm}
A.~Hanany and K.~Hori, {\it {Branes and N=2 theories in two-dimensions}},  {\em
  Nucl.Phys.} {\bf B513} (1998) 119--174,
  [\href{http://xxx.lanl.gov/abs/hep-th/9707192}{{\tt hep-th/9707192}}].

\bibitem{Hori:2006dk}
K.~Hori and D.~Tong, {\it {Aspects of Non-Abelian Gauge Dynamics in
  Two-Dimensional N=(2,2) Theories}},  {\em JHEP} {\bf 0705} (2007) 079,
  [\href{http://xxx.lanl.gov/abs/hep-th/0609032}{{\tt hep-th/0609032}}].

\bibitem{Hori:2011pd}
K.~Hori, {\it {Duality In Two-Dimensional (2,2) Supersymmetric Non-Abelian
  Gauge Theories}},  {\em JHEP} {\bf 1310} (2013) 121,
  [\href{http://xxx.lanl.gov/abs/1104.2853}{{\tt arXiv:1104.2853}}].

\bibitem{Benini:2011mf}
F.~Benini, C.~Closset, and S.~Cremonesi, {\it {Comments on 3d Seiberg-like
  dualities}},  {\em JHEP} {\bf 1110} (2011) 075,
  [\href{http://xxx.lanl.gov/abs/1108.5373}{{\tt arXiv:1108.5373}}].

\bibitem{Closset:2012eq}
C.~Closset, {\it {Seiberg duality for Chern-Simons quivers and D-brane
  mutations}},  {\em JHEP} {\bf 1203} (2012) 056,
  [\href{http://xxx.lanl.gov/abs/1201.2432}{{\tt arXiv:1201.2432}}].

\bibitem{Xie:2013lya}
D.~Xie, {\it {Three dimensional Seiberg-like duality and tropical cluster
  algebra}},  \href{http://xxx.lanl.gov/abs/1311.0889}{{\tt arXiv:1311.0889}}.

\bibitem{Doroud:2012xw}
N.~Doroud, J.~Gomis, B.~Le~Floch, and S.~Lee, {\it {Exact Results in D=2
  Supersymmetric Gauge Theories}},  {\em JHEP} {\bf 1305} (2013) 093,
  [\href{http://xxx.lanl.gov/abs/1206.2606}{{\tt arXiv:1206.2606}}].

\bibitem{Sugishita:2013jca}
S.~Sugishita and S.~Terashima, {\it {Exact Results in Supersymmetric Field
  Theories on Manifolds with Boundaries}},  {\em JHEP} {\bf 1311} (2013) 021,
  [\href{http://xxx.lanl.gov/abs/1308.1973}{{\tt arXiv:1308.1973}}].

\bibitem{Honda:2013uca}
D.~Honda and T.~Okuda, {\it {Exact results for boundaries and domain walls in
  2d supersymmetric theories}},  \href{http://xxx.lanl.gov/abs/1308.2217}{{\tt
  arXiv:1308.2217}}.

\bibitem{Hori:2013ika}
K.~Hori and M.~Romo, {\it {Exact Results In Two-Dimensional (2,2)
  Supersymmetric Gauge Theories With Boundary}},
  \href{http://xxx.lanl.gov/abs/1308.2438}{{\tt arXiv:1308.2438}}.

\bibitem{Cachazo:2001sg}
F.~Cachazo, B.~Fiol, K.~A. Intriligator, S.~Katz, and C.~Vafa, {\it {A
  Geometric unification of dualities}},  {\em Nucl.Phys.} {\bf B628} (2002)
  3--78, [\href{http://xxx.lanl.gov/abs/hep-th/0110028}{{\tt hep-th/0110028}}].

\bibitem{Feng:2002kk}
B.~Feng, A.~Hanany, Y.~H. He, and A.~Iqbal, {\it {Quiver theories, soliton
  spectra and Picard-Lefschetz transformations}},  {\em JHEP} {\bf 0302} (2003)
  056, [\href{http://xxx.lanl.gov/abs/hep-th/0206152}{{\tt hep-th/0206152}}].

\bibitem{Herzog:2003zc}
C.~P. Herzog, {\it {Exceptional collections and del Pezzo gauge theories}},
  {\em JHEP} {\bf 0404} (2004) 069,
  [\href{http://xxx.lanl.gov/abs/hep-th/0310262}{{\tt hep-th/0310262}}].

\bibitem{Derksen1}
H.~Derksen, J.~Weyman, and A.~Zelevinsky, {\it Quivers with potentials and
  their representations. {I}. {M}utations},  {\em Selecta Math. (N.S.)} {\bf
  14} (2008), no.~1 59--119, [\href{http://xxx.lanl.gov/abs/0704.0649}{{\tt
  arXiv:0704.0649}}].

\bibitem{gekhtman2003cluster}
M.~Gekhtman, M.~Shapiro, and A.~Vainshtein, {\it {Cluster algebras and Poisson
  geometry}},  {\em Mosc. Math. J} {\bf 3} (2003), no.~3 899--934,
  [\href{http://xxx.lanl.gov/abs/math/0208033}{{\tt math/0208033}}].

\bibitem{Witten:1988xj}
E.~Witten, {\it {Topological Sigma Models}},  {\em Commun.Math.Phys.} {\bf 118}
  (1988) 411.

\bibitem{Witten:1991zz}
E.~Witten, {\it {Mirror manifolds and topological field theory}},  in {\em
  {Mirror Symmetry}}, pp.~121--160.
\newblock S.~T.~Yau, 1991.
\newblock \href{http://xxx.lanl.gov/abs/hep-th/9112056}{{\tt hep-th/9112056}}.

\bibitem{Witten:1993yc}
E.~Witten, {\it {Phases of N=2 theories in two-dimensions}},  {\em Nucl.Phys.}
  {\bf B403} (1993) 159--222,
  [\href{http://xxx.lanl.gov/abs/hep-th/9301042}{{\tt hep-th/9301042}}].

\bibitem{Jockers:2012dk}
H.~Jockers, V.~Kumar, J.~M. Lapan, D.~R. Morrison, and M.~Romo, {\it
  {Two-Sphere Partition Functions and Gromov-Witten Invariants}},  {\em
  Commun.Math.Phys.} {\bf 325} (2014) 1139--1170,
  [\href{http://xxx.lanl.gov/abs/1208.6244}{{\tt arXiv:1208.6244}}].

\bibitem{Nekrasov:2009uh}
N.~A. Nekrasov and S.~L. Shatashvili, {\it {Supersymmetric vacua and Bethe
  ansatz}},  {\em Nucl.Phys.Proc.Suppl.} {\bf 192-193} (2009) 91--112,
  [\href{http://xxx.lanl.gov/abs/0901.4744}{{\tt arXiv:0901.4744}}].

\bibitem{Nekrasov:2009ui}
N.~A. Nekrasov and S.~L. Shatashvili, {\it {Quantum integrability and
  supersymmetric vacua}},  {\em Prog.Theor.Phys.Suppl.} {\bf 177} (2009)
  105--119, [\href{http://xxx.lanl.gov/abs/0901.4748}{{\tt arXiv:0901.4748}}].

\bibitem{Gomis}
J.~Gomis and B.~Le~Floch, ``{2d $\mathcal{N}=(2,2)$ Dualities as Toda CFT
  Symmetries and Surface Operators}.'' To appear.

\bibitem{Jockers:2012zr}
H.~Jockers, V.~Kumar, J.~M. Lapan, D.~R. Morrison, and M.~Romo, {\it
  {Nonabelian 2D Gauge Theories for Determinantal Calabi-Yau Varieties}},  {\em
  JHEP} {\bf 1211} (2012) 166, [\href{http://xxx.lanl.gov/abs/1205.3192}{{\tt
  arXiv:1205.3192}}].

\bibitem{Lindstrom:2007vc}
U.~Lindstrom, M.~Rocek, I.~Ryb, R.~von Unge, and M.~Zabzine, {\it {New N =
  (2,2) vector multiplets}},  {\em JHEP} {\bf 0708} (2007) 008,
  [\href{http://xxx.lanl.gov/abs/0705.3201}{{\tt arXiv:0705.3201}}].

\bibitem{Lindstrom:2008hx}
U.~Lindstrom, M.~Rocek, I.~Ryb, R.~von Unge, and M.~Zabzine, {\it {Nonabelian
  Generalized Gauge Multiplets}},  {\em JHEP} {\bf 0902} (2009) 020,
  [\href{http://xxx.lanl.gov/abs/0808.1535}{{\tt arXiv:0808.1535}}].

\bibitem{Witten:1993xi}
E.~Witten, {\it {The Verlinde algebra and the cohomology of the Grassmannian}},
   {\em Cambridge 1993, Geometry, topology, and physics} (1993) 357--422,
  [\href{http://xxx.lanl.gov/abs/hep-th/9312104}{{\tt hep-th/9312104}}].

\bibitem{Donagi:2007hi}
R.~Donagi and E.~Sharpe, {\it {GLSM's for partial flag manifolds}},  {\em
  J.Geom.Phys.} {\bf 58} (2008) 1662--1692,
  [\href{http://xxx.lanl.gov/abs/0704.1761}{{\tt arXiv:0704.1761}}].

\bibitem{Jia:2014ffa}
B.~Jia, E.~Sharpe, and R.~Wu, {\it {Notes on nonabelian (0,2) theories and
  dualities}},  \href{http://xxx.lanl.gov/abs/1401.1511}{{\tt
  arXiv:1401.1511}}.

\bibitem{Lerche:1989uy}
W.~Lerche, C.~Vafa, and N.~P. Warner, {\it {Chiral Rings in N=2 Superconformal
  Theories}},  {\em Nucl.Phys.} {\bf B324} (1989) 427.

\bibitem{Gadde:2013dda}
A.~Gadde and S.~Gukov, {\it {2d Index and Surface operators}},  {\em JHEP} {\bf
  1403} (2014) 080, [\href{http://xxx.lanl.gov/abs/1305.0266}{{\tt
  arXiv:1305.0266}}].

\bibitem{Benini:2013nda}
F.~Benini, R.~Eager, K.~Hori, and Y.~Tachikawa, {\it {Elliptic genera of
  two-dimensional N=2 gauge theories with rank-one gauge groups}},  {\em
  Lett.Math.Phys.} {\bf 104} (2014) 465--493,
  [\href{http://xxx.lanl.gov/abs/1305.0533}{{\tt arXiv:1305.0533}}].

\bibitem{Benini:2013xpa}
F.~Benini, R.~Eager, K.~Hori, and Y.~Tachikawa, {\it {Elliptic genera of 2d N=2
  gauge theories}},  \href{http://xxx.lanl.gov/abs/1308.4896}{{\tt
  arXiv:1308.4896}}.

\bibitem{Gates:1983nr}
S.~Gates, M.~T. Grisaru, M.~Rocek, and W.~Siegel, {\it {Superspace Or One
  Thousand and One Lessons in Supersymmetry}},
  \href{http://xxx.lanl.gov/abs/hep-th/0108200}{{\tt hep-th/0108200}}.

\bibitem{Komargodski:2010rb}
Z.~Komargodski and N.~Seiberg, {\it {Comments on Supercurrent Multiplets,
  Supersymmetric Field Theories and Supergravity}},  {\em JHEP} {\bf 1007}
  (2010) 017, [\href{http://xxx.lanl.gov/abs/1002.2228}{{\tt
  arXiv:1002.2228}}].

\bibitem{Dumitrescu:2011iu}
T.~T. Dumitrescu and N.~Seiberg, {\it {Supercurrents and Brane Currents in
  Diverse Dimensions}},  {\em JHEP} {\bf 1107} (2011) 095,
  [\href{http://xxx.lanl.gov/abs/1106.0031}{{\tt arXiv:1106.0031}}].

\bibitem{Goddard:1976qe}
P.~Goddard, J.~Nuyts, and D.~I. Olive, {\it {Gauge Theories and Magnetic
  Charge}},  {\em Nucl.Phys.} {\bf B125} (1977) 1.

\bibitem{Benini:2012cz}
F.~Benini and N.~Bobev, {\it {Exact two-dimensional superconformal R-symmetry
  and $c$-extremization}},  {\em Phys.Rev.Lett.} {\bf 110} (2013), no.~6
  061601, [\href{http://xxx.lanl.gov/abs/1211.4030}{{\tt arXiv:1211.4030}}].

\bibitem{Benini:2013cda}
F.~Benini and N.~Bobev, {\it {Two-dimensional SCFTs from wrapped branes and
  $c$-extremization}},  {\em JHEP} {\bf 1306} (2013) 005,
  [\href{http://xxx.lanl.gov/abs/1302.4451}{{\tt arXiv:1302.4451}}].

\bibitem{Hori:2013ewa}
K.~Hori, C.~Y. Park, and Y.~Tachikawa, {\it {2d SCFTs from M2-branes}},  {\em
  JHEP} {\bf 1311} (2013) 147, [\href{http://xxx.lanl.gov/abs/1309.3036}{{\tt
  arXiv:1309.3036}}].

\bibitem{Benini:2013yva}
F.~Benini and W.~Peelaers, {\it {Higgs branch localization in three
  dimensions}},  {\em JHEP} {\bf 1405} (2014) 030,
  [\href{http://xxx.lanl.gov/abs/1312.6078}{{\tt arXiv:1312.6078}}].

\bibitem{Dimofte:2010tz}
T.~Dimofte, S.~Gukov, and L.~Hollands, {\it {Vortex Counting and Lagrangian
  3-manifolds}},  {\em Lett.Math.Phys.} {\bf 98} (2011) 225--287,
  [\href{http://xxx.lanl.gov/abs/1006.0977}{{\tt arXiv:1006.0977}}].

\bibitem{Closset:2014pda}
C.~Closset and S.~Cremonesi, {\it {Comments on N=(2,2) Supersymmetry on
  Two-Manifolds}},  \href{http://xxx.lanl.gov/abs/1404.2636}{{\tt
  arXiv:1404.2636}}.

\bibitem{Gomis:2012wy}
J.~Gomis and S.~Lee, {\it {Exact Kahler Potential from Gauge Theory and Mirror
  Symmetry}},  {\em JHEP} {\bf 1304} (2013) 019,
  [\href{http://xxx.lanl.gov/abs/1210.6022}{{\tt arXiv:1210.6022}}].

\bibitem{Gerchkovitz:2014gta}
E.~Gerchkovitz, J.~Gomis, and Z.~Komargodski, {\it {Sphere Partition Functions
  and the Zamolodchikov Metric}},
  \href{http://xxx.lanl.gov/abs/1405.7271}{{\tt arXiv:1405.7271}}.

\bibitem{FominZnotes}
S.~Fomin and A.~Zelevinsky, {\it Cluster algebras: notes for the {CDM}-03
  conference},  in {\em Current developments in mathematics, 2003}, pp.~1--34.
\newblock Int. Press, Somerville, MA, 2003.
\newblock \href{http://xxx.lanl.gov/abs/math/0311493}{{\tt math/0311493}}.

\bibitem{Derksen2}
H.~Derksen, J.~Weyman, and A.~Zelevinsky, {\it Quivers with potentials and
  their representations {II}: applications to cluster algebras},  {\em J. Amer.
  Math. Soc.} {\bf 23} (2010), no.~3 749--790,
  [\href{http://xxx.lanl.gov/abs/0904.0676}{{\tt arXiv:0904.0676}}].

\bibitem{2013arXiv1307.3379D}
B.~Davison, D.~Maulik, J.~Schuermann, and B.~Szendroi, {\it {Purity for graded
  potentials and quantum cluster positivity}},
  \href{http://xxx.lanl.gov/abs/1307.3379}{{\tt arXiv:1307.3379}}.

\bibitem{Labardini}
D.~Labardini-Fragoso, {\it Quivers with potentials associated to triangulated
  surfaces},  {\em Proc. Lond. Math. Soc. (3)} {\bf 98} (2009), no.~3 797--839,
  [\href{http://xxx.lanl.gov/abs/0803.1328}{{\tt arXiv:0803.1328}}].

\bibitem{Geiss}
C.~Gei{\ss}, D.~Labardini-Fragoso, and J.~Schr{\"o}er, {\it {The representation
  type of Jacobian algebras}},  \href{http://xxx.lanl.gov/abs/1308.0478}{{\tt
  arXiv:1308.0478}}.

\bibitem{Gaiotto:2009hg}
D.~Gaiotto, G.~W. Moore, and A.~Neitzke, {\it {Wall-crossing, Hitchin Systems,
  and the WKB Approximation}},  \href{http://xxx.lanl.gov/abs/0907.3987}{{\tt
  arXiv:0907.3987}}.

\bibitem{Gaiotto:2009we}
D.~Gaiotto, {\it {N=2 dualities}},  {\em JHEP} {\bf 1208} (2012) 034,
  [\href{http://xxx.lanl.gov/abs/0904.2715}{{\tt arXiv:0904.2715}}].

\bibitem{Benini:2009gi}
F.~Benini, S.~Benvenuti, and Y.~Tachikawa, {\it {Webs of five-branes and N=2
  superconformal field theories}},  {\em JHEP} {\bf 0909} (2009) 052,
  [\href{http://xxx.lanl.gov/abs/0906.0359}{{\tt arXiv:0906.0359}}].

\bibitem{Gulliksen}
T.~H. Gulliksen and O.~G. Neg\r{a}rd, {\it Un complexe r\'esolvant pour
  certains id\'eaux d\'eterminantiels},  {\em C. R. Acad. Sci. Paris S\'er.
  A-B} {\bf 274} (1972) A16--A18.

\bibitem{Gromov:2007ky}
N.~Gromov and P.~Vieira, {\it {Complete 1-loop test of AdS/CFT}},  {\em JHEP}
  {\bf 0804} (2008) 046, [\href{http://xxx.lanl.gov/abs/0709.3487}{{\tt
  arXiv:0709.3487}}].

\bibitem{Bazhanov:2010ts}
V.~V. Bazhanov, T.~Lukowski, C.~Meneghelli, and M.~Staudacher, {\it {A Shortcut
  to the Q-Operator}},  {\em J.Stat.Mech.} {\bf 1011} (2010) P11002,
  [\href{http://xxx.lanl.gov/abs/1005.3261}{{\tt arXiv:1005.3261}}].

\bibitem{Orlando:2010uu}
D.~Orlando and S.~Reffert, {\it {Relating Gauge Theories via Gauge/Bethe
  Correspondence}},  {\em JHEP} {\bf 1010} (2010) 071,
  [\href{http://xxx.lanl.gov/abs/1005.4445}{{\tt arXiv:1005.4445}}].

\bibitem{Witten:1997yu}
E.~Witten, {\it {On the conformal field theory of the Higgs branch}},  {\em
  JHEP} {\bf 9707} (1997) 003,
  [\href{http://xxx.lanl.gov/abs/hep-th/9707093}{{\tt hep-th/9707093}}].

\bibitem{Benini:2011cma}
F.~Benini, C.~Closset, and S.~Cremonesi, {\it {Quantum moduli space of
  Chern-Simons quivers, wrapped D6-branes and AdS$_4$/CFT$_3$}},  {\em JHEP}
  {\bf 1109} (2011) 005, [\href{http://xxx.lanl.gov/abs/1105.2299}{{\tt
  arXiv:1105.2299}}].

\bibitem{Closset:2012ep}
C.~Closset and S.~Cremonesi, {\it {Toric Fano varieties and Chern-Simons
  quivers}},  {\em JHEP} {\bf 1205} (2012) 060,
  [\href{http://xxx.lanl.gov/abs/1201.2431}{{\tt arXiv:1201.2431}}].

\end{thebibliography}\endgroup
}
\end{document}